\documentclass[hidelinks,11pt]{article} 
\usepackage{amsthm,amssymb,amsmath}
\usepackage{graphicx}
\usepackage{booktabs}
\usepackage{mathtools}
\usepackage{setspace}
\usepackage{tabularx}
\usepackage{color}
\usepackage{epstopdf}
\usepackage[titletoc]{appendix} 
\usepackage{float}
\usepackage{elsevierstyle}

\usepackage{pdflscape}
\usepackage{afterpage}
\usepackage{bm}

\usepackage{enumerate}   
\usepackage[colorlinks=true]{hyperref}  
\hypersetup{
	linkcolor=blue,     
	citecolor=blue,    
	urlcolor=blue   
}

\usepackage{mathtools}

\newtheorem{assumption}{Assumption}[section]

\newtheorem{lemma}{Lemma}[section]
\newtheorem{theorem}[lemma]{Theorem}
\newtheorem{corollary}[lemma]{Corollary}
\newtheorem{proposition}[lemma]{Proposition}

\theoremstyle{definition}  
\newtheorem{remark}{Remark}

\newcommand{\inprod}[1]{\left\langle #1 \right\rangle}
\newcommand\norm[1]{\left\lVert#1\right\rVert}

\definecolor{blendedblue}{rgb}{0.2,0.2,0.7}

\begin{document}

	\title{A Kernelization-Based Approach to Nonparametric Binary Choice Models 
		\thanks{ 
			Based on my job market paper, which was circulated as \emph{Nonparametric Estimation of Large Dimensional Binary Choice Models}.  
			I am grateful to my advisors Joon Y. Park, Yoosoon Chang, and Keli Xu for their advice. 
			For helpful comments and discussions, I thank Sami Stouli, David Harris, Raffaella Giacomini, Giuseppe Cavaliere, Shin Kanaya, 
			Ye Lu, and seminar and conference participants at Essex, Indiana University, Manchester, Peking University, U Melbourne, U Sydney, U Queensland, York University, Melbourne-Monash Workshop 2025, ESWC 2025, SETA 2025, EWMES 2024, ESAM 2024, ANZESG 2024, MEG Meeting 2022. 
			All errors are my own.} 
		\medskip}
	\date{
		January 5, 2026
	}

	\singlespacing
	\author{  
		Guo Yan\footnote{Department of Economics, University of Melbourne, 
			Australia. Email: \href{mailto:yan.g@unimelb.edu.au}{yan.g@unimelb.edu.au} }  
	}
	\maketitle

	\begin{abstract}\normalsize
		\noindent 
		We propose a new estimator for nonparametric binary choice models that does not impose a parametric structure on either the systematic function of covariates or the distribution of the error term. 
		A key advantage of our approach is its computational scalability in the number of covariates. 
		For instance, even when assuming a normal error distribution as in probit models, commonly used sieves for approximating an unknown function of covariates can lead to a large-dimensional optimization problem when the number of covariates is moderate. 
		Our approach, motivated by kernel methods in machine learning, views certain reproducing kernel Hilbert spaces as special sieve spaces, coupled with spectral cut-off regularization for dimension reduction. 
		We establish the consistency of the proposed estimator and asymptotic normality of the plug-in estimator for weighted average partial derivatives. 
		Simulation studies show that, compared to parametric estimation methods, the proposed method effectively improves finite sample performance in cases of misspecification, and has a rather mild efficiency loss if the model is correctly specified. 
		Using administrative data on the grant decisions of US asylum applications to immigration courts, along with nine case-day variables on weather and pollution, we re-examine the effect of outdoor temperature on court judges' ``mood'', and thus, their grant decisions. 
		
	\end{abstract}
	\vfill
	
	\singlespacing

	\newpage
	\renewcommand{\thefootnote}{\arabic{footnote}}
	\setcounter{footnote}{0}
	\onehalfspacing
	\section{Introduction}

	Binary choice problems arise widely in economics. 
	Examples include an individual's choice to work or not, a firm's decision to enter a market, and a household's intention to migrate. 
	Binary choice models (BCMs) are extensively used to analyze these problems due to its underlying microeconomic interpretations of latent utility maximization (e.g., \citealp{bhattacharya-21}). 
	Typically, the latent utility—net utility of one choice over another—comprises two components: a \emph{systematic component}, which is a deterministic function $G$ of covariates, and a \emph{random component} $\varepsilon$ representing idiosyncratic error. 
	
	To avoid possible misspecification and resulting inconsistency for the estimation, \cite{matzkin-92} first studied fully nonparametric BCMs, allowing for nonparametric $G$ and $\varepsilon$. 
	However, the proposed estimator is not practical, as it 
	relies on maximizing empirical likelihood under constraints, e.g., monotonicity of $\varepsilon$'s cumulative distribution function (CDF),  
	whose computation becomes intractable as the number of regressors or sample size increases. 
	Moreover, this approach cannot be used to estimate partial effects, which are important policy parameters.

	One might consider using sieve approximations for nonparametric functions. 
	However, allowing for a nonparametric $G$ can pose computational challenges when dealing with multiple covariates, 
	even if the error distribution is known. 
	This is because estimating nonlinear models requires numerical optimization, and sieve approximations of $G$ can require a large number of basis functions.\footnote{For example, the polynomial expansions of 50 variables up to the 2nd and 3rd orders produce 1,325 and 23,425 basis functions, respectively. 
	} 
	Although there is recent work addressing computational concerns in linear index models (e.g., \citealp{ahn-ichimura-powell-ruud-18}; \citealp{khan-lan-tamer-21}), 
	an additional challenge here lies in handling a nonparametric $G$, which was assumed to be linear in these studies.

	In light of these practical challenges, we propose a new estimation method for a broad class of nonparametric BCMs. 
	We approximate the nonparametric component of covariates using functions in a reproducing kernel Hilbert space (RKHS), which can be viewed as a special sieve space, and couple it with further regularization through spectral cutoff for dimensional reduction. 
	For the nonparametric error component, we follow \cite{gallant-nychka-87} and approximate its density by squared Hermite polynomials, resulting in simple closed-form approximate CDFs that can be easily evaluated without numerical integration.

	We highlight the key computational differences between using classical sieve choices (e.g., polynomials or splines) and using RKHS as special sieves.\footnote{The notations used here are temporary for illustration, with formal results presented in Section \ref{section-estimator-implementation}.} 
	An estimator for a nonparametric function $G$ is often obtained by optimizing over a set of functions with certain basis functions, either by maximizing likelihood or minimizing least squares. 
	Conventional sieve methods typically optimize over the coefficients of basis functions, which can become high-dimensional with multiple covariates. 
	In contrast, when optimizing $G$ over an RKHS $\mathbb G_k$ with reproducing kernel $k(\cdot,\cdot)$, the estimator takes a different form. 
	Here, $\mathbb G_k$ consists of functions spanned by $\{k(x,\cdot): x\in \mathbb R^d\}$, with the inner product induced by $\langle k(s,\cdot), k(t,\cdot) \rangle_{\mathbb{G}_k} = k(s,t)$; see Appendix \ref{appendix-rkhs} for a brief introduction to RKHSs and further references. 
	Specifically, for observed covariates $X_1,\dots,X_n$, the optimization over $\mathbb G_k$ effectively reduces to optimizing over the coefficients of $k(\cdot,X_i)$'s.\footnote{This follows from the representer theorem; see, e.g., Theorem 4.2 in \cite{scholkopf-smola-02}, and the references therein for the history of the development of the representer theorem. } 
	Notably, the number of these coefficients is independent of the covariate dimension.

	To balance the bias-variance trade-off, we use RKHS balls with radii that increase to infinity at certain rates of the sample size. 
	This radius constraint simplifies to a quadratic constraint in optimization. 
	The optimization over the coefficients of $k(\cdot,X_i)$'s can be computationally challenging when the sample size $n$ is large. 
	To address this, we employ spectral cutoff regularization on the $n\times n$ matrix given by $k(X_i,X_j)$ to further reduce the dimensionality for optimization, which is particularly convenient in our setting. 
	We provide an upper bound on the difference between the objective function values at the optima with and without regularization. 
	In our theory, this difference is assumed to vanish asymptotically, allowing the spectral cutoff regularized estimator to be considered as a near-optimal solution to the original problem.

	A key theoretical contribution of this paper is a simple perspective of viewing RKHS balls as special sieve spaces—an idea that can be applied in many nonparametric problems beyond BCMs, and allows for RKHS-based methods to be seamlessly integrated into existing sieve estimation frameworks (e.g., \citealp{chen-07}). 
	It is especially helpful when there are multiple covariates and classical sieves are practically challenging. 
	Furthermore, our approach is more robust to misspecification compared to recent literature on RKHS-based methods in econometrics (e.g., \citealp{singh-22}; \citealp{singh-xu-gretton-24}), which typically assumes that the true nonparametric function belongs to a specific RKHS—an assumption that can be restrictive or misspecified.\footnote{ 
		E.g., assuming that the true function lies in the RKHS with a Gaussian kernel—one of the most widely used kernels in practice—requires the function to be infinitely differentiable, which may be overly restrictive.
	}  
	In our theoretical framework, 
	we impose standard smoothness conditions on the differentiability and boundedness of $G$, as is common in the 
	nonparametric literature, and appropriately choose an RKHS so that $G$ can be approximated by functions within RKHS balls. 
	By explicitly accounting for the approximation error rate that arises when approximating a smooth function using elements from RKHS balls, our approach ensures robustness to misspecification when the true function does not belong to a prespecified RKHS.\footnote{
		Our theoretical analysis builds on 
		approximation results for RKHSs (e.g., \citealp{steinwart-01, micchelli-xu-zhang-06}) and on bounds for approximation error and entropy numbers of Gaussian RKHS balls (e.g., \citealp{smale-zhou-03, kuhn-11}). 
	}

	Our proposed method is not only computationally effective but also theoretically sound. 
	We show the consistency of the proposed \emph{kernelized non-parametric} (KNP) estimator for both the systematic component and the distribution function of the random component. 
	The KNP estimation procedure provides a natural plug-in estimator for the conditional choice probability (CCP) function, for which we establish the convergence rate.

	The KNP approach is useful for estimating important parameters of policy interest, including average partial effects (APEs) and, when accounting for heterogeneity, conditional APEs. 
	Both APEs and conditional APEs are special cases of weighted average derivative functionals of the CCP. 
	We establish the asymptotic normality of the estimators for weighted average derivatives.
	Moreover, these estimators are easy to compute, with a computational procedure that remains unchanged regardless of the covariate dimension.

	The effectiveness of the KNP estimator is demonstrated using extensive simulation studies. We find that, compared to parametric estimation methods, the proposed method effectively improves the finite sample performance in case of misspecification and has rather mild efficiency loss if the model is correctly specified. 
	To demonstrate the practical use of our proposed method, we revisit an empirical application of \cite{heyes-saberian-19,heyes-saberian-22}, examining the effect of outdoor temperature on court judges' decisions.

	\paragraph{Outline} The rest of the paper is organized as follows. Section \ref{section-preliminaries} describes the model. 
	Section \ref{section-estimator-implementation} defines the proposed estimator and describes its implementation. Section \ref{section-theoretical-results} presents the asymptotic properties. 
	Simulation studies are in Section \ref{section-simulations}. 
	In Section \ref{section-application}, the KNP estimation procedure is applied in a model on judges’ decisions and outdoor environments. 
	Section \ref{section-conclusion} concludes the paper. All of the proofs and other technical details are collected in the Appendix. 
	Programs for implementation, along with replication packages for the simulation studies and the empirical application, are available online on the author’s webpage.

	\section{The Model}\label{section-preliminaries}

	We consider the BCM that generates the binary outcome variable $Y \in \{0,1\}$ as follows. 
	\begin{equation}\label{bcm-latent-utility} 
		Y = 1\{ G_0(X) -\varepsilon >0 \}, 
	\end{equation}
	where $G_0$ is an unknown function of covariates $X \in \mathbb{R}^{d_x}$, and $\varepsilon$ is the idiosyncratic error term.  
	The conditional choice probability (CCP) is therefore 
	\begin{align}\label{ccp}
		p_0(x) = \mathbb{P}\{ Y=1| X=x \} =\mathbb{E}\big( Y| X=x \big).
	\end{align}
	Let  $F_0$ and  $f_0$ denote the CDF and Lebesgue density of $\varepsilon$. 
	Let $\mathcal{X} \subset \mathbb{R}^{d_x}$ be the support of $X$.\footnote{We follow the conventional definition that the support of a random vector $Z$ with distribution $P_Z$ is the smallest closed set $A$ which satisfies $P_Z(A)=1$. See, e.g., Page 181 in \cite{billingsley-95} for the existence and uniqueness of the support.}

	For the identification of $G_0$ and $F_0$, we assume that one component of $X$, $V$, which has large support, enters $G_0$ linearly and is separable from the other components, $W$. 
	This assumption is more general than that of many parametric (e.g., probit or logit) or semiparametric BCMs (e.g., \citealp{manski-75,manski-85}), which typically impose a fully linear form on $G_0$.  
	We let $X = (V,W')'$, with $\mathcal V$ and $\mathcal W$ denoting the supports of $V$ and $W$, respectively.
	The assumptions for the identification of $G_0$ and $F_0$ are as follows.

	\begin{assumption}\label{assumption-id-g0-F0}
		We assume that, for $X=(V,W')'$, 
		\[
		G_0(X) = V+ g_0(W), 
		\]
		$\varepsilon$ is independent of $X$, and 
		\begin{enumerate}[(a)]\itemsep-0.1em
			
			\item $F_0\in \mathcal{F}$, where $\mathcal F$ is a set of CDFs which are continuous and strictly increasing on $\mathbb{R}$; 
			\label{condn-id-eps-large-support}
			
			\item $g_0\in \mathcal{G}$, where $\mathcal G$ is a set of continuous functions $g: \mathcal{W} \to \mathbb R$;  
			\label{condn-id-g-cts}
			
			\item There exists a point $w_\ast \in \mathcal{W}$ such that $g(w_\ast)=0$ for all $g\in \mathcal G$; 
			\label{condn-id-gwast-0}
			
			\item 
			Both the conditional distribution $\mathcal{L}(V|W=w_\ast)$ and the marginal distribution $\mathcal L(V)$ have support $\mathbb R$. 
			\label{condn-id-V-large-support}
			
		\end{enumerate} 
	\end{assumption}

	\noindent  
	Assumption \ref{assumption-id-g0-F0}(\ref{condn-id-V-large-support}) requires that 
	$V$ has support $\mathbb R$ both marginally and conditional on $W=w_\ast$.\footnote{Note that $\mathcal{L}(V|W=w_\ast)$ having support $\mathbb R$ does not ensure that $\mathcal L(V)$ has support $\mathbb R$: If for any $w\neq w_\ast$ the support of $\mathcal L(V|W=w)$ is contained in a fixed bounded set, and $\mathbb P\{W=w_\ast\} = 0$, then the support of $\mathcal L(V)$ is bounded. 
	}  
	When $V$ is independent of $W$, Assumption \ref{assumption-id-g0-F0}(\ref{condn-id-V-large-support}) reduces to requiring that the support of $V$ is $\mathbb R$. 
	Moreover, Assumption~\ref{assumption-id-g0-F0}(\ref{condn-id-eps-large-support}) requires that $\varepsilon$ also has support $\mathbb{R}$.

	We denote the true parameter by $\theta_0 = (g_0,F_0)$, and let $\Theta=\mathcal G\times \mathcal F$. 
	Under the independence of $\varepsilon$ from $X$, 
	we define the CCP given by $\theta = (g,F) \in \Theta$ as $p_{\theta}(x) = F(v+g(w))$,  
	where $x = (v,w')'$. 
	We also write $p_{0}  = p_{\theta_0}$ the true CCP, i.e. $p_0(x) =  F_0(v + g_0(w) )$.

	For a criterion function $\ell$ given by either
	\begin{align}\label{ell-ls-g-f}
		\ell(z,\theta) = \big( y - p_{\theta}(x) \big)^2 ,
	\end{align}
	or
	\begin{align}\label{ell-ml-g-f} 
		\ell(z,\theta)  = - y \log p_{\theta}(x) - (1-y) \log \big( 1 - p_{\theta}(x) \big) ,
	\end{align}
	where $z=(y,x')'$, the following theorem gives the identification of $\theta_0$ in the sense that $\theta_0 \in \Theta$ is the unique minimum of the population objective function $Q(\theta) = \mathbb{E} \ell(Z,\theta)$.

	\begin{theorem}\label{theorem-id-g0-F0-unique-min}
		Let Assumption \ref{assumption-id-g0-F0} hold. 
		For $\ell$ given by either \eqref{ell-ls-g-f} or \eqref{ell-ml-g-f}, $\theta \mapsto \mathbb{E} \ell(Z,\theta)$ has a unique minimum at $\theta_0 = (g_0,F_0)$ in $\Theta$. 
	\end{theorem}

	\noindent
	Here, the uniqueness is in the sense that, for any $\theta\in \Theta$ which minimizes $\theta \mapsto \mathbb{E} \ell(Z,\theta)$, it must hold that $g(w) = g_0(w)$ for any $w\in \mathcal W$ and $F(u) = F_0(u)$ for any $u\in \mathbb R$.

	The identification argument follows the general framework in \citet{matzkin-92,matzkin-93,matzkin-94} for nonparametric binary choice models. 
	Our contribution is not a new identification strategy. 
	To the best of our knowledge, however, the specific separable-index structure we consider is not stated in exactly this form in \citet{matzkin-92,matzkin-93,matzkin-94}. 
	For completeness, we state the identification result explicitly and provide a self-contained proof by adapting her proof steps.  
	
	\begin{remark}
		We comment on the conditions imposed for the identification. 
		\begin{itemize}\itemsep-0.1em 
			\item 
			The condition that $g_0$ is continuous on $\mathcal W$ allows the components of $W$ to be discrete, continuous, or a mix of both types of random variable.  
			In particular, when $\mathcal W$ is a finite set, any function with domain $\mathcal W$ is continuous trivially. 
			\item The condition $g(w_\ast)=0$ serves purely for location normalization. 
			Alternatively, one may use the normalization scheme that the error term has a zero mean or median, noting that the class of models $Y = 1\{ V+  (g(W) - c) - (\varepsilon-c)>0\}$, for any constant $c \in \mathbb R$, are observationally equivalent. 
			Similarly, the coefficient on $V$ being one serves for scale normalization, since models $Y = 1\{ sV + sg(W) - s\varepsilon >0\}$, for any constant $s>0$, are also observational equivalent.

			\item  
			The assumption that $\varepsilon$ is independent of $X$ can be relaxed without affecting the identification of $g_0$.  
			However, this relaxation may incur a computational cost in large data environments. 
			For instance, when $\varepsilon$ depends on $X$ only through $V+g_0(W)$, even assuming $g_0$ is linear, the estimators proposed by \cite{ichimura-93} and \cite{klein-spady-93}  for such semiparametric index models become computationally challenging as the sample size increases or when the number of regressors is moderate. This is because these estimators rely on local smoothing procedures, where local smoothers must be recalculated afresh from the data for each observation at every iteration. 
			
		\end{itemize}
	\end{remark}

	A word on notation. 
	For a function $f$ whose domain is a subset of $\mathbb{R}^d$, let
	\[
	D^{(\lambda)} f(x) = \frac{\partial^{\lambda_1}}{\partial x_1^{\lambda_1}} \frac{\partial^{\lambda_2}}{\partial x_2^{\lambda_2}} \cdots \frac{\partial^{\lambda_d}}{\partial x_d^{\lambda_d}}, 
	\]
	where $\lambda = (\lambda_1,\lambda_2,\cdots,\lambda_d)'$ and its elements are nonnegative integers. For such multi-index $\lambda$, let $|\lambda| = \sum_{j=1}^{d} \lambda_j$. Let $D^{(0)} f = f$. 
	For functions whose domain is a subset of $\mathbb{R}$, the $\lambda$-th derivative is denoted as $f^{(\lambda)}$ for any nonnegative integer $\lambda$, and let $f^{(0)} = f$. 
	We write $P_W$ and $P_X$ for the distribution of $W$ and $X$.

	\section{Kernelized Non-Parametric Estimator}\label{section-estimator-implementation}

	Let $\{Y_i,X_i\}_{i=1}^n$ denote $n$ independent observations on the dependent variable $Y$ and covariate vector $X = (V,W')'$, and let $Z_i =(Y_i,X_i')'$.  
	In this section, we first define the proposed estimator in Section \ref{subsection-knp-defn}, followed by the discussions for implementations in Section \ref{subsection-implementation}, and practical implementation procedure in Section \ref{subsection-knp-practical-implementation}.

	\subsection{Kernelized Non-Parametric Estimator}\label{subsection-knp-defn}

	Motivated by Theorem \ref{theorem-id-g0-F0-unique-min}, we propose an estimator, which will be referred to as \emph{kernelized non-parametric} (KNP) estimator, for $\theta_0 = (g_0,F_0)$ as follows. 
	
	The KNP estimator $\hat \theta = (\hat g, \hat F)$ is given by 
	\begin{align}\label{knp-estimator}
		(\hat g, \hat F) \in \arg\min_{g\in \mathcal{G}_n, F\in \mathcal{F}_n}  \bigg\{ \hat Q(\theta)  :=  \frac{1}{n}\sum_{i=1}^n \ell(Z_i,\theta) \bigg\},
	\end{align}
	where we choose the least squares loss\footnote{Here, we choose the least squares loss function in \eqref{ell-ls-g-f}, as its boundedness properties facilitate the proofs. The MLE objective function in \eqref{ell-ml-g-f} could also be used with additional conditions, including assumptions controlling the tails of $\log p_{\theta}(x), \log (1- p_{\theta}(x))$ over $\theta \in \Theta, x\in \mathcal X$. 
		In addition, for the plug-in estimator of the weighted average partial derivative, the MLE objective will lead to the same asymptotic distribution as in Theorem \ref{theorem-asym-normality-wape}, under an analogous set of conditions to Assumption \ref{assumption-wape-asym-norm}.  
	}
	\begin{align*} 
		\ell(z,\theta)  & =  \big( y - p_{\theta}(x) \big)^2 =  \left( y - F(v+g(w))  \right)^2,  
	\end{align*}	
	and the sets $\mathcal G_n, \mathcal F_n$ for optimization are defined in \eqref{set-Gn} and \eqref{set-Fn} below.

	\paragraph{Set $\mathcal G_n$} 
	To describe the set $\mathcal G_n$, we first briefly introduce the RKHS notation. 
	Let a \emph{kernel} $k: \mathcal W \times \mathcal W \to \mathbb{R}$ be a symmetric function which is positive definite, in the sense that 
	$\sum_{i=1}^N \sum_{j=1}^N a_i a_j k(s_i, s_j) \geq 0$ for any $a_i \in \mathbb{R}$, $s_i \in \mathcal W $, $i =1,\cdots,N$, and any positive integer $N$. 
	Let $\mathbb G_k$ be the \emph{reproducing kernel Hilbert space} (RKHS) with reproducing kernel $k$, defined as the completion of the linear span of $\{ k(\cdot,w) |w\in \mathcal W \}$ under the RKHS norm $\|\cdot\|_{\mathbb{G}_k}$, 
	where $\|g\|_{\mathbb G_k} = \sqrt{\langle g,g\rangle_{\mathbb G_k}}$ is induced by the inner product  
	$\left\langle \sum_{i=1}^N a_i k(\cdot,s_i), \sum_{j=1}^M b_j k(\cdot,t_j) \right \rangle_{\mathbb{G}_k} := \sum_{i=1}^N \sum_{j=1}^M a_i b_j k(s_i,t_j)$ for any $s_i,t_j\in \mathcal W, a_i,b_j \in \mathbb{R}$ and any integers $M, N$.  
	In particular, $\|\sum_{i=1}^N a_i k(\cdot,s_i) \|_{\mathbb G_k}^2 = \sum_{i=1}^N \sum_{j=1}^N a_i a_j k(s_i, s_j) $.  
	An example of a commonly used kernel $k$ is the class of Gaussian kernels: For a $\sigma>0$, 
	\[
	k(s,t) = \exp\!\left( - \frac{\|s-t\|^2}{2\sigma^2} \right), \quad s,t \in \mathcal W. 
	\] 
	For any prespecified $\sigma^2>0$, the resulting RKHS $\mathbb{G}_k$ is dense in the space of all continuous functions on $\mathcal{W}$ under the uniform norm.

	With these notations, the set for the optimization of $g$ is chosen so that $g_0:\mathcal W \to \mathbb R$ with $g_0(w_\ast)=0$ is approximated based on functions within the balls of a RKHS $\mathbb G_k$. 
	Specifically, 
	\begin{equation}\label{set-Gn}
		\mathcal{G}_n = \Big\{ g: \mathcal W\to \mathbb R \Big| g(w) = \tilde g(w) - \tilde g(w_\ast) \ \forall w\in \mathcal W, \tilde g\in\mathbb{G}_k, \|\tilde g\|_{\mathbb{G}_k} \leq B_n \Big\}, 
	\end{equation} 
	where $B_n$ is the radius of the RKHS ball. 
	The form $g(w) = \tilde g(w) - \tilde g(w_\ast)$ is used to ensure that $g(w_\ast)=0$, which is location normalization for identification. 
	Intuitively, one can view $\tilde g_0(\cdot)$  (so that $g_0(\cdot)  = \tilde g_0(\cdot) - \tilde g_0(w_\ast)$)  
	as being approximated by a linear combination of kernel basis functions from the class $\{k(w,\cdot)| w\in \mathcal W\}$, and the constraint $\|\tilde g\|_{\mathbb G_k} \leq B_n$ then controls the ``roughness'' 
	of the approximating function and helps prevent over-fitting.

	In practice, we may choose $k$ to be a Gaussian kernel. 
	Moreover, $w_\ast$ is specified as a point in $\mathcal{W}$. 
	In practice, it is convenient to set $w_\ast$ to zero, coupled with standardizing the observations $(W_i)_{i=1}^n$ to have zero mean by subtracting their sample average.

	The radius $B_n$ governs the bias-variance tradeoff when estimating $g_0$ and is required to grow to infinity as $n\to \infty$ to ensure the consistency of the estimator. 
	A larger $B_n$ reduces approximation error, since $g_0$ can be better approximated by functions in $\mathcal G_n$, but also makes $\hat g$ more variable because it is optimized over a larger class  $\mathcal G_n$. 
	A theoretically optimal rate for $B_n$, established later in Corollary \ref{corollary-cgce-rate}, depends on unknown constants, including the number of uniformly bounded derivatives of $g_0$ that exist. 
	In practice, $B_n$ can be chosen via multi-fold cross-validation.

	\paragraph{Set $\mathcal F_n$} 
	Following \cite{fenton-gallant-96, fenton-gallant-96-joe}, the class of functions used to approximate $F_0$ is 
	\begin{equation}\label{set-Fn}
		\begin{split} 
			& \mathcal F_n = \Bigg\{F(\cdot;\tau) = \int_{-\infty}^{\cdot} f(u;\tau) du \bigg|  f(u; \tau) = \bigg( \sum_{j=0}^{J_n} \tau_j u^j \bigg)^2 e^{-u^2/2}, \tau \in \mathcal T_n \Bigg\} \\ 
			&  \text{with } \ \mathcal T_n  =\bigg\{ \tau =  (\tau_0,\tau_1,\dots,\tau_{J_n})' \in \mathbb R^{J_n+1} \bigg| \int f(u; \tau) du = 1 \bigg\}, 
		\end{split}
	\end{equation} 
	where $J_n$ is some positive integer growing to infinity as $n\to \infty$. 
	The idea behind $\mathcal F_n$ is that the density of $F_0$ is approximated by $f(u; \tau) =  \big( e^{-u^2/4}  \sum_{j=0}^{J_n} \tau_j u^j \big)^2$, which is effectively the square of the product of the density of $N(0,2)$ and a polynomial of order $J_n$.

	The polynomial order $J_n$ governs the bias-variance tradeoff when estimating $F_0$ and is required to grow to infinity as $n\to \infty$ to ensure the consistency of the estimator. 
	A larger $J_n$ reduces approximation error, since $F_0$ can be better approximated by CDFs in the larger class $\mathcal F_n$, but the estimator $\hat F$  becomes more variable. 
	A theoretically optimal rate for $J_n$, as established later in Corollary \ref{corollary-cgce-rate}, depends on unknown constants, including the number of derivatives of $F_0$'s density $f_0$ that exist. 
	In practice, $J_n$ can be chosen via multi-fold cross-validation. 
	
	\begin{remark}
		We comment on the choice of kernels and the analytical form of CDFs. 
		\begin{itemize}
			\itemsep -0.1em
			\item (Choice of kernels) Our consistency result later requires that the RKHS $\mathbb{G}_k$ is dense in the space of all continuous functions on compact supports under the uniform norm, and the convergence rate and asymptotic normality results focus on the Gaussian kernels. 
			The literature in statistical learning theory has discussed the conditions under which such denseness conditions are satisfied; 
			see, e.g., \cite{steinwart-01}, \cite{micchelli-xu-zhang-06} among others, where Gaussian kernels serve as one of the examples. See also Remark \ref{remark-largeRKHS} in Appendix~\ref{appendix-rkhs} for more details. 
			In the Monte Carlo experiments and the empirical application, we use a Gaussian kernel with $\sigma^2 = 1$, i.e., $k(s,t) = \exp( - \|s-t\|^2/2 )$.\footnote{
				In this paper, $\sigma^2$ is treated as a prespecified constant; we leave for future research the theory of optimal choice of $\sigma^2$. 
			}
			\item (Analytical form of CDFs in implementation) The CDFs in $\mathcal F_n$ with constraint $\int f(\cdot;\tau)=1$ have simple closed-form expressions parameterized by unconstrained $\tau=(\tau_1,\dots,\tau_{J_n})'$ in $\mathbb R^{J_n}$. 
			See Appendix~\ref{appendix-implementation-cdf-form} for the specific form. 
			This makes the implementation of the proposed estimation procedure straightforward, as no numerical integration is required. 
			Other basis functions, such as wavelets, may also be used to estimate densities on $\mathbb{R}$ and yield closed-form distribution functions (e.g., \citealp{vidakovic-09}). 
			Here, we restrict our attention to Hermite polynomial approximations, since they are not only suitable for our model but also widely used in practice; see, e.g., \cite{merlo-paula-17}, \cite{larsen-21}, \cite{beneito-et-al-21}, and \cite{freyberger-larsen-22}.  
			
		\end{itemize}
		
	\end{remark}

	\begin{remark}[Nonuniqueness and the definition of the estimator]
		\label{remark-KNP-non-unique}
		Since $\hat Q(\theta)$ is not convex in $\theta = (g,F)$, it may have multiple minimizers, even though $Q(\theta)$ has a unique minimizer $\theta_0$ over $\Theta$ by Theorem \ref{theorem-id-g0-F0-unique-min}. 
		In \eqref{knp-estimator}, we define $\hat\theta$ as any minimizer of $\hat Q(\theta)$. 
		Our asymptotic results, including $\|\hat g-g_0\|_\infty\to_p 0$ and $\|\hat F - F_0\|_\infty \to_p 0$ in Theorem \ref{theorem-consistency}, apply to general 
		\emph{approximate solutions} to \eqref{knp-estimator} in the sense of 
		$\hat Q(\hat \theta) \leq \inf_{\theta\in \Theta_n} \hat Q(\theta) + o_p(1)$.\footnote{The requirement that $\hat\theta$ only approximately solves the minimization problem is standard for sieve estimators; see, e.g., \citet[][p. 5561]{chen-07}. Our estimator can be viewed as a special case of a sieve estimator.} 
		So, the theoretical properties of an estimator $\hat \theta$ do not require finite-sample uniqueness. 
		In practice, we compute \emph{one solution} by numerical optimization, as explained in the implementation subsection below. 
	\end{remark}

	\subsection{Explanations for the Implementation and the Estimator}\label{subsection-implementation}
	
	Now we discuss how we optimize over $g\in \mathcal{G}_n$ and over $F \in \mathcal{F}_n$, in order to obtain an estimator 
	$(\hat g, \hat F)$. 
	As in Remark \ref{remark-KNP-non-unique}, it suffices to find one solution that solves \eqref{knp-estimator}, either exactly or approximately.

	\subsubsection{From optimization over RKHS to a finite-dimensional problem}

	Since $\mathcal{G}_n$ is infinite-dimensional, the optimization over $g\in \mathcal{G}_n$ is not directly solvable in practice. 
	The following Proposition \ref{proposition-representer} ensures that 
	we can find a solution to $g$ in \eqref{knp-estimator} in a finite-dimensional subspace, by setting 
	\begin{align}\label{g-aux-1}
		g = \tilde g - \tilde g(w_\ast), \quad \tilde g(w) = \sum_{j=0}^n \delta_j k(W_j,w), \quad W_0:=w_\ast, 
	\end{align}
	and optimizing over $\delta = (\delta_0, \delta_1,\dots,\delta_n)'$, a finite-dimensional Euclidean space. 
	Then the constraint $\|\tilde g\|_{\mathbb{G}_k} \leq B_n$ 
	in the definition of $\mathcal G_n$ can be imposed via $\delta' K\delta \leq B_n^2$, 
	where $\delta :=(\delta_0, \delta_1,\dots,\delta_n)'$ and $K$ is the $(n+1)\times (n+1)$ square matrix whose elements are given by $k(W_i,W_j)$ for $i,j=0,1,\dots,n$.  
	This is because $\big\| \sum_{j=0}^n \delta_j k(W_j, \cdot) \big\|_{\mathbb G_k}^2 = \sum_{i,j=0}^n \delta_i \delta_j k(W_i,W_j)$.  
	
	\begin{proposition}\label{proposition-representer} 
		For any $g\in \mathcal{G}_n$ and $F\in \mathcal{F}_n$, there exists 
		$g_\ast = \tilde g_\ast - \tilde g_\ast(w_\ast) \in \mathcal G_n$ such that  $\hat Q(g_\ast,F) = \hat Q(g,F)$, where $\tilde g_\ast(w) = \sum_{j=0}^n \delta_j k(W_j,w)$ for some real numbers $\delta_j$'s and $\|\tilde g_\ast\|_{\mathbb{G}_k} \leq B_n$. 
	\end{proposition}

	\noindent
	Proposition \ref{proposition-representer} shows that, for any $(g,F)\in \mathcal G_n\times \mathcal F_n$, the criterion value $\hat Q( (g,F))$ can be attained by some $(g_\ast,F)$ where $g_\ast(\cdot) = \sum_{j=0}^n \delta_j \big( k(W_j,\cdot) - k(W_j,w_\ast) \big) \in \mathcal G_n$. 
	Therefore, when minimizing $\hat Q( (g,F))$ over $\mathcal G_n\times \mathcal F_n$, we can restrict, without loss of generality, the search for $g$ to a finite-dimensional class of functions with the form in \eqref{g-aux-1}. 
	This restriction leaves the minimum value of the criterion unchanged and yields a valid solution to \eqref{knp-estimator}.

	\begin{remark} 
		Proposition \ref{proposition-representer} is an analogue of the representer theorem (e.g., Theorem 4.2 in \citealp{scholkopf-smola-02}) for the constrained formulation, where the estimator is defined under the explicit norm constraint $\|\tilde g\|_{\mathbb G_k}\leq B_n$, rather than under the penalized formulation that adds a Ridge-type penalty $c_n\|\tilde g\|_{\mathbb G_k}^2$ to the objective function as in \citet[Theorem~4.1]{scholkopf-smola-02}; the proof follows the same arguments.  
		Unlike Theorem 4.2 in \cite{scholkopf-smola-02},  
		it is not ensured here that every minimizer admits the representer form since the constraint may not be binding; see Remark \ref{remark-non-representer-min} in Appendix~\ref{appendix-proof-representer} for details. 
		Finally, the radius $B_n$ in the norm constraint is convenient for explicitly controlling the bias-variance trade-off and the convergence rate later.  
		Since $\hat Q(\theta)$ is not convex, 
		it is not ensured that there exists a one-to-one mapping between the Ivanov radius $B_n$ and the Ridge penalty parameter $c_n$ that would make the penalized and constrained problems yield the same solutions.  
		We thus work directly with the norm constraint. 
	\end{remark}

	Therefore, a solution $\hat\theta$ in \eqref{knp-estimator} is given by 
	\begin{equation}\label{ghat-fhat-tau-delta}
		\begin{split}
			\hat g(w) & = \sum_{j = 0}^n \hat \delta_j  \big( k(W_j,w) -k(W_j,w_\ast) \big) \\ 
			\hat F(u) & = F(u; \hat \tau)
		\end{split}
	\end{equation}
	where evaluations of $F(\cdot;\tau)$ 
	are computed using the closed-form expression given in Appendix~\ref{appendix-implementation-cdf-form} for any $\tau =(\tau_1,\dots,\tau_{J_n})'\in\mathbb R^{J_n}$, 
	and $(\hat \delta, \hat \tau)$ solves\footnote{Note that $g = \tilde g - \tilde g(w_\ast)$ for some $\tilde g(w) =\sum_{j=0}^n \delta_j  k(W_j,w)$ implies that, for each observation $i=1,\dots,n$, $g(W_i)$ admits the form 
		$g(W_i) = \sum_{j=0}^n \delta_j \big( k(W_j,W_i) - k(W_j,W_0) \big) = [K\delta]_{i+1} - [K\delta]_{1}  $. 
	}   
	\begin{equation}\label{obj-aux-1}
		\begin{split}
			\min_{\delta \in\mathbb R^{n+1}, \tau\in \mathbb R^{J_n} 
			} & \quad   
			\frac{1}{n}\sum_{i=1}^n  \Big( Y_i - F\big(V_i +  [K\delta]_{i+1} - [K\delta]_{1} ; \tau \big)   \Big)^2  
			\\ 
			\textsl{s.t. } & \quad  \delta'K\delta \leq B_n^2, 
		\end{split} 
	\end{equation}
	where 
	$K$ is the $(n+1)\times (n+1)$ Gram matrix whose elements are given by $k(W_i,W_j)$ for $i,j=0,1,\dots,n$, 
	and 
	$[K\delta]_{j}$ denotes the $j$-th element of the vector $K\delta$.

	\begin{remark}[Rank of $K$]
		Regarding the rank of $K$, we provide the following facts, focusing on the kernel functions $k$ whose RKHSs are dense under the uniform norm in the space of all continuous functions on compact domains, including special cases such as Gaussian kernels. 
		$K$ has full rank if observations $W_1,\dots, W_n$ are mutually different, which occurs with probability one when $W$ contains a random variable that has Lebesgue density; See Lemma \ref{lemma-kernel-universal-psd} in Appendix~\ref{appendix-rank-K} for a formal statement and proof. 
		When $\mathcal W$ is finite, which occurs when all components of $W$ are categorical variables, $K$ has a rank no greater than the cardinality of $\mathcal W$. 
	\end{remark}

	\subsubsection{Dimension reduction by spectral cut-off}
	
	The optimization in \eqref{obj-aux-1} is over $\delta \in \mathbb R^{n+1}, \tau\in \mathbb R^{J_n}$. 
	Following the numerical evidence provided by \cite{fenton-gallant-96-joe}, one may set $J_n=n^{1/5}$. 
	We also find that using cross-validation in practice to select $J_n$ often results in a relatively small choice of $J_n$. 
	When the sample size is large, the optimization over $\tau$ is generally not demanding since $J_n$ remains relatively small. 
	However, the optimization over $\delta \in \mathbb R^{n+1}$ can be computationally intensive.  
	To address this, we reduce the dimensionality by using the leading eigenvectors to approximate the $(n+1)\times (n+1)$ matrix $K$.

	More specifically, let $(\hat \lambda_j)_{j=0}^n$ be the eigenvalues of $K$ in descending order and $(\hat u_j)_{j=0}^n$ be the associated orthonormal eigenvectors. Let $\hat U_m$ be the $(n+1)\times m$ matrix collecting the first $m$ columns of  $\hat u_{j}$ for $j=0,1\dots,n$, and let $\hat \Lambda_m$ be the $m\times m$ diagonal matrix collecting the first $m$ eigenvalues of $(\hat u_j)_{j=0}^n$. 
	Then, by the Eckart-Young-Mirsky theorem (see, e.g., Theorem 2.4.8 in \citealp{golub-vanLoan-13}), the best rank-$m$ approximation of the matrix $K$, under both the operator norm and the Frobenius norm, is given by  
	\[
	K \approx \hat U_m \hat \Lambda_m \hat U_m' .
	\] 
	In particular, the approximation of $K$ under the operator norm implies that we can approximate $K\delta$ in the Euclidean norm by 
	\[
	K\delta \approx \hat U_m \hat \Lambda_m \hat U_m' \delta = \hat U_m  \zeta_{\delta}, \quad  \text{where} \quad \zeta_{\delta}= \hat \Lambda_m \hat U_m'\delta \in \mathbb R^{m}. 
	\]

	Motivated by the approximation described above, i.e., using a low-rank approximation of $K$ based on its first $m$ principal components (PC), a reduced problem of \eqref{obj-aux-1} is given by 
	\begin{equation}\label{obj-aux-1-pca}
		\begin{split}
			\min_{\zeta\in \mathbb R^{m}, \tau\in \mathbb R^{J_n} 
			} & \quad  
			\frac{1}{n}\sum_{i=1}^n  \Big( Y_i - F\big(V_i + [\hat U_m\zeta]_{i+1} -  [\hat U_m\zeta]_{1} ; \tau \big)   \Big)^2      
			\\
			\textsl{s.t. } & \quad  \zeta'\hat \Lambda_m^{-1} \zeta \leq B_n^2, 
		\end{split}
	\end{equation}
	where $[\hat U_m \zeta]_j$ denotes the $j$-th element of $\hat U_m\zeta$.  
	Let $(\hat \zeta_{pc}, \hat \tau_{pc})$ be a solution to \eqref{obj-aux-1-pca}, and let $\hat \delta_{pc} = \hat U_m \hat \Lambda_m^{-1} \hat \zeta_{pc}$.\footnote{This is by the approximation $\delta = \sum_{j=0}^n \hat u_j \hat u_j' \delta  \approx \hat U_m \hat U_m' \delta = \hat U_m \hat \Lambda_m^{-1} \hat \Lambda_m \hat U_m' \delta =  \hat U_m \hat \Lambda_m^{-1} \zeta$. 
	} 
	The PC regularized KNP estimator $(\hat g_{pc}, \hat F_{pc})$ is given by $(\hat \delta_{pc}, \hat \tau_{pc})$ via 
	\begin{equation}\label{ghat-fhat-tau-zeta}
		\begin{split}
			\hat g_{pc}(w) & = \sum_{j = 0}^n \hat \delta_{pc,j}  \big( k(W_j,w) -k(W_j,w_\ast) \big) \\ 
			\hat F_{pc}(u) & = F(u;\hat \tau_{pc}) . 
		\end{split}
	\end{equation}
	The optimization in \eqref{obj-aux-1-pca} reduces the dimension of the one in \eqref{obj-aux-1} from $n+J_n$ to $m +J_n$. 
	Here, $m$ is expected to increase with the sample size $n$, both in theory and in practice, but we suppress this dependence for notational simplicity.

		\paragraph{Numerical optimization based on analytical gradients:}
		For implementation, 
		we solve the constrained nonlinear minimization problems in \eqref{obj-aux-1} or \eqref{obj-aux-1-pca}  by	standard numerical optimization routines in existing software.\footnote{
			One can formally write a Lagrangian and consider an unconstrained penalized formulation with a Ridge-type penalty in place of the constraint in \eqref{obj-aux-1} and \eqref{obj-aux-1-pca}. 
			However, because the objective is nonconvex, strong duality is not guaranteed, so the constrained and penalized formulations need not be equivalent. 
		}$^{,}$
		\footnote{The objectives in \eqref{obj-aux-1} and \eqref{obj-aux-1-pca} are not convex in $(\delta,\tau)$ and $(\zeta,\tau)$, and may have multiple minimizers.  
			As in Remark \ref{remark-KNP-non-unique}, $\hat\theta$ is defined as any minimizer of $\hat Q(\theta)$, 
			so any minimizer $(\hat\delta,\hat\tau)$ of \eqref{obj-aux-1} yields an estimator via \eqref{ghat-fhat-tau-delta}. 
			Furthermore, any minimizer $(\hat\zeta,\hat\tau)$ of the dimension reduced objective \eqref{obj-aux-1-pca} yields, through \eqref{ghat-fhat-tau-zeta}, an approximate estimator that shares the same theoretical properties, 
			provided $m$ is chosen appropriately. 
		}
		In this paper, we use the \texttt{fmincon} routine in MATLAB with the default interior-point algorithm to handle the constraint, and we provide the routine with the analytical gradients to facilitate computation.\footnote{
			We use zero as the initial value for the numerical optimization throughout, and the obtained estimates perform well in the Monte Carlo experiments. 
			We also experimented with different starting values, and they gave similar objective values; 
			see Remark~\ref{remark-simulation-diff-initialization} for details. 
		}

	In practice, we choose $m$, along with $J_n, B_n$, based on multi-fold cross-validation. 
	In the theoretical framework of this paper, the number $m$ of eigenvectors used to approximate $K$ needs to satisfy certain conditions to ensure that $\hat\theta_{pc}$ is a near minimum of the optimization problem \eqref{knp-estimator}. 
	See Theorems \ref{theorem-consistency}, \ref{theorem-cgce-rate-aux} and Corollary \ref{corollary-cgce-rate} for the specific conditions, which depend on the eigenvalue decay of the Gram matrix $K$ and the radius $B_n$.

	\begin{remark}[Comparison with classical sieves] 
		A standard sieve approach typically approximates $g_0$ based on $\varphi(w)' \beta$, where $\varphi(w) = (\varphi_1(w),\dots,\varphi_D(w))' \in \mathbb R^D$ collects polynomials or spline basis functions.\footnote{E.g., $\varphi(w) = (1,w,w^2,\dots, w^{D-1})'$ when $w$ is a scalar, or $\varphi(w)$ is a collection of $w$ and higher order and interaction terms if $w$ is multidimensional.}  
		The RKHS-based approach replaces the manually chosen basis $\{\varphi_j(\cdot)\}_{j=1}^D$ with the data-dependent kernel basis $\{k(W_j,\cdot)\}_{j=0}^n$ from the observed $W_j$'s, as in \eqref{ghat-fhat-tau-delta}, 
		and, in the PC-regularized problem \eqref{ghat-fhat-tau-zeta}, with $m$ basis functions $( k(W_0,\cdot), \dots, k(W_n,\cdot) ) \hat U_m \hat\Lambda_m^{-1}$.  
		In this sense, the RKHS-based approach can be viewed as a special sieve method that uses \emph{data-driven} basis functions instead of \emph{pre-specified deterministic} ones, 
		together with norm constraints on the target function.

		The RKHS approach offers a computational advantage over classical sieves 
		when $W$ is multivariate and $d_w$ is such that expanding basis functions and computing derivatives for a potentially large set of sieve basis functions is inconvenient. 
		Series or classical sieve methods rely on taking products of univariate basis functions of different variables to generate sieves of multivariate functions, so the number $D$ of basis functions $\{\varphi_j(w)\}_{j=1}^D$ increases quickly with $d_w$ due to interaction terms, and the optimization for this part must be solved in $\mathbb R^D$ for the coefficients of these basis functions. 
		By contrast, the RKHS-based method works with basis $\{k(W_j,\cdot)\}_{j=0}^n$, 
		which does not depend on the dimension $d_w$ or on the construction of interaction terms.   
		More importantly, when $n$ is large, the spectral cutoff effectively reduces the dimension of optimization further to $\mathbb R^m$, where $m$ is the number of retained principal components of $K$ (or, equivalently, the number of data-driven basis functions in $( k(W_0,\cdot), \dots, k(W_n,\cdot) ) \hat U_m \hat\Lambda_{m}^{-1}$), and $m$ can be chosen much smaller than $n$. 
		Furthermore, computing derivatives of $\hat g(w)$ with respect to $w$ 
		only requires derivatives of the kernel $k(\cdot,\cdot)$, rather than derivatives of a potentially large set of sieve basis functions.

		The RKHS approach makes the estimation and computation feasible when $d_w$ is moderate, but it does not remove the curse of dimensionality. 
		Despite the computational gains, the convergence rate of the proposed estimator in Theorem \ref{theorem-cgce-rate-aux} becomes slower when $d_w$ increases, as in other nonparametric methods. 
	\end{remark}

	\subsubsection{KNP for CCP, APEs, and conditional APEs}\label{subsection-wape-computation}
	The KNP estimation procedure provides a natural plug-in estimator for CCP, in the form $\hat p(x) = \hat F(v+\hat g(w))$. 
	Moreover, the derivatives of $\hat p(x)$, i.e. $\frac{\partial
		\hat p(x)}{\partial x} = \hat f(v+\hat g(w)) \frac{\partial (v+\hat g(w) )}{\partial x}$, can be easily evaluated using the estimated density $\hat f$ indexed by $\hat \tau$ and the derivatives of the kernel function. 
	In particular, from the form of $\hat g$ in \eqref{ghat-fhat-tau-delta} or \eqref{ghat-fhat-tau-zeta}, the derivative $\frac{\partial (v+\hat g(w) )}{\partial x}$ (and hence $\frac{\partial
		\hat p(x)}{\partial x} $) involves only derivatives of the kernel function.
	Accordingly, the computation remains essentially unchanged as the dimension $d_w$ of $W$ increases.\footnote{
		The computation of partial effects here does not require constructing and differentiating a large set of basis functions whose number can grow rapidly with $d_w$, as in classical series/sieve methods. 
	}

	This facilitates the estimation of weighted average partial derivative estimators. 
	In particular, the APE of $X$ and its KNP estimator are given by 
	\[
	APE_{x} = \mathbb E \frac{\partial }{ \partial x} p_0(X), \quad \widehat{APE}_x = \frac{1}{n} \sum_{i=1}^n \frac{\partial }{ \partial x} \hat p(X_i).
	\]
	The APE of $W$ conditioning on $X\in \mathcal S$ for some region $ \mathcal S\subset \mathcal X$ and its estimator are 
	\[
	cAPE_{x| \mathcal S} = \mathbb E \Big( \frac{\partial }{ \partial x} p_0(X) \Big| X\in S\Big), \quad 
	\widehat{cAPE}_{x| \mathcal S} = \frac{ \frac{1}{n} \sum_{i=1}^n \left( 1\{X_i\in \mathcal S\} 
		\frac{\partial }{ \partial x} \hat p(X_i) \right) }{ \frac{1}{n} \sum_{i=1}^n 1\{X_i\in  \mathcal S\} }.
	\]

	\subsection{Practical Implementation}\label{subsection-knp-practical-implementation}
	
	Let $\{(Y_i, V_i,W_i)\}_{i=1}^n$ be the sample. 
	Specify a point $w_\ast \in \mathcal{W}$ for the location normalization $g(w_\ast)=0$. 
	In practice, it is convenient to set $w_\ast =0$ and standardize $(W_i)_{i=1}^n$ to have sample mean zero by subtracting their sample average.  
	Specify a kernel $k$, e.g., $k(s,t)=\exp(-\|s-t\|^2/2)$ used in this paper's Monte Carlo experiments and the empirical application.

	Given $J_n, B_n$ and $m$, the PC-regularized KNP estimator is implemented as follows. 
	\begin{enumerate} [(1)]
		\itemsep -0.1em
		\item Construct the $(n+1)\times (n+1)$ matrix $K$ whose elements are given by $k(W_i,W_j)$ for $i,j=0,1,\dots,n$, where $W_0:=w_\ast$.  
		Compute its eigendecomposition and obtain $\hat \Lambda_m, \hat U_m$, where $\hat \Lambda_m$ is the diagonal matrix collecting the first $m$ eigenvalues in nonincreasing order, and $\hat U_m$ is the matrix whose columns are the corresponding eigenvectors. 
		
		\item 
		Solve for a minimum $(\hat \zeta_{pc},\hat\tau_{pc})$ in \eqref{obj-aux-1-pca} 
		using standard constrained nonlinear numerical optimization routines in existing software, for example \texttt{fmincon} in MATLAB with the default interior-point method to handle the constraint. 
		Supply the routine with the gradients with respect to $\zeta$ and $ \tau$ to facilitate computation. 
		The analytic forms of the objective, constraint, gradients are given in Appendix \ref{appendix-implementation}.

		\item $\hat g_{pc}, \hat F_{pc}$ are obtained via \eqref{ghat-fhat-tau-zeta}. 
		Compute the estimated choice probabilities, APEs, and conditional APEs using the formulas in Section \ref{subsection-wape-computation} as needed.

	\end{enumerate}

	MATLAB code for the implementation of the proposed estimator, together with replication packages for the simulation studies and the empirical application, is available online on the author’s website.

	\section{Theoretical Properties}\label{section-theoretical-results}
	
	This section presents the main theoretical properties of the proposed KNP estimator, both with and without spectral cut-off regularization. 
	In the following subsections, we establish the consistency of the estimator for $\theta_0 = (g_0, F_0)$, the convergence rate of the estimated CCP, and asymptotic normality of weighted average partial derivatives of the CCP.

	For the theory of the PC regularized KNP estimator defined through \eqref{obj-aux-1-pca} and \eqref{ghat-fhat-tau-zeta}, we provide a bound on the difference between the values of $\hat Q$ at $\hat \theta$ and at $\hat \theta_{pc}$, based on which we impose conditions so that $\hat \theta_{pc}$ can be viewed as a near optimum solution when optimizing $\hat Q(\cdot)$ over $\mathcal G_n \times \mathcal F_n$.

	\begin{lemma}\label{lemma-pc-regularization-error} 
		Let $\hat g_{pc}, \hat F_{pc}$ be the PC regularized KNP estimator defined through \eqref{obj-aux-1-pca} and \eqref{ghat-fhat-tau-zeta}. 
		Provided that the densities of all distribution functions in $\mathcal F_n$ satisfy $\|f\|_\infty <M_{\mathcal F}$, it holds that 
		\[
		\hat{Q}( \hat g_{pc}, \hat F_{pc} ) \leq  \inf_{g\in \mathcal G_n, F\in \mathcal F_n} \hat Q(\theta) + 4 M_{\mathcal F} B_n \hat \lambda_{m+1}^{1/2}. 
		\]
	\end{lemma}
	
	\noindent 
	Recall that $\hat \lambda_m$ is the $m$-th eigenvalue of the Gram matrix $K$, and $B_n$ is the radius of the RKHS ball used when approximating $g_0$. 
	
	\begin{remark}[$\hat \theta_{pc}$ as a near optimum]
		Lemma \ref{lemma-pc-regularization-error} shows that $\hat{Q}( \hat g_{pc}, \hat F_{pc} ) \leq  \inf_{g\in \mathcal G_n, F\in \mathcal F_n} \hat Q(\theta) + O_p( B_n \hat\lambda_{m+1}^{1/2})$. 
		For the theoretical properties established later, we assume that $B_n \hat\lambda_{m+1}^{1/2}$ is asymptotically negligible, so that $\hat \theta_{pc}$ can be viewed as a near optimum solution when optimizing $\hat Q(\cdot)$ over $\mathcal G_n \times \mathcal F_n$. 
	\end{remark}

	\subsection{Consistency of the KNP estimator}\label{subsection-consistency}
	
	To establish consistency, we need to impose some assumptions. 
	We define the weighted Sobolev norm 
	\[
	\|f\|_{m_0+m_e,2,\eta_0} := \left( \sum_{0 \leq \lambda \leq m_0+m_e} \int \left| f^{(\lambda)} (u) \right|^2 (1+u^2)^{\eta_0} du \right)^{1/2}, 
	\]
	where $m_0, m_e$ are positive integers, constant $\eta_0>1/2$, and we focus on $\eta \in (1/2,\eta_0)$. 
	
	\begin{assumption}\label{assumption-consistency}
		Assume that
		\begin{enumerate}[(a)]\itemsep-0.1em
			\item $\mathcal W$ is compact.  
			
			\item 
			$\mathcal G$ consists of functions $g:\mathcal W\to \mathbb R$ with $g(w_\ast) = 0$ which have derivatives up to order $m_w$ and all derivatives are uniformly bounded by a constant $M>0$.  
			
			\item $\mathcal{G}_n$ consists of functions in $\mathcal G$ admitting the form $g(\cdot) = \tilde g(\cdot) -  \tilde g(w_\ast)$, where $\tilde g\in\mathbb{G}_k$,  $\|\tilde g\|_{\mathbb{G}_k} \leq B_n $, and $B_n \to \infty$ as $n \to \infty$. 
			Moreover, there exists $g_n\in \mathcal G_n$ such that $\sup_{w\in \mathcal W} |g_n(w) - g_0(w)| \to 0$.  
			
			\item $\mathcal{F}$ consist of distribution functions whose Lebesgue densities have the form 
			$f(u) = \big( f_{sr}(u) \big)^2 $ where $f_{sr}$ is $(m_0+m_e)$-times differentiable with uniformly bounded weighted Sobolev norm, that is  $\|f_{sr}\|_{m_0+m_e,2,\eta_0}<M$  for some positive integers $m_e, m_0$ and constants $\eta_0>1/2$ and $M>0$. 
			\label{aspn-F}
			\item $\mathcal{F}_n$  consists of distribution functions in $\mathcal F$ whose Lebesgue densities have the form 
			\[
			f(u) = \big( f_{sr}(u;\tau) \big)^2, \quad f_{sr}(u;\tau) := \sum_{j=0}^{J_n} \tau_j u^{j}  e^{-u^2/4}
			\] 
			for some $\tau \in \mathbb{R}^{J_n+1}$ and positive integers $J_n$ satisfying that $J_n \to \infty$ as $n \to \infty$.  
		\end{enumerate}
	\end{assumption}

	\begin{remark}
		Assumption \ref{assumption-consistency} is used to ensure the compactness of the parameter sets $\mathcal{G}$ and $\mathcal{F}$ and the denseness of the estimation sets. Before establishing consistency, we comment on some of the conditions imposed. 
		\begin{itemize}\itemsep-0.1em
			
			\item Condition (a) is satisfied automatically when $\mathcal W$ is a set of finitely many points, accommodating naturally for the case where all components $W$ are discrete random variables with finite supports. Similarly, it also accommodates the case where $W$ has both discrete components and continuous components with compact supports. 
			
			\item Condition (b) imposes a smoothness restriction on members of $\mathcal G$, which ensures that $\mathcal G$ is compact under the uniform norm.  
			When $W$ contains discrete components, the condition is considered satisfied as long as functions $g:\mathcal W \to \mathbb R$ can be extended to a function with a domain on an open set and all derivatives of this extended function are uniformly bounded. 
			
			\item Condition (c) ensures that $g_0$ can be approximated by functions in $\mathcal G_n$ arbitrarily well as $n$ becomes large. This condition is satisfied for a variety of kernel functions, including Gaussian kernels for any $\sigma^2$, satisfying the property that the RKHSs $\mathbb G_k$ are dense in the space of all continuous functions. See Remark~\ref{remark-largeRKHS} in Appendix~\ref{appendix-rkhs} for more examples and references therein. 
			
			\item Condition (d) imposes a smoothness restriction on members of $\mathcal F$. 
			The definitions of $\mathcal F_n, \mathcal F$ are taken from \cite{gallant-nychka-87} in a slightly modified form, to better align with the version in \cite{fenton-gallant-96, fenton-gallant-96-joe} and for the convenience of imposing tail conditions later. 
			Conditions (d)-(e) ensures that, under the uniform norm, $\mathcal F$ is compact and $\mathcal F_n$ is dense in $\mathcal F$.  See Lemma \ref{lemma-aux-Fset-prelim-properties} in Appendix~\ref{appendix-tech-lmas-consistency} for more details.  
		\end{itemize}
	\end{remark}

	Now we are ready to establish consistency of the proposed estimator for $\hat \theta = (\hat g,\hat F)$, 
	as well as the estimator $\hat p $ of the true conditional probability function $p_0$ given by 
	\begin{align}\label{p-hat}
		\hat p(x) := p(x, \hat \theta)  = \hat F(v+\hat g(w)). 
	\end{align}
	We denote by $\hat p_{pc}$ when the PC regularized KNP estimator $\hat \theta_{pc}$ is used.

	\begin{theorem}\label{theorem-consistency}
		Let Assumptions \ref{assumption-id-g0-F0}, \ref{assumption-consistency} hold.  
		Then for the KNP estimator given by \eqref{knp-estimator}, it holds that
		\begin{align}\label{consistency-theta}
			d_{\Theta}(\hat \theta, \theta_0) :=  \sup_{w\in \mathcal W} | \hat g(w) - g_0(w) | + \sup_{u\in \mathbb R} |\hat F(u) - F_0(u) |  \to_p 0, 
		\end{align} 
		and 
		\begin{align}\label{consistency-p}
			d_{\Pi}(\hat p, p_0) := \sup_{x \in \mathcal X} \big| \hat p(x) - p_0(x) \big| \to_p 0. 
		\end{align} 
		For the PC regularized KNP estimator, if $m$ is chosen such that $\hat \lambda_{m+1}^{1/2} B_n = o_p(1)$, then 
		\[
		d_{\Theta}(\hat \theta_{pc}, \theta_0) \to_p 0 \quad  \text{and} \quad d_{\Pi}(\hat p_{pc}, p_0)\to_p 0.
		\]
	\end{theorem}

	\subsection{Rate of Convergence}\label{subsection-cgce-rate}
	
	Now we consider the convergence rates for the estimator $\hat p$ of the conditional probability function $p_0(x) = \mathbb{P}\{ Y=1| X=x \}$, under the $L_2(P_{X})$ norm.

	We need to impose some technical conditions. 
	
	\begin{assumption}\label{assumption-cgce-rate}
		Assume that 
		\begin{enumerate}[(a)]\itemsep-0.1em 
			
			\item For $F_0\in \mathcal F$ where $\mathcal F$ given in Assumption \ref{assumption-consistency}(\ref{aspn-F}), its density $f_0(u) = h_0(u)^2 e^{-u^2/2}$ satisfies that, for every $a_0, a_1>0$, there exists $k_0,k_1$ such that 
			\[
			\int_{ u^2 > a_0+a_1 C  }  \big(h_0(u) \big)^2 e^{-u^2/2} du \leq k_0 e^{-k_1 \sqrt{C}}. 
			\] 
			Moreover, $\int_{\mathbb{R}} \big(h^{(j)}_0(u) \big)^2  e^{-u^2/2} du <\infty$ for $j=0,1,\cdots, m_e$. 
			\label{aspn-density-tail}
			
			\item There exists a constant $M_{1,op}>0$ such that $\int_{\mathcal X} h(x)^2 dx \leq M_{1,op}^2 \int_{\mathcal X} h(x)^2 P_X(dx)$  for any function $h$ satisfying $\int_{\mathcal X} h(x)^2 dx<\infty$. 
			\label{aspn-bdd-L2leb-L2Px}
			
			\item Either $\mathcal W$ is finite, or there exists a constant $M_{2,op}>0$ such that $\int_{\mathcal W} h(w)^2 P_W(dw) \leq M_{2,op}^2  \int_{\mathcal W} h(w)^2 dw $ for any function $h$ satisfying $ \int_{\mathcal W} h(w)^2 dw <\infty$.  
			\label{aspn-bdd-L2Pw-L2leb}
			
		\end{enumerate}
	\end{assumption}

	\begin{remark}\label{remark-aspn-cgce-rate}
		Assumption \ref{assumption-cgce-rate}(\ref{aspn-density-tail}), taken from \cite{fenton-gallant-96}, imposes restrictions on the tail behavior of the density $f_0$ of $F_0$. 
		It requires that the tail of the true density $f_0$ not be too heavy, allowing it to be well approximated by the product of a normal density and a squared polynomial. 
		This condition is used to bound the approximation error rate of approximating $F_0$ using functions in $\mathcal F_n$ by the number $J_n$ of basis functions. 
		Conditions~(\ref{aspn-bdd-L2leb-L2Px}) and~(\ref{aspn-bdd-L2Pw-L2leb}) are analogous to the norm equivalence conditions commonly used in sieve literature, e.g., Condition 3.9 in \cite{chen-07}, although we require only one-sided bounds here. 
		Condition~(\ref{aspn-bdd-L2leb-L2Px}) and~(\ref{aspn-bdd-L2Pw-L2leb}) do not exclude cases where $W$ contains categorical random variables. 
	\end{remark}

	Define the space $L_2(X) := L_2(P_X) := \big\{ h:\mathcal X\to \mathbb R \big|  \int h(x)^2 P_X(dx) <\infty \big\}$, where $P_X$ denotes the distribution of $X$, equipped with the norm
	\[
	\|h\|_{L_2(X)} := \bigg( \int  h(x)^2 P_X(dx)   \bigg)^{1/2} .
	\]
	The following theorem provides the convergence rate of $\|\hat p - p_0 \|_{L_2(X)}$. In particular, $\|\hat p - p_0 \|_{L_2(X)} =  \big( \int \big( \hat p(x) - p_0(x) \big)^2 P_X(dx) \big)^{1/2} = \big( \mathbb{E} \big( \hat p(X) - p_0(X) \big)^2 \big)^{1/2}$.

	\begin{theorem}\label{theorem-cgce-rate-aux}
		Let Assumptions \ref{assumption-consistency}, \ref{assumption-cgce-rate} hold.  
		Let $k(s,t) = \exp(-\frac{\|s-t\|^2}{2\sigma^2})$, and $\sigma>0$ be a fixed constant.  
		Let  $\gamma_n = \sqrt{ \frac{(\log B_n)^{d_w+1} \vee J_n}{ n } \log n } $, and assume that $\gamma_n =O(1)$ with $(\log B_n)^{d_w+1} \vee J_n \gtrsim (\log n)^{d_w}$. 
		As $n\to \infty$, 
		\begin{align}\label{ccp-cgce-rate-aux}
			\|\hat p - p_0 \|_{L_2(X)}^2  = O_p\left( \delta_n \right), \quad \delta_n := \max\Big\{ \gamma_n^2, \left(\log B_n\right)^{-m_w/2} + J_n^{-m_e} \Big\}. 
		\end{align} 
		Furthermore, \eqref{ccp-cgce-rate-aux} also holds for $\hat p_{pc}$, provided that $m$ is chosen such that $\hat\lambda_{m+1}^{1/2} B_n = O_p\!\left( \delta_n \right)$. 
	\end{theorem}
	
	As in the sieve literature, the two terms in the rate $\delta_n$ can be interpreted as measures of variance and bias, respectively. 
	Specifically, the first term, $\gamma_n^2$, increases with $B_n$ and $J_n$, reflecting the complexity of the sieve $\Theta_n = \mathcal G_n \times \mathcal F_n$, which arises from the covering numbers of $\mathcal G_n $ and $\mathcal F_n$. 
	This term can be interpreted as a measure of variance.
	The second term $ \left(\log B_n\right)^{-m_w/2} +J_n^{-m_e}$ decreases with $B_n$ and $J_n$, which arises as the square of the deterministic approximation error when using functions in $\Theta_n$ to approximate $\theta_0 = (g_0, F_0)$. 
	Choosing $B_n, J_n$ to balance these two terms in $\delta_n$ yields the following rate of convergence.

	\begin{corollary}\label{corollary-cgce-rate} 
		Let the conditions in Theorem \ref{theorem-cgce-rate-aux} hold. 
		Let $\beta_w := \frac{m_w}{2(d_w+1)}$, and $\beta := \beta_w \wedge m_e$. 
		Let $\log B_n \asymp n^{1/(d_x+m_w/2 ) }$, $n^{\beta_w/(m_e(1+\beta_w) )} \lesssim J_n \lesssim n^{1/(1+\beta_w)}$ when $\beta_w \leq m_e$, and $J_n \asymp n^{1/(1+m_e)}$, $n^{2m_e/(m_w(1+m_e)) } \lesssim \log B_n \lesssim n^{1/(d_x(1+m_e))}$ when $\beta_w > m_e$. Then 
		\begin{align}\label{ccp-cgce-rate}
			\|\hat p - p_0 \|_{L_2(X)}^2 = O_p\!\left(  n^{-\frac{\beta  }{ 1+ \beta }}  \log n \right).
		\end{align}
		Furthermore, \eqref{ccp-cgce-rate} holds for $\hat p_{pc}$, provided that $m$ is chosen such that $\hat\lambda_{m+1}^{1/2} B_n = O_p\!\big( n^{-\frac{\beta  }{ 1+ \beta  }} \log n \big)$. 
	\end{corollary}

	\begin{remark}
		We give some comments on Theorem \ref{theorem-cgce-rate-aux} and Corollary \ref{corollary-cgce-rate}
		\begin{itemize}
			\itemsep -0.1em 
			
			\item The proof of Theorem \ref{theorem-cgce-rate-aux} follows the sieve literature. See, e.g., \cite{chen-07} and the references therein. 
			A key difference here is that the sieve spaces for estimating $g_0$ are RKHS balls with radii growing to infinity, unlike the commonly studied sieve spaces based on polynomials, splines, or wavelets, which are finite-dimensional and linear in parameters with numbers of basis functions growing to infinity. 
			
			\item The view of the RKHS-based approach as a special sieve method appears to be new in the current literature on RKHS-based estimators in econometrics. Typically, the true unknown function to be estimated is assumed to be in the RKHS or some interpolation space between RKHS and a larger space. 
			See, e.g., \cite{singh-sahani-gretton-19}, \cite{singh-22}, \cite{bennett-kallus-mao-newey-syrgkanis-uehara-23}, \cite{singh-xu-gretton-24}.  
			If this assumption holds when using the Gaussian kernel, the true function is implicitly assumed to be infinitely differentiable, and the convergence rate here will reduce to the parametric rate $\sqrt{n}$, provided that $F_0$ is also infinitely differentiable.  
			
			\item The condition in Corollary \ref{corollary-cgce-rate} (e.g., $\log B_n \asymp n^{\frac{1}{d_x+m_w/2}}$) requires $B_n$ to increase at a exponential rate. This is because of the particular use of the infinitely differentiable Gaussian kernel. 
			To have a small approximation error or bias of using functions in RKHS balls with radii $B_n$ to approximate $g_0$, we need $B_n$ to grow fast. 
			On the other hand, the entropy of RKHS balls increases at a logarithm rate of $B_n$, ensuring that the exponential rate of $B_n$ still results in a polynomial rate of the entropy complexity.  
			The choice such as $\log B_n \asymp n^{\frac{1}{d_x+m_w/2}}$ balances the bias and variance.  
			
			\item The rate $\delta_n$ in Theorem \ref{theorem-cgce-rate-aux} consists of two terms that depend on $B_n$ and $J_n$; these two terms can be viewed as variance and squared approximation error, as explained earlier. 
			Note that the number $m$ of retained eigenvectors does not appear in $\delta_n$. This is because we regard $\hat\theta_{pc}$ as a near-optimal solution to the objective function \eqref{knp-estimator}. In practice, the condition  $\hat\lambda_{m+1}^{1/2} B_n = O_p\!\left( \delta_n \right)$ requires selecting $m$ based on the decay rate of the eigenvalues $\hat\lambda_j$'s of the Gram matrix whose elements are $k(W_i,W_j)$. 
			A faster decay of $(\hat\lambda_j)_{j=1}^n$ allows for choosing a smaller value of $m$. 
			In our simulations, we choose $m$ using cross-validation.  
			We leave for future research the theoretical study that considers $m$ as a part of the bias-variance tradeoff,  particularly in the context of using data-driven basis functions to approximate an unknown regression function.  
			
		\end{itemize}
		
	\end{remark}

	\subsection{Asymptotic Normality of Weighted Average Derivatives}\label{subsection-wape}
	
	In this subsection, we establish the asymptotic normality of the weighted average partial derivatives of $\hat p$. The results show that the proposed KNP approach can be used to estimate other functionals of the CCP, including APEs and, when accounting for heterogeneity, conditional APEs, which are often of policy interest.  
	
	We consider the weighted average partial effect of the $j$-th element of $X$. For that, we define the functional $\gamma: \Theta \to \mathbb{R}$ given by
	\begin{align*}
		\gamma(\theta) = \int b(x) \frac{\partial p(x,\theta)}{\partial x_j} dx,
	\end{align*}
	where $b(x)$ is a weighting function defined on $ \mathcal{X}:= \text{Supp}(X)$, i.e. $\int b(x) dx = 1$ and $b(x) \geq 0$. 
	Assume here that $b(\cdot)$ is zero outside some compact set. Then integration by parts gives
	\begin{align}\label{gamma-theta-form-linear-in-ptheta}
		\gamma(\theta) &  = \int - \frac{\partial b(x)}{\partial x_j} p(x,\theta) dx  \notag \\
		& = \mathbb{E} b_{\gamma}(X) p(X,\theta), \quad \text{ where }  b_{\gamma}(x) = - \frac{1}{f_{X}(x)} \frac{\partial b(x)}{\partial x_j},
	\end{align}
	and $f_X(\cdot)$ denotes the density of $X$. 
	Assume that $f_X$ is bounded away from zero on the set where $b(x)$ is positive. 
	Note that $p(x,\theta) = F(v+g(w))$ is smooth in $\theta$ due to the conditions imposed on the parameter space $\Theta = \mathcal G\times \mathcal{F}$; See Lemma \ref{lemma-aux-pathwise-derivatives} in Appendix~\ref{appendix-tech-lmas-for-AN} for its pathwise derivatives. 
	Consequently, $\gamma(\theta)$ is smooth, although nonlinear.

	Let $\gamma_0 = \gamma(\theta_0)$ denote the true weighted average derivative, and $\hat \gamma = \gamma(\hat \theta)$ be given by the KNP estimator $\hat \theta$, with or without PC regularization.  
	The following theorem establishes the limit distribution of the estimator $\hat \gamma := \gamma(\hat\theta)$.

	\begin{theorem}\label{theorem-asym-normality-wape} 
		Let the conditions in Corollary \ref{corollary-cgce-rate} hold with $\beta>1$ so that $\|\hat p - p_0 \|_{L_2(X)} = o_p(n^{-1/4})$ for the KNP estimator with proper choices of $B_n, J_n$, and $m$ if PC regularization is used. 
		Let Assumption \ref{assumption-wape-asym-norm} in Appendix~\ref{appendix-AN} hold. 
		It holds that, as $n\to \infty$, 
		\begin{align}
			\sqrt{n}\big( \hat \gamma- \gamma_0 \big) \to_d \mathbb{N}\Big( 0, \mathbb{E} b_{\gamma}(X)^2 p_0(X)\big( 1-p_0(X) \big) \Big).
		\end{align}
	\end{theorem}
	
	The limit distribution in Theorem \ref{theorem-asym-normality-wape} is the same as in Theorem 3 in \cite{newey-97}, which estimates CCP $p_0(x) = \mathbb{E} (Y | X=x)$ using series regression, and obtains the plug-in estimator for the weighted average derivatives. 
	In our approach, we estimate the latent structure, including both the systematic function and the density of the error term, beyond the reduced form CCP. Theorem \ref{theorem-asym-normality-wape} shows that, in terms of estimating the weighted average derivatives, the KNP procedure provides an estimator with the same asymptotic variance as in \cite{newey-97}. 
	However, since $\theta$ enters the objective function in a highly nonlinear manner compared to the cases when approximating a regression function, we need to impose more assumptions than \cite{newey-97} in Assumption \ref{assumption-wape-asym-norm} to control the high-order terms in the expansions of the objective functions and the functional $\gamma(\theta)$.

	\section{Simulation Studies}\label{section-simulations}
	
	Compared to parametric and semi-parametric estimation methods, the proposed estimator is expected to perform well in large samples, as it is robust to misspecification of both the systematic function of the covariates and the distribution of the error term.  
	For the usefulness of the proposed estimator, below we examine its finite sample performance by a series of simulations, in order to see (i) if there exists serious issues of efficiency loss relative to a correct fully parametric specification, and (ii) if the proposed estimator is effective in situations where the correct parametric specification is unusual.  
	In addition, we examine whether the bootstrap confidence intervals for APEs and cAPEs have reasonable finite-sample coverage.

	We consider the model $Y = 1\{ V + g_0(W) -\varepsilon >0 \}$, where $\varepsilon$ is independent of $X = (V,W')'$. We let $V =_d \mathbb{N}(0,1)$.

	\subsection{Unidimensional $W$}\label{subsection-simulation-1d}
	
	We first focus on unidimensional $W$, and let $W =_d \text{Unif}[-2,2]$. 
	We consider two specifications for $g_0$, where the first one corresponds to the commonly assumed linear index model, and the other is nonlinear.  
	\begin{equation}\label{g0-spec-I-II}
		\begin{split}
			\mbox{I\ :}&\ \ g_0(w) = w \\ 
			\mbox{II\,:}&\ \ g_0(w) = w^2/2 + \sin(\pi w) .
		\end{split}
	\end{equation} 
	Note that at $w_\ast = 0$, $g_0(w_\ast)=0$ under (I) or (II). 
	Specification (II) is chosen so that $g_0$ does not lie in any Gaussian RKHS or any finite-order polynomial RKHS.\footnote{The reproducing kernel Hilbert spaces with Gaussian kernels do not contain any nonzero constant, nor any finite-order polynomials. See, e.g., Theorem 2 in \cite{minh-10}. For the $q$-th order polynomial kernel $k(s,t) = (1+s't)^q$, its RKHS is effectively finitely dimensional, with basis functions consisting of all polynomials up to order $q$. }  
	The error term $\varepsilon$ is given by one of the two following specifications. 
	\begin{equation}\label{eps-spec-A-B}
		\begin{split}
			\mbox{A\,:}&\ \ \varepsilon =_d \mathbb N(0,1) \\
			\mbox{B\,:}&\ \ \varepsilon =_d   
			\frac{1}{4} \mathbb N(-3,1) + \frac{3}{4} \mathbb N(2,1), 
		\end{split}
	\end{equation} 
	where B indicates the mixture of two normal distributions $\mathbb N(-3,1)$ and $\mathbb N(2,1)$, i.e. with probability 1/4, $\varepsilon$ follows $\mathbb N(-3,1)$,  and with probability 3/4,  $\varepsilon$ follows $\mathbb N(2,1)$.

	For simplicity, we refer to the specifications introduced above as (I) and (II) for the systematic function and (A) and (B) for the error distribution. 
	Moreover, we will refer to as (IA), (IB), (IIA), and (IIB) the four cases that are given by the combinations of I and II with A and B.

	We compare the KNP estimator with 
	(a) Kernelized probit (KPB), which specifies the standard normal error term and approximates $g_0$ based on functions in RKHS as in KNP
	(b) Semi-Nonparametric (SNP), which approximates $F_0$ using \cite{gallant-nychka-87}'s method and specifies $g_0$ as a linear function, 
	(c) Probit, which specifies standard normal $\varepsilon$ and linear function of $g_0$. In addition, we consider methods specifying standard normal $\varepsilon$ and approximating $g_0$ based on 2nd, 3rd, 4th polynomials, respectively. 
	For RKHS-based methods, we use the Gaussian kernel $k(s,t) = \exp(-\|s-t\|^2/2 )$. 
	For both KNP and KPB, the number $m$ of eigenvectors in the spectral cut-off regularization are selected based on 5-fold cross-validation.

	Table \ref{table-comparison-1d} compares the performance of these methods for estimating $g_0, p_0$ under each of the four specifications (IIB), (IIA), (IB), (IA), , based on Monte Carlo simulations with $Nsim = 1000$ replications and sample size $ntrain=2000$ observations.  
	The table reports $\text{RMSE}(\hat g) = \sqrt{\mathbb E(\hat g(W)- g_0(W) )^2}, \text{MAD}(\hat g) = \mathbb E |\hat g(W)-g_0(W) |$, along with $\text{RMSE}(\hat p) = \sqrt{\mathbb E(\hat p(X)- p_0(X) )^2}, \text{MAD}(\hat p) = \mathbb E |\hat p(X)-p_0(X) |$, where the expectations are estimated using sample means of $ntest=10,000$ observations in test sample.  
	More details of the simulation procedure are given in the footnote of Table \ref{table-comparison-1d}.

	Table \ref{table-comparison-1d} shows that under specification (IIB), KNP provides the best estimators for $g_0$ and $p_0$, which are much better than all of the other methods. This suggests that the proposed estimator is effective in situations where the correct parametric specification is unusual. 
	Under specification (IIA), KPB performs best as expected, followed closely by KNP, whereas all of the other methods are much worse. 
	Under specification (IB), SNP performs best as expected, and KNP again does the second best. In particular, KNP's estimates for $p_0$ are very close to that of SNP. 
	In specification (IA), where the probit model is correctly specified, the KNP estimator performs comparably to the probit model. 
	This suggests that the efficiency loss of using the proposed method relative to a correct fully parametric specification is rather mild.
	
	While Table \ref{table-comparison-1d} uses sample size $2000$, Tables \ref{table-comparison-1d-ntrain1k}, \ref{table-comparison-1d-ntrain500} report the comparisons using sample sizes $1000, 500$, respectively.  
	The comparisons using the sample size $1000$ are the same as above. When using the sample size $500$ in Table \ref{table-comparison-1d-ntrain500}, one difference is that under (IIB), KNP's estimator for $g_0$ does slightly worse than that of KPB. However, in terms of estimating $p_0$, KNP's estimator is still the best and much better than all of the other methods.

	As supplements to Table \ref{table-comparison-1d}, Fig.~\ref{fig-IA-IIB} presents $Nsim=1000$ simulated estimates of $g_0$ given by each of the methods using sample size $ntrain=2000$ under specification (IA) and (IIB). 
	We note that, compared to other methods, KNP fits the true function best when $g_0$ is nonlinear under (IIB) and also performs well when $g_0$ is linear under (IA).

	\begin{remark}[Sensitivity to starting values]\label{remark-simulation-diff-initialization}
		As a practical check for sensitivity to initialization, 
		we reran the experiment for specification (IIB) with $ntrain=2000$ using four different starting values: zero (as used throughout in the simulations and the empirical application) and three random draws from $N(0,0.1^2)$, over 1000 Monte Carlo replications. 
		The average attained objective values are  (0.1992, 0.1991, 0.1991, 0.1993), which are very similar across starts,  
		and the RMSE and MAD for $\hat g, \hat p$ are also very similar.  
		Overall, this suggests that the numerical solution in this specification is not sensitive to initialization. 
	\end{remark}

	\begin{remark}[Designs with heavier-tailed distributions for $V$]   
		We also repeat the simulation studies above using heavier-tailed distributions for $V$, including logistic and Student's $t$ with degrees of freedom 3. 
		The comparisons across methods remain unchanged. 
		Moreover, the finite-sample performance of $\hat p$ remains essentially unchanged, which is consistent with Theorem~\ref{theorem-cgce-rate-aux}: the convergence rate of $\|\hat p-p_0\|_{L_2(X)}$ does not depend on the tail behavior of $V$. 
	\end{remark}

	\begin{figure}[h]  
		\caption{Simulated estimates using different methods } 
		\label{fig-IA-IIB}
		\begin{center} 
			\includegraphics[width=0.7\textwidth]{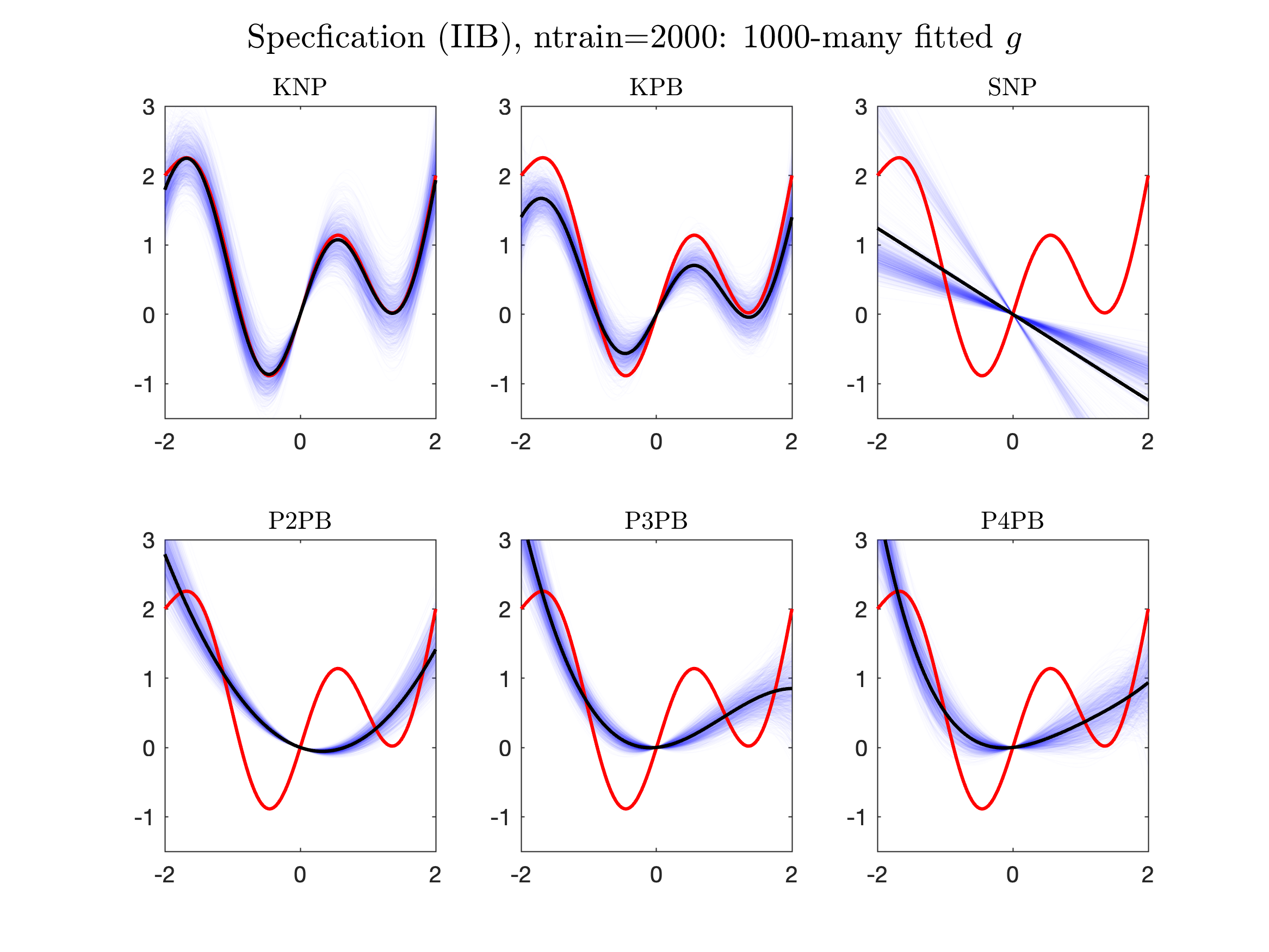}
			\includegraphics[width=0.7\textwidth]{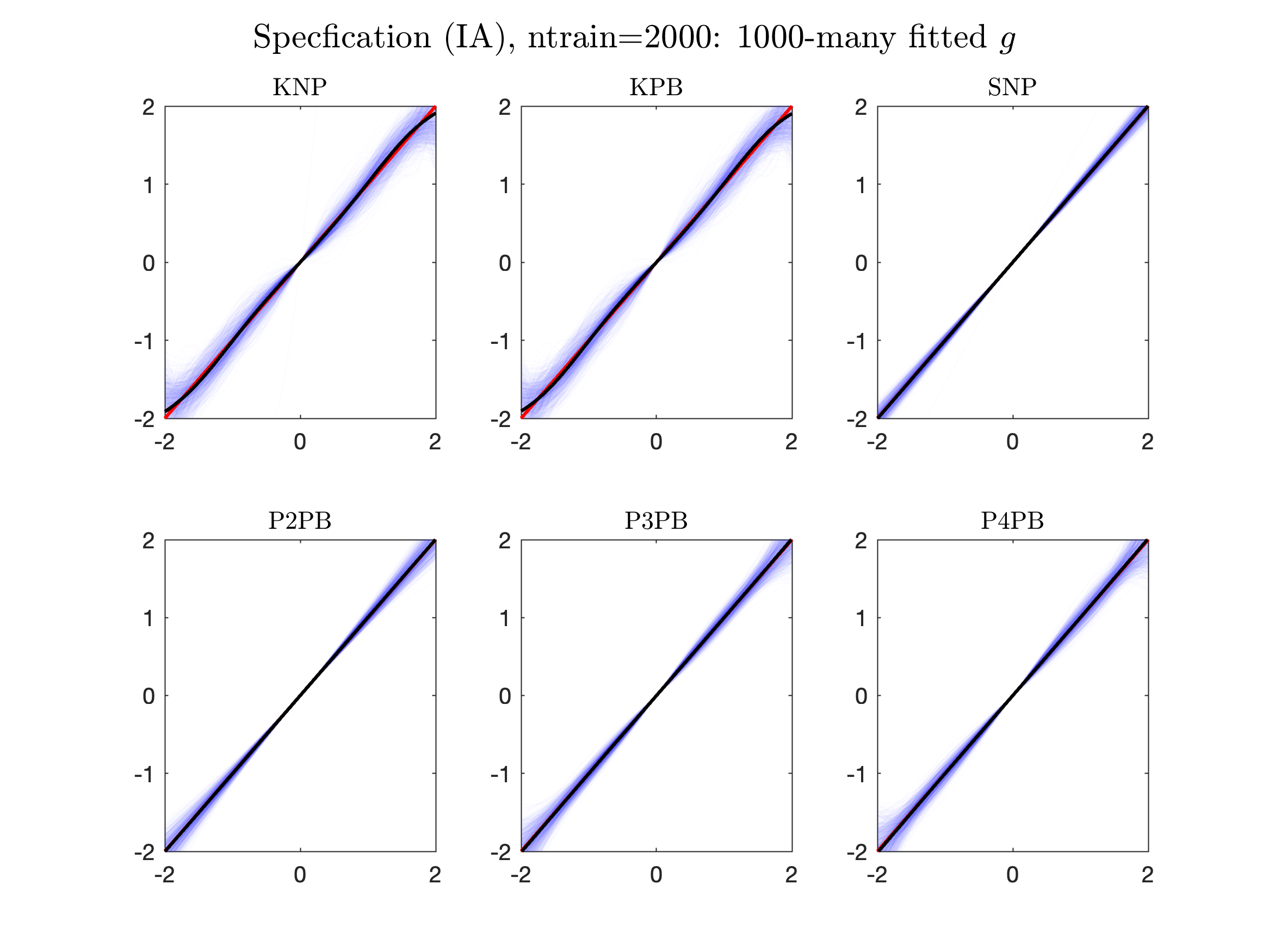}
		\end{center}
		
		\footnotesize
		Notes: 
		This figure presents the estimates of $g_0$ given by different methods compared in Table \ref{table-comparison-1d}.  In this figure, the transparent curves are the estimates $\hat g$ given by each method in $Nsim=1000$ many replications. Refer to the footnote in Table \ref{table-comparison-1d} for more details of the simulation. 
		The black curves are the means of $\hat g$ in $Nsim=1000$ many replications, and the red curve is the true function $g_0$. 
		The DGP is $Y = \{ V + g_0(W) - \varepsilon >0 \}$, where $V =_d\mathbb{N}(0,1)$, $W=_d \text{Unif}[-2,2]$, and $\varepsilon$ is independent of $(V,W)'$. Specification (IA) sets (I) $g_0(w) = w$, and (A) $\varepsilon =_d \mathbb{N}(0,1)$. Specification (IIB) sets (II) $g_0(w) = w^2/2 + \sin(\pi w)$, and (B) $\varepsilon$'s distribution to be the normal mixture $\frac{1}{4} \mathbb N(-3,1) + \frac{3}{4} \mathbb N(2,1)$. 
	\end{figure}

	\afterpage{
		\clearpage  
		\begin{table}[h!]
			\caption{Comparison of methods' performance by simulation: $d_w=1$,  $ntrain=2000$, $Nsim = 1000$  } 
			\label{table-comparison-1d} 
			\small 
			\begin{center}
				\centering
				\begin{tabular}{c c c c c c c c c c c c c c c c c c}
					\hline \hline					
					Method&&	KNP     & KPB	& SNP	 & Probit &	P2PB   & P3PB	& P4PB   \\  [0.5ex] \hline 
					for $F_0$	&&	GN(87) & probit	& GN(87) & probit &	probit & probit	& probit \\ [0.5ex]
					for $g_0$	&&	RKHS   & RKHS   & linear & linear &	Poly2  & Poly3	& Poly4	 \\ \hline 
					&&		&		&		&		&		& & \\   
					&&	 \multicolumn{7}{c}{ Specification (IIB) } 				\\ 
					RMSE($\hat g$)&& 0.235 & 0.390 & 1.259 & 1.120 & 0.680 & 0.637 & 0.650 &  \\  
					MAD($\hat g$)&& 0.193 & 0.332 & 1.097 & 0.997 & 0.573 & 0.547 & 0.554 &  \\   [0.5ex]  
					RMSE($\hat p$)&& 0.037 & 0.139 & 0.138 & 0.141 & 0.109 & 0.105 & 0.106 &  \\  
					MAD($\hat p$)&& 0.027 & 0.118 & 0.103 & 0.111 & 0.089 & 0.085 & 0.085 &  \\   [1ex]  
					&&	 \multicolumn{7}{c}{ Specification (IIA) } 				\\ 
					RMSE($\hat g$)&& 0.154 & 0.149 & 1.109 & 1.113 & 0.663 & 0.672 & 5.195 &  \\  
					MAD($\hat g$)&& 0.112 & 0.109 & 0.952 & 0.944 & 0.580 & 0.556 & 2.055 &  \\   [0.5ex]  
					RMSE($\hat p$)&& 0.029 & 0.027 & 0.224 & 0.281 & 0.180 & 0.152 & 0.088 &  \\  
					MAD($\hat p$)&& 0.020 & 0.018 & 0.178 & 0.226 & 0.138 & 0.116 & 0.054 &  \\   [1ex]  
					&&	 \multicolumn{7}{c}{ Specification (IB) } 				\\ 
					RMSE($\hat g$)&& 0.416 & 0.588 & 0.111 & 0.131 & 0.385 & 0.392 & 0.422 &  \\  
					MAD($\hat g$)&& 0.308 & 0.447 & 0.096 & 0.114 & 0.269 & 0.281 & 0.308 &  \\   [0.5ex]  
					RMSE($\hat p$)&& 0.036 & 0.151 & 0.025 & 0.071 & 0.068 & 0.068 & 0.069 &  \\  
					MAD($\hat p$)&& 0.027 & 0.130 & 0.019 & 0.060 & 0.055 & 0.055 & 0.056 &  \\   [1ex]  
					&&	 \multicolumn{7}{c}{ Specification (IA) } 				\\ 
					RMSE($\hat g$)&& 0.140 & 0.134 & 0.041 & 0.038 & 0.061 & 0.077 & 0.091 &  \\  
					MAD($\hat g$)&& 0.105 & 0.101 & 0.035 & 0.033 & 0.045 & 0.057 & 0.065 &  \\   [0.5ex]  
					RMSE($\hat p$)&& 0.028 & 0.026 & 0.012 & 0.008 & 0.012 & 0.016 & 0.018 &  \\  
					MAD($\hat p$)&& 0.019 & 0.018 & 0.009 & 0.006 & 0.008 & 0.010 & 0.011 &  \\   [1ex]  
					
					\hline \hline
				\end{tabular}
			\end{center}
			
			\footnotesize
			Notes: The DGP is $Y = \{ V + g_0(W) - \varepsilon >0 \}$, where $V =_d\mathbb{N}(0,1)$, $W=_d \text{Unif}[-2,2]$, and $\varepsilon$ is independent of $(V,W)'$. For $g_0$, Specification (I) $g_0(w) = w$, and (II) $g_0(w) = w^2/2 + \sin(\pi w)$. For $\varepsilon$, Specification (A) $\varepsilon =_d \mathbb{N}(0,1)$, and (B) sets $\varepsilon$'s distribution to be the normal mixture $\frac{1}{4} \mathbb N(-3,1) + \frac{3}{4} \mathbb N(2,1)$. 
			The Monte Carlo simulations have $Nsim=1000$ replications, and for each replication, we generate $ntrain = 2000$ for estimation and  $ntest = 10,000$ observations in the test sample for evaluation. 
			RMSE($\hat g) = \sqrt{\mathbb{E}(\hat g(W)-g_0(W))^2} $, MAD($\hat g) = \mathbb E|\hat g(W) - g_0(W) |$, and RMSE($\hat p$), MAD($\hat p$) are defined similarly. Here the expectations in RMSE and MAD of $\hat g, \hat p$ are estimated by sample means using the test sample.

			For the error distribution function $F_0$, method ``GN(87)" indicates using \cite{gallant-nychka-87} method, and ``probit'' indicates specifying $F_0$ to be the cdf of $N(0,1)$. 
			For function $g_0(w) = \tilde g(w)-\tilde g(w_\ast)$, method ``RKHS" indicates approximating $\tilde g$ using functions in Gaussian RKHS with reproducing kernel $k(s,t) = \exp(-\|s-t\|^2/2)$, whereas ``Poly2", ``Poly3", ``Poly4" indicate that $\tilde g$ is approximated using polynomials of order $2, 3, 4$ respectively. 
			Each method is fitted using $ntrain$-many observations, where tuning parameters—such as $J_n$, the order of Hermite polynomials for method ``GN(87)", and $m$, the number of eigenvectors retained when using method ``RKHS" with spectral cut-off regularization—are selected based on 5-fold cross-validation.  
			
		\end{table}
		
		\normalsize
		\clearpage 
	}

	\afterpage{ 
		\clearpage 
		\begin{landscape}
			
			\begin{table}[h!]
				\caption{Comparison of methods' performance by simulation: Designs (IIIA), (IIIB), (IVA), (IVB)}  
				\label{table-comparison-10d}
				
				\small
				\begin{center}
					\centering
					\begin{tabular}{c c c c c c c c c c c c c c c c c c c c c c}
						\hline \hline
						
						Method&&	KNP     & KPB	& SNP	 & Probit &	P2PB      	&& KNP	&	KPB	&	SNP	&	Probit	&	P2PB    	&& KNP	&	KPB	&	SNP	&	Probit	&	P2PB \\  [0.5ex] \hline 
						&&		&		&		&		&		    &&		&		&		&		&		       &&		&		&		&		&		& & \\   
						&&	 \multicolumn{17}{c}{ Specification (IVB) } 				\\ [0.5ex] 
						&& \multicolumn{5}{c}{ ntrain = 2000 } 	    && \multicolumn{5}{c}{ ntrain = 5000 }	    && \multicolumn{5}{c}{ ntrain = 10000 } \\  
						RMSE($\hat g$)&& 0.977 & 1.099 & 1.460 & 1.466 & 1.300 & & 0.701 & 0.952 & 1.493 & 1.442 & 0.843 & & 0.617 & 0.949 & 1.491 & 1.427 & 0.573 & &  \\  
						MAD($\hat g$)&& 0.856 & 0.997 & 1.220 & 1.227 & 1.199 & & 0.633 & 0.881 & 1.251 & 1.202 & 0.779 & & 0.568 & 0.891 & 1.248 & 1.186 & 0.527 & &  \\   [0.5ex]  
						RMSE($\hat p$)&& 0.111 & 0.156 & 0.177 & 0.177 & 0.111 & & 0.066 & 0.133 & 0.173 & 0.175 & 0.086 & & 0.048 & 0.125 & 0.172 & 0.174 & 0.077 & &  \\  
						MAD($\hat p$)&& 0.077 & 0.127 & 0.139 & 0.143 & 0.089 & & 0.046 & 0.107 & 0.136 & 0.141 & 0.070 & & 0.034 & 0.100 & 0.135 & 0.141 & 0.063 & &  \\   [1ex]  
						&&	 \multicolumn{17}{c}{ Specification (IVA) } 				\\ [0.5ex] 
						&& \multicolumn{5}{c}{ ntrain = 2000 } 	    && \multicolumn{5}{c}{ ntrain = 5000 }	    && \multicolumn{5}{c}{ ntrain = 10000 } \\  
						RMSE($\hat g$)&& 0.700 & 0.699 & 1.108 & 1.014 & 0.421 & & 0.478 & 0.476 & 1.080 & 1.006 & 0.250 & & 0.394 & 0.393 & 1.069 & 1.003 & 0.179 & &  \\  
						MAD($\hat g$)&& 0.604 & 0.604 & 0.895 & 0.814 & 0.325 & & 0.417 & 0.416 & 0.870 & 0.808 & 0.195 & & 0.346 & 0.345 & 0.861 & 0.806 & 0.139 & &  \\   [0.5ex]  
						RMSE($\hat p$)&& 0.090 & 0.090 & 0.205 & 0.210 & 0.086 & & 0.058 & 0.058 & 0.203 & 0.209 & 0.054 & & 0.044 & 0.044 & 0.202 & 0.208 & 0.039 & &  \\  
						MAD($\hat p$)&& 0.055 & 0.055 & 0.148 & 0.142 & 0.052 & & 0.035 & 0.035 & 0.147 & 0.141 & 0.032 & & 0.027 & 0.027 & 0.146 & 0.140 & 0.023 & &  \\   [1ex]  
						&&	 \multicolumn{17}{c}{ Specification (IIIB) } 				\\ [0.5ex] 
						&& \multicolumn{5}{c}{ ntrain = 2000 } 	    && \multicolumn{5}{c}{ ntrain = 5000 }	    && \multicolumn{5}{c}{ ntrain = 10000 } \\  
						RMSE($\hat g$)&& 0.630 & 0.904 & 0.305 & 0.333 & 1.261 & & 0.693 & 0.899 & 0.203 & 0.209 & 0.765 & & 0.707 & 0.894 & 0.142 & 0.146 & 0.560 & &  \\  
						MAD($\hat g$)&& 0.587 & 0.872 & 0.268 & 0.294 & 1.168 & & 0.673 & 0.877 & 0.181 & 0.185 & 0.706 & & 0.696 & 0.876 & 0.125 & 0.128 & 0.519 & &  \\   [0.5ex]  
						RMSE($\hat p$)&& 0.054 & 0.123 & 0.041 & 0.067 & 0.105 & & 0.035 & 0.118 & 0.026 & 0.061 & 0.079 & & 0.027 & 0.116 & 0.018 & 0.059 & 0.068 & &  \\  
						MAD($\hat p$)&& 0.040 & 0.101 & 0.032 & 0.054 & 0.084 & & 0.026 & 0.096 & 0.020 & 0.049 & 0.063 & & 0.020 & 0.094 & 0.014 & 0.048 & 0.055 & &  \\   [1ex]  
						&&	 \multicolumn{17}{c}{ Specification (IIIA) } 				\\ [0.5ex] 
						&& \multicolumn{5}{c}{ ntrain = 2000 } 	    && \multicolumn{5}{c}{ ntrain = 5000 }	    && \multicolumn{5}{c}{ ntrain = 10000 } \\  
						RMSE($\hat g$)&& 0.510 & 0.269 & 0.194 & 0.136 & 0.400 & & 0.517 & 0.217 & 0.127 & 0.085 & 0.233 & & 0.521 & 0.204 & 0.089 & 0.060 & 0.162 & &  \\  
						MAD($\hat g$)&& 0.474 & 0.218 & 0.167 & 0.110 & 0.308 & & 0.496 & 0.172 & 0.110 & 0.069 & 0.181 & & 0.508 & 0.162 & 0.077 & 0.049 & 0.126 & &  \\   [0.5ex]  
						RMSE($\hat p$)&& 0.046 & 0.048 & 0.033 & 0.031 & 0.085 & & 0.032 & 0.038 & 0.021 & 0.020 & 0.052 & & 0.025 & 0.034 & 0.014 & 0.014 & 0.036 & &  \\  
						MAD($\hat p$)&& 0.030 & 0.031 & 0.022 & 0.020 & 0.053 & & 0.021 & 0.025 & 0.013 & 0.013 & 0.032 & & 0.016 & 0.022 & 0.009 & 0.009 & 0.023 & &  \\   [1ex]  
						
						\hline \hline
					\end{tabular}
				\end{center}
				
				\footnotesize
				Notes:  
				The DGP is $Y = \{ V + g_0(W) - \varepsilon >0 \}$, where $V =_d\mathbb{N}(0,1)$, components of $W:=(W_1,\dots, W_{10})'$ are iid $\text{Unif}[0,1]$, and $\varepsilon$ is independent of $(V,W)'$. The specifications (III) or (IV) for $g_0$ are given in \eqref{g0-spec-III-IV}, and specifications (A) or (B) for $\varepsilon$ are given in \eqref{eps-spec-A-B}.  
				The Monte Carlo simulations have $Nsim=1000$ replications, and for each replication we generate $ntrain \in \{ 2000,5000,10000 \}$ observations for estimation and $ntest = 1,000,000$ observations for evaluation. 
				See the footnote of Table \ref{table-comparison-1d} for explanations of each method.  
				
			\end{table}
			
		\end{landscape}
		\normalsize
		\clearpage
	}

	\subsection{10-Dimensional $W$}\label{subsection-simulation-10d}
	
	We now consider a case where $W$ is 10-dimensional to evaluate the performance of the proposed estimator in more complex settings. 
	We let $W = (W_{1},\dots,W_{10})'$,  where each $W_{j} =_d \text{Unif}[0,1]$ for $j=1,\cdots,10$ and is independent of each other. 
	We consider two specifications for $g_0$, where the first one corresponds to the commonly assumed linear index model, and the other is nonlinear.  
	\begin{equation}\label{g0-spec-III-IV}
		\begin{aligned}
			\mbox{III\ :} &\ \ g_0(w) = \sum_{j=1}^{10} \beta_{j} w_{j} \\ 
			\mbox{IV\,:}&\ \ g_0(w) = \sum_{j=1}^{10} \beta_j \left( w_j^2 /2 + \sin( \pi w_{j} ) \right), 
		\end{aligned}
	\end{equation} 
	where $\beta = (0.63,	0.81,	-0.75,	0.83,	0.26,	-0.80,	-0.44,	0.09,	0.92,	0.93)'$
	\footnote{These numbers are generated as the first 10 numbers from $\text{Unif}[-1,1]$ using rng(`default') in Matlab. } 
	The error term $\varepsilon$ is independent of $X$ and is given by specifications (A) and (B) as before.  
	We will refer to as (IIIA), (IIIB), (IVA), and (IVB) the four cases that are given by the combinations of III and IV with A and B. 
	
	Similar as Table \ref{table-comparison-1d}, Table \ref{table-comparison-10d} presents the simulation results for designs (IVB), (IVA), (IIIB), (IIIA) using $Nsim = 1000$ replications, sample size $ntrain \in\{2000,5000,10000\}$ for estimation and $ntest =1,000,000$ for out-of-sample evaluation. Table \ref{table-comparison-10d} further demonstrates that the KNP estimator is robust to both misspecification of the systematic function of the covariates and misspecification of the density of error term. Moreover, for moderate sample sizes, the KNP estimator shows desirable properties: It effectively improves the finite sample performance in case of misspecification, and has a rather mild efficiency loss if the model is correctly specified.

	One unexpected pattern in Table \ref{table-comparison-10d} is that, in designs (IIIA) and (IIIB), $\text{RMSE}(\hat g)$ and $\text{MAD}(\hat g)$ are not monotone in $ntrain$, even though $\text{RMSE}(\hat p)$ and $\text{MAD}(\hat p)$ decrease. 
	This pattern arises from the data-driven tuning, where $J_n,m, B_n$ are selected by 5-fold cross-validation to minimize the squared loss in the objective \eqref{obj-aux-1-pca} for $p(\cdot)$, which is a composition of $g$ and $F$ and thus targets the error of $\hat p$ rather than the error of $\hat g$.  
	To check that this pattern is indeed driven by the data-driven model selection, we conduct an additional experiment in which $J_n$ is fixed at $J_n=1$ for design (IIIA) and $J_n = 3$ for design (IIIB), while $m$ is still chosen by cross-validation as before and $B_n=100$ is fixed. 
	Table~\ref{table-comparison-10d-appendix-J1or3} in Appendix~\ref{appendix-table-figures} reports the corresponding results. 
	In this setting, $\text{RMSE}(\hat g)$ and $\text{MAD}(\hat g)$ decrease with the training sample size, and $\text{RMSE}(\hat p)$ and $\text{MAD}(\hat p)$ also decrease, which is consistent with the theoretical properties of the KNP estimator.

	\subsection{Coverage of Bootstrap Confidence Intervals for APE and cAPEs}\label{subsection-simulation-ape}
	
	Theorem \ref{theorem-asym-normality-wape} establishes the asymptotic normality of weighted average partial derivatives. 
	In the empirical application in Section \ref{section-application}, we construct nonparametric bootstrap confidence intervals for inference, since the asymptotic variance in Theorem \ref{theorem-asym-normality-wape} depends on the density of $X$ and its derivatives, which are challenging to estimate given that $X$ is 10-dimensional in that example.\footnote{While it may be possible to adopt methods such as kernel density estimation or the ones proposed in \cite{spady-stouli-20} to estimate the asymptotic variance, we opted for bootstrap intervals for practicality. 
		Moreover, Theorem \ref{theorem-asym-normality-wape} shows a standard root-$n$ asymptotically normal distribution centered at zero, so arguments as in \cite{chen-linton-vankeilegom-03} can, in principle, be adapted to show the bootstrap validity under suitable conditions. 
		We leave a formal proof of bootstrap validity for future research. 
	}  
	The empirical application focuses on the APE of one component of $W$ and on two conditional APEs (cAPEs) conditioning on this component being below or above its average.

	\afterpage{ 
		\begin{table}[h!]
			\caption{Coverage of bootstrap confidence intervals: $Nsim=1000$, $Nboot = 1000$
			} 
			\label{table-boot-ci-coverage} 
			\small 
			\begin{center}
				\centering
				\begin{tabular}{c c c c c c c c c c c c c c c c c c}
					\hline \hline					
					
					Confidence level	&	&	90\% 	&	&	95\%	&	&	99\%	\\  [0.5ex] \hline  
					&	&	 \multicolumn{5}{c}{$ntrain= 2000$} 							\\ [0.5ex]
					$APE_{W_1}$	&	&	0.951	&	&	0.977	&	&	0.994	\\
					$cAPE_{W_1|W_1<0.5}$	&	&	0.932	&	&	0.965	&	&	0.992	\\
					$cAPE_{W_1|W_1>0.5}$	&	&	0.925	&	&	0.966	&	&	0.996	\\  [1ex]
					
					&	&	 \multicolumn{5}{c}{$ntrain= 5000$} 							\\ [0.5ex]
					$APE_{W_1}$	&	&	0.923	&	&	0.969	&	&	0.994	\\
					$cAPE_{W_1|W_1<0.5}$	&	&	0.909	&	&	0.958	&	&	0.991	\\
					$cAPE_{W_1|W_1>0.5}$	&	&	0.930	&	&	0.965	&	&	0.994	\\  [1ex]
					
					&	&	 \multicolumn{5}{c}{$ntrain= 10000$} 							\\ [0.5ex]
					$APE_{W_1}$	&	&	0.906	&	&	0.948	&	&	0.987	\\
					$cAPE_{W_1|W_1<0.5}$	&	&	0.926	&	&	0.967	&	&	0.994	\\
					$cAPE_{W_1|W_1>0.5}$	&	&	0.893	&	&	0.936	&	&	0.988	\\  [1ex]

					\hline \hline
				\end{tabular}
			\end{center}
			
			\footnotesize 
			Notes: The DGP is specification (IVB) described in Section~\ref{subsection-simulation-10d}. 
			For each sample size $ntrain \in\{2000,5000,10000\}$, the empirical coverage is computed from $Nsim=1000$ replications. 
			In each replication,  
			we estimate the APE of $W_1$ and two conditional APEs of $W_1$ conditioning on $W_1<0.5$ and $W_1>0.5$, denoted by $APE_{W_1}$, $cAPE_{W_1|W_1<0.5}$, and $cAPE_{W_1|W_1>0.5}$. 
			Within each replication, we generate $Nboot =1000$ bootstrap samples and re-estimate the three parameters for each bootstrap sample, from which we construct the 90\%, 95\%, 99\% bootstrap confidence intervals. 
			To reduce computation time, we fix the choice of tuning parameters during estimation. The choice is based on cross-validation and corresponds to the values that minimize the squared loss objective: (i) $J_n = 2, m=68, B_n=100$ for $ntrain=2000$, (ii) $J_n = 4, m=68, B_n=100$ for $ntrain=5000$, (iii) $J_n = 4, m=74, B_n=100$ for $ntrain=10000$.  
		\end{table}
		
		\normalsize 
	}

	Motivated by this empirical application, we examine the finite-sample coverage of bootstrap confidence intervals (CIs) in the 10-dimensional $W$ design (IVB). 
	We focus on the APE of $W_1$ and two cAPEs conditioning on $W_1<\mathbb E W_1 =0.5$ and $W_1>0.5$, denoted by $APE_{W_1}, cAPE_{W_1|W_1<0.5}, cAPE_{W_1|W_1>0.5}$, respectively. 
	We generate $Nsim=1000$ Monte Carlo replications, with sample size $ntrain\in\{2000, 5000, 10000\}$. 
	For each replication, we generate a sample of size $ntrain$, 
	fit the KNP estimator, compute the estimates for the APE and cAPEs, and draw $Nboot = 1000$ bootstrap samples and re-estimate APEs to construct 90\%, 95\%, and 99\% bootstrap confidence intervals. 
	Bootstrap confidence intervals are constructed using the basic (reverse-percentile) bootstrap method.

	Table \ref{table-boot-ci-coverage} reports the empirical coverage probabilities.   
	The results indicate that the bootstrap CIs have coverage generally close to the nominal levels for both the APE and the two cAPEs: although the CIs are slightly conservative at $ntrain=2000$, coverage tends to be closer to the nominal as the sample size increases. 
	Table~\ref{table-boot-ci-length} in Appendix~\ref{appendix-table-figures} also reports the average CIs length, which decreases with sample size.

	\section{Application: Temperature and Judge’s Decision}\label{section-application}

	\cite{heyes-saberian-19, heyes-saberian-22} analyze the effect of outdoor temperature on the probability of an asylum application being granted. 
	Based on a linear probability model (LPM) including other weather and pollution characteristics, \cite{heyes-saberian-19} finds that, in their preferred specification, a $10^{\circ}$F increase in case-day temperature reduces the grant probability by 1.075 percent. 
	Their results suggest that high temperatures may damage decision consistency, even for experienced professional decision-makers who work indoors and are ``protected" by climate control. 
	The evidence that such socially and economically important high-stakes decisions can be affected by extraneous variables suggests a potential welfare loss. 
	
	How do temperature and other weather and pollution characteristics affect a judge’s grant decision? \cite{heyes-saberian-19} highlight that their findings are consistent with established links from temperature to mental function, decision-making, risk attitudes, and mood. 
	Following their argument, we may naturally expect these environmental variables to affect judges’ perceived utility from the outdoor environment, thereby influencing their grant decisions.

	Two further questions arise. 
	First, is the effect of temperature on judges' perceived utility approximately linear, or does it vary nonlinearly with temperature? 
	Second, do environmental variables (e.g., temperature, air pressure, dew point, PM2.5) affect utility in an additive and separable way, or do they interact with each other when affecting decision-makers' mental states? 
	A fully nonparametric utility function of environmental variables is useful for addressing both questions.

	In this section, we apply KNP to investigate the effect of temperature on judges' grant decisions, allowing for these important features. 
	This exercise illustrates the empirical relevance of our proposed method. 
	We find a bell-shaped relationship between temperature and judges' perceived utility, 
	which is consistent with existing evidence that outdoor environments may affect decisions through decision-makers' mental states or moods. 
	Moreover, the effects of temperature on judges' decisions are heterogeneous across temperature ranges: 
	while the conditional APE at high temperatures is significantly negative, the conditional APE at low temperatures is not statistically significant at the 10\% significance level. 
	This heterogeneity can potentially inform more targeted policy interventions across temperature ranges.

	\subsection{Empirical Specification and Implementation} 
	
	Let $Y_i= 1$ if the application case $i$ is granted, and 0 otherwise. 
	We apply the proposed KNP estimation procedure to the following model: 
	\begin{equation}\label{bcm-application}
		Y_i  = 1\{ Y_{i}^\ast >0 \}  \quad \text{with} \quad Y_{i}^\ast = V_{i} + g_0(W_{i}) - \varepsilon_{i}, 
	\end{equation} 
	where $Y_i^\ast$ is a latent index that can be viewed as an unobserved score of case $i$.  
	For each case $i$, let $j$ denote its assigned judge, $t$ the decision time (year-month-day), and $c$ the applicant's country of nationality.  
	$W_i$ is the vector of \emph{9 outdoor environmental variables} that case $i$'s assigned judge $j$ was exposed to at the decision time $t$, including mean daily temperature, air pressure, dew point, precipitation, wind speed, sky cover, ozone, CO, PM2.5; 
	see Section II of \cite{heyes-saberian-19} for the definitions and description of 
	$W_i$ and $Y_i$, and Table 1 therein for summary statistics.  
	Note that $W_i$ depends only on the judge and the decision time; hence, we also write $W_i = W_{jt}$. 
	Here $g(W_i)$, or $g(W_{jt})$ more specifically, may be interpreted as the utility given by the outdoor environment with variables $W_{jt}$.

	Variable $V_i$ is chosen to be the log-odds of the mean approval rate for different types of applications from country $c$ over each month.\footnote{
		As explained in \cite{heyes-saberian-19}, ``There are two types of cases in immigration courts: affirmative cases in which the applicant presents in the
		courts on her/his own and defensive cases in which the applicant is instructed to attend on the initiative of the immigration authorities."} 
	Since the observed case characteristics in the data are limited, this choice allows us to account for some of the heterogeneity in the applicant's nationality, types of application, and time at the level of year-month. 
	Fig.~\ref{fig-v-kde-hist} in Appendix~\ref{appendix-table-figures} plots the histogram and a kernel density estimate of $V_i$ based on 99{,}773 observations, using Silverman’s rule-of-thumb bandwidth. 
	The empirical range of $V_i$ is $[-3.761, 2.772]$; the $0.1\%, 0.5\%, 1\%, 50\%$, $ 99\%, 99.5\%$ and $99.9\%$ quantiles are $-3.71, -3.32, -3.08, -0.45$, $1.61, 1.87$, and $2.20$, respectively.  
	Fig.~\ref{fig-v-kde-hist} and the quantiles indicate that $V_i$ takes many distinct values and has a large variation in the data, supporting its use as a regressor with a large support for identification.

	The error term $\varepsilon_{i}$ is the idiosyncratic case-level component of the latent score $Y_i^\ast$, after accounting for the baseline approval propensity $V_{i}$ and the environmental influence $g_0(W_{jt})$.  
	We assume that $\varepsilon_{i}$ is independent of $V_i,W_i$, where $\varepsilon_i$ has an unknown CDF $F_0$. 
	This condition is crucial for identification in our model and is treated as a maintained assumption in this illustration. 
	By construction, $V_{i}$ is a proxy for baseline approval propensity by applicant nationality, calendar month (year-month), and case type.  
	It is intended to absorb systematic heterogeneity along these dimensions, including at the $(c,t)$ level, so that the remaining variation in $\varepsilon_i$ reflects the residual within-cell idiosyncratic component.   
	In addition, environmental exposure $W_{jt}$ is plausibly exogenous in this setting because, as stated in \cite{heyes-saberian-19}, ``the setting of dates for cases and the rostering of judges is done well in advance'', and weather on the decision date is not manipulable by judges or applicants.

	For implementation, we estimate $(g_0, F_0)$ by the PC-regularized KNP method; that is, we solve \eqref{obj-aux-1-pca}, after which the estimator is given in \eqref{ghat-fhat-tau-zeta}, and then the APE and cAPEs are computed following Section \ref{subsection-wape-computation}. 
	Specifically, we standardize each component of $W_{jt}$ to have zero mean and unit variance and set $w_\ast = 0$, so the location normalization $g_0(w_\ast) = 0$ normalizes the utility to zero at the mean environmental variables (in the original units). 
	We use the Gaussian kernel $k(s,t)=\exp(-\|s-t\|^2/2)$. 
	The tuning parameters $B_n, J_n, m$ are selected by 5-fold cross-validation, which yields the choice $B_n=\sqrt{10}, J_n=4, m = 50$.\footnote{
		The candidate sets are $B_n^2 \in \{5,10,20,\dots, 50, 75, 100, 200, 500\}$, $J_n\in\{0,1,\dots, 7\}$, and $m\in \{10,25, 50, 75,100, 150,200\}$. 
		For each triplet $(B_n,J_n,m)$, we solve \eqref{obj-aux-1-pca} on the training folds and evaluate the squared-loss objective on the held-out fold. 
		We choose the triplet that minimizes the average validation loss across folds. 
	} 
	We solve the constrained optimization in \eqref{obj-aux-1-pca} using the MATLAB \texttt{fmincon} routine with the default interior-point method to handle the constraint. 
	We supply the routine with analytical gradients with respect to $\zeta$ and $\tau$;  
	the closed-form expressions for the objective, constraint, and gradients are given in Appendix \ref{appendix-implementation}. 
	All numerical optimizations are initialized at zeros.  
	When constructing bootstrap confidence intervals, we use a nonparametric bootstrap with $1000$ bootstrap replications, resampling cases $i$ with replacement. 
	In each bootstrap replication, we re-estimate $(g_0,F_0)$ using the same procedure as above while fixing $B_n=\sqrt{10}, J_n=4, m = 50$ at the values selected by cross-validation in the original sample.  
	Bootstrap confidence intervals are constructed using the basic (reverse-percentile) bootstrap method.\footnote{
		Specifically, for any scalar target parameter $\gamma_0$, let $\hat\gamma$ be the estimate obtained from the original sample, 
		and let $q^\ast_u$ denote the $u$-percentile of the centered bootstrap distribution $\{\hat\gamma^\ast_b - \hat\gamma\}_{b=1}^{1000}$, where $\hat\gamma^\ast_b$ denotes the estimate of $\gamma_0$ obtained from the $b$-th bootstrap replication. 
		The $(1-\alpha)$ confidence interval for $\gamma_0$ is given by $[\hat \gamma - q^\ast_{1-\alpha/2}, \hat \gamma - q^\ast_{\alpha/2}]$. 
		In our application, $\gamma_0$ may represent a pointwise value $g_0(w)$ at any prespecified evaluation point $w$, an APE or a cAPE.  
	}

	\subsection{Results}

	Our empirical goals are twofold: (a) to characterize how temperature enters the utility function $g_0(\cdot)$ of outdoor environmental variables, and (b) to quantify how temperature affects judges' grant decisions.

	\begin{figure}[h] 
		\caption{Estimated utility as a function of temperature}  
		\label{fig-utility-temperature}
		\begin{center} 
			\includegraphics[width=0.8\textwidth]{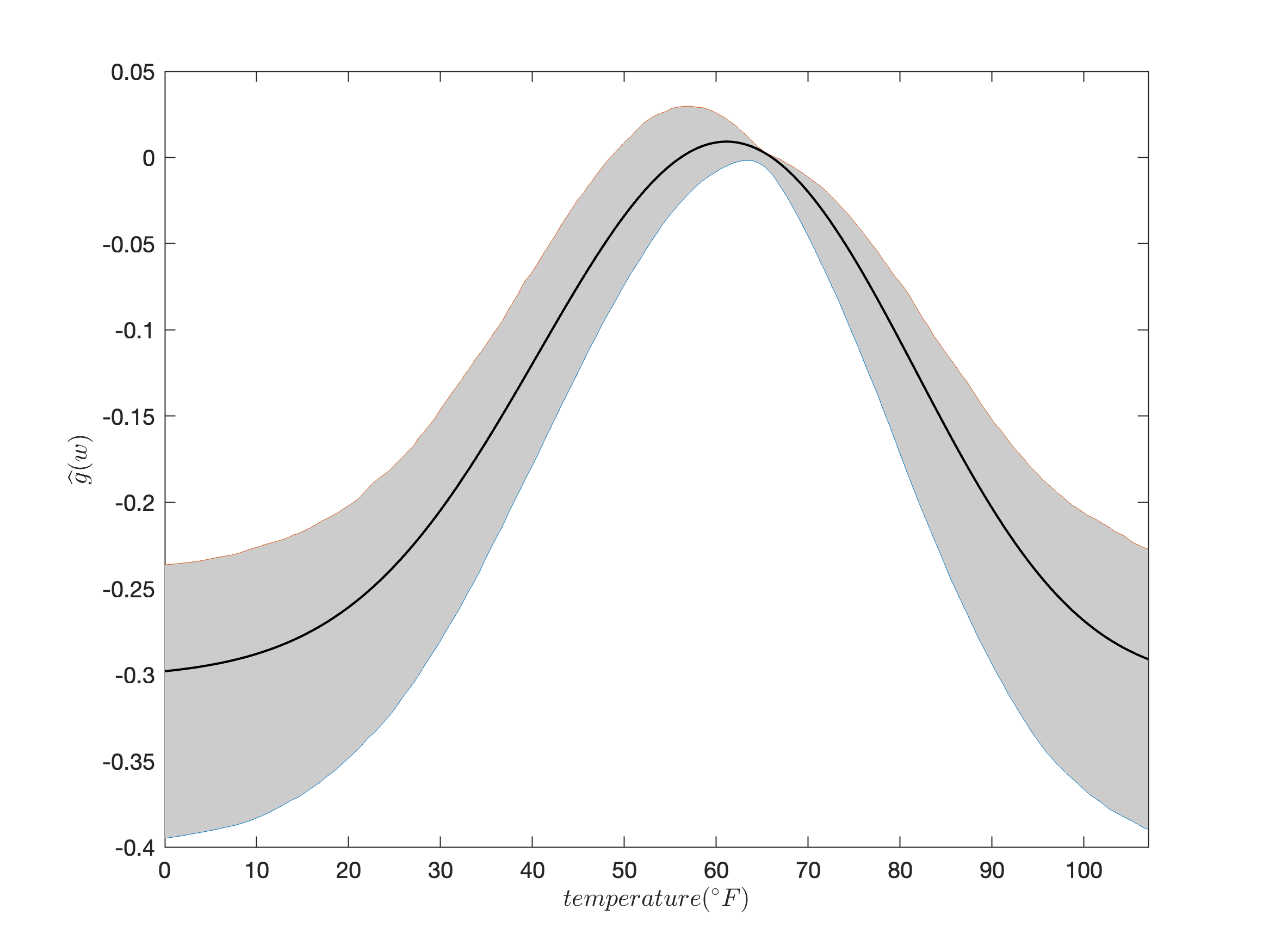}
		\end{center}
		
		\footnotesize
		Notes: This figure plots the KNP estimated utility function $g(w)$ of environmental variables $w$ as a function of temperature, with all other environmental variables fixed at their mean levels.  
		The effective sample size is $n= 99,773$.  
		The shaded area indicates a 90\% pointwise confidence band based on bootstrap with 1000 replications. 
		See the implementation paragraph for more details, including estimation and tuning parameter selections.
	\end{figure}

	For (a), Fig.~\ref{fig-utility-temperature} presents the estimated utility as a function of temperature in the original units, with all other environmental variables fixed at the mean level. The shaded area indicates 90\% pointwise bootstrap confidence band.  
	The figure shows that, when the temperature is at a high level, increasing temperature decreases utility. In contrast, when the temperature is at a low level, increasing temperature will increase the utility.

	For (b), we estimate the average partial effect (APE) and conditional average partial effects (cAPEs) of temperature $T$. In particular, other than the APE
	\[
	APE_T  = \mathbb{E} \frac{\partial}{\partial T} p(X),
	\]
	we consider the cAPEs conditioning on temperature higher (or lower) than $70^{\circ}$F, 
	\[
	cAPE_{T|T>70^{\circ}F} = \mathbb{E} \left( \frac{\partial}{\partial T} p(X) \bigg| T>70^{\circ}F  \right), \quad cAPE_{T|T<70^{\circ}F} = \mathbb{E} \left( \frac{\partial}{\partial T} p(X) \bigg| T<70^{\circ}F  \right).
	\]
	We choose $70^\circ$F as a benchmark close to typical indoor comfort temperature (``room temperature'') and use it to summarize heterogeneity in the temperature effect across low- and high-temperature ranges.

	Table \ref{table-apes} reports 90\% bootstrap confidence intervals from KNP for the APE and cAPEs, based on 1000 replications. 
	We compare the KNP results with those from an LPM benchmark that regresses $Y_i$ on the same covariates $(V_i,W_i)$ as in \eqref{bcm-application}. 
	Under the LPM, the estimated APE and both cAPEs are significantly negative, including  $cAPE_{T|T<70^{\circ}F}$.  
	In contrast, KNP yields an insignificant $cAPE_{T|T<70^{\circ}F}$ and a significantly negative $cAPE_{T|T>70^{\circ}F}$, with the overall APE also significantly negative.

	\begin{table}[h!]
		\caption{
			90\% confidence interval of estimated APEs and cAPEs (in \% of prob) of temperature  
		} 
		\label{table-apes}
		\begin{center}
			\centering
			\begin{tabular}{c c c c c c c c}
				\hline \hline 
				
				&	&	 APE &&  $cAPE_{T|T<70^{\circ}F}$	 	&&		$cAPE_{T|T>70^{\circ}F}$ 	\\ [0.5ex] \hline   
				KNP &  &	$[-0.378,  -0.065]$	&& $ [-0.267,     0.045] $	&&	$ [-0.636,   -0.050] $		\\   
				
				LPM & &  $[-3.083, -1.644]$  && $[-3.083, -1.644]$  && $ [-3.083, -1.644] $  \\
				
				\hline \hline
			\end{tabular}
		\end{center}

		\footnotesize
		Notes: This table presents 90\% confidence intervals for the estimated APEs and cAPEs obtained from fitting KNP and the linear probability model (LPM). 
		Effects are reported in percentage points of granting probability and correspond to a one standard deviation increase in temperature ($16.9^{\circ}$F). 
		The first row reports nonparametric bootstrap (1000 replications) confidence intervals given by the plug-in KNP estimator for model \eqref{bcm-application}.  
		The effective sample size is $n= 99,773$. 
		See the implementation paragraph for more details, including estimation and tuning parameter selections.  
		The second row reports confidence intervals obtained from fitting LPM on the same covariates $(V_i,W_i)$ as in \eqref{bcm-application}. 
	\end{table}

	\medskip

	These results complement \cite{heyes-saberian-19,heyes-saberian-22} by providing two empirically relevant outputs: an estimate of $g_0$ as a fully nonparametric function of the environmental variables $W$, and cAPEs of temperature across temperature ranges.  
	Notably, the estimated heterogeneous effects of temperature can potentially inform more targeted policy interventions. 
	While an LPM can allow for nonlinearity by adding higher-order and interaction terms in environmental variables, this quickly yields many polynomial terms (e.g., a 4-th order polynomial expansion of the 9-dimensional $W$ yielding 714 terms) and hard-to-interpret specifications.  
	In contrast, the estimation of cAPEs in our method is straightforward, with a computational procedure that remains essentially unchanged regardless of the dimension of $W$, since it does not require constructing or differentiating a large interaction basis.

	In terms of the average effect, our estimated APE is significantly negative and qualitatively consistent with \cite{heyes-saberian-19}, although the magnitude is smaller.\footnote{
		\cite{heyes-saberian-19} reports a 1.075\% decrease in grant probability per $10^{\circ}$F increase in temperature, while our 90\% CI implies a change in $[-0.378\%, -0.065\%]$ per $16.9^{\circ}$F (1 standard deviation) increase.  
	}  
	However, because we use a nonparametric nonlinear model, whereas \cite{heyes-saberian-19} uses an LPM with a rich set of fixed effects, differences in magnitude should not be taken literally, as they may reflect the different specifications.   
	While our approach has the advantage of capturing flexible nonlinear and interactive effects of the full environmental vector within a latent utility framework,\footnote{Under the maintained assumptions, the latent-utility formulation  
		provides a structural mapping from $(V,W)$ to the choice probabilities, which can be used to analyze counterfactuals holding the distribution of the error term fixed. 
	} 
	the tradeoff is that it does not accommodate a very fine set of fixed effects in the same way as linear models;  
	however, coarse controls (e.g., year indicators) can be incorporated by including dummies in the systematic component.

	\begin{remark}[Robustness check: Single-city subsample (New York)]
		Motivated by the concern that $\varepsilon_i$ may contain unobserved time components or court-location components that are correlated with environmental exposure $W_{jt}$ or with $V_{i}$, we conduct a robustness analysis using only cases decided in New York (NY). 
		In this subsample analysis, we include year fixed effects for 2001-2004 by adding dummy variables linearly in the systematic component. 
		See Appendix~\ref{appendix-NY-sample} for a detailed description of this analysis. 
		Table~\ref{table-apes-NY} in Appendix~\ref{appendix-NY-sample} reports results corresponding to Table~\ref{table-apes}. 
		Notably, KNP continues to capture heterogeneous effects of temperature: $cAPE_{T| T>70^\circ F}$ is significantly negative, while $cAPE_{T| T<70^\circ F}$ is significantly positive.
	\end{remark}

	\section{Conclusion}\label{section-conclusion}

	In this paper, we propose a new estimation procedure for a class of identified nonparametric binary choice models. 
	Compared to other possible methods, our estimation procedure is amenable to easier computation, especially when the number of covariates is non-small which may lead to a large number of basis functions if using the commonly used sieve method. 
	
	We show the proposed estimator has desirable asymptotic properties, and simulation studies suggest that the KNP estimator works well in finite samples. 
	We demonstrate the practical relevance of the proposed method by revisiting the effect of temperature on immigration judges’ latent utility, and thus, their decisions on asylum applications.

	In future work, a natural extension is to allow for further structures, motivated by either economic application-specific assumptions or existing theory, on the nonparametric and distribution-free BCM, without losing the tractability of the KNP procedure in terms of computation.\footnote{See, e.g., \cite{chetverikov-santos-shaikh-18} and the reference therein for a review of the roles of shape restrictions in identification, estimation and inference.} 
	For example, one may expect some covariates to affect the latent utility partially linearly. 
	Such a partial linear structure of latent utility can be easily incorporated into the KNP estimation procedure. 
	However, for more complex structures, such as monotonicity or concavity/convexity of the latent utility function in some covariates, the computation will be much more difficult under big data environments, since it involves more constraints on the function evaluations at the data points to restrict the shape of the function.

	\bibliographystyle{aea}
	\bibliography{bcm-knp}

	\clearpage
	
	\section*{Appendices}\label{section-appendix}
	
	\appendix 
	
	The appendices are organized as follows. 
	Appendix~\ref{appendix-rkhs} provides a brief introduction to RKHS and examples of large RKHS. 
	Appendices~\ref{appendix-id-proof} and \ref{appendix-proof-representer} contain the proofs for the identification result in Theorem~\ref{theorem-id-g0-F0-unique-min} and the representer-form solution for $\hat g$ in Proposition~\ref{proposition-representer}. 
	Appendix~\ref{appendix-rank-K} provides a lemma on the rank of the Gram matrix $K$ and its proof.  
	Appendices~\ref{appendix-proof-pc-near-opt}, \ref{appendix-consistency}, and~\ref{appendix-cgce-rate} contain proofs of Lemma \ref{lemma-pc-regularization-error}, the near-optimality result for the PC-regularized KNP, consistency in Theorem~\ref{theorem-consistency}, and the convergence-rate results in Theorem~\ref{theorem-cgce-rate-aux} and Corollary~\ref{corollary-cgce-rate}. 
	Appendix~\ref{appendix-AN} contains the assumptions, discussions, and the proof for Theorem \ref{theorem-asym-normality-wape} on the asymptotic normality of the weighted average partial derivatives. 
	Appendix~\ref{appendix-tech-lmas} collects additional technical lemmas used in the proofs of the main theorems. 
	Appendix~\ref{appendix-NY-sample} contains robustness analyses for the empirical application. 
	Appendix~\ref{appendix-implementation} provides closed-form expressions for the objective functions and analytical gradients used for computation. 
	Appendix~\ref{appendix-table-figures} collects tables from additional simulations and an additional figure for the empirical application.

	\section{Reproducing Kernel Hilbert Space}\label{appendix-rkhs}  
	
	Below we present a brief introduction to reproducing kernel Hilbert spaces, which are popular nonparametric settings in machine learning.

	Let $ \mathcal{W}\subset \mathbb R^d$. Let \emph{kernel} $k: \mathcal{W} \times \mathcal{W} \to \mathbb{R}$ be a symmetric function, which is \emph{positive definite}, in the sense that 
	\[
	\sum_{i=1}^N \sum_{j=1}^N a_i a_j k(w_i, w_j) \geq 0
	\]
	for any $a_i \in \mathbb{R}$, $w_i \in \mathcal W $, $i =1,\cdots,N$ and any positive integer $N$.
	One example is the class of Gaussian kernels, given by
	\begin{align}
		k(s,t) = \exp\Big( -\frac{\|s-t\|^2}{2\sigma^2} \Big)
	\end{align}
	for some $\sigma>0$. 
	
	For a given kernel $k$, by the Moore-Aronszajn theorem (e.g., Theorem 3 in \citealp[p.~19]{berlinet-thomas-11}), there exists a unique reproducing kernel Hilbert space $\mathbb{G}_k$ with reproducing kernel $k$, where $\mathbb{G}_k$ is the completion of the linear span of $\{k(\cdot, w) | w\in \mathcal W \}$ with inner product defined based on 
	\[
	\left\langle \sum_{i=1}^M a_i k(\cdot,s_i), \sum_{j=1}^N b_j k(\cdot,t_j) \right \rangle_{\mathbb{G}_k} = \sum_{i=1}^M \sum_{j=1}^N a_i b_j k(s_i,t_j)
	\]
	for any $s_i,t_j\in \mathcal W, a_i,b_j \in \mathbb{R}$ and integers $M, N$. Moreover, the reproducing property states that
	\begin{align}\label{reproducing-property}
		\langle g, k(\cdot,w) \rangle_{\mathbb{G}_k} = g(w)
	\end{align}
	for any $g\in \mathbb{G}_k, w\in \mathcal W$.

	\bigskip 
	
	\begin{remark}[Large RKHSs]\label{remark-largeRKHS}
		The consistency of the KNP estimator relies on the ability of functions in certain RKHSs to approximate the unknown function to be estimated. 
		Here are some examples. 
		\begin{enumerate}[(a)]\itemsep-0.1em
			\item When $\mathcal{W}$ is compact, for any $\sigma^2$, the Gaussian kernel RKHS is dense, under the supremum norm, in $C(\mathcal{W})$ the space of all continuous functions on $\mathcal{W}$. See, e.g., Corollary 4.58 in \cite{steinwart-christmann-08}. 
			
			\item For $\mathcal{W} = \mathbb{R}^d$, the Gaussian kernel RKHS is dense in $L_p(\mu)$ for any finite measure $\mu$ on $\mathbb{R}^d$ and any $p>1$. See, e.g., Theorem 4.63 in \cite{steinwart-christmann-08}. 
			
			\item The Sobolev space of functions on $\mathbb R^d$, whose derivatives up to order $m$ exist and are square-integrable, is an RKHS. However, the reproducing kernels of Sobolev spaces are typically difficult to compute, as they involve integrals over $\mathbb R^d$ that generally require numerical integration, except in certain simple cases, such as when $d=1$ or $m=\infty$; see, e.g., \cite{novak-et-al-18}. 
		\end{enumerate}
	\end{remark}

	\section{Identification: Proof for Theorem \ref{theorem-id-g0-F0-unique-min} } 
	\label{appendix-id-proof}
	
	We first give an identification result that $g_0$ and $F_0$ can be recovered from $p_0(x) := F_0(v+g_0(w) )$, which allows for possibly bounded support of $V$. 
	Then we prove the Theorem \ref{theorem-id-g0-F0-unique-min} for identification. 
	
	\begin{lemma}\label{lemma-aux-id-g0-F0-strongest-result}
		Assume that 
		$G_0(X) = V+ g_0(W)$ for $X=(V,W')'$. 
		Let $\mathcal V, \mathcal W$ denote the supports of $V,W$, respectively.   
		Let Assumption \ref{assumption-id-g0-F0} (\ref{condn-id-eps-large-support}), (\ref{condn-id-g-cts}), (\ref{condn-id-gwast-0}) hold. 
		Assume further 
		\begin{enumerate}\itemsep-0.1em 
			\item [(d)] The support of the conditional distribution $\mathcal{L}(V|W=w_\ast)$ is the same as $\mathcal V$, the support of $\mathcal{L}(V)$;  
			\label{condn-lmaB1-Vwast-large-support}
			\item [(e)] Either $\mathcal V = \mathbb R$, or there exists a point $v_\ast \in \mathcal V$ such that $\mathcal{L}(W|V=v_\ast)$ has support $\mathcal W$ and the ranges of all functions in $\mathcal G$ are subsets of $[L_{lb}, L_{ub}]$ with $[v_\ast+L_{lb}, v_\ast+L_{ub} ] \subset \mathcal V$.
			
			\label{condn-lmaB1-V-large-support}
		\end{enumerate} 
		Then for any $\theta = (g,F) \in \mathcal{G} \times \mathcal{F}$ satisfying $p_{\theta}(X) = p_{\theta_0}(X)$ almost everywhere, where $p_{\theta}(x) = F(v+g(w))$, it holds that $g(w) = g_0(w)$ for any $w \in \mathcal{W}$, and $F(u) = F_0(u)$ for any $u \in \mathcal V$.
	\end{lemma}

	\begin{remark}[Identifying $g$ does not require $\mathcal V = \mathbb R$]\label{remark-id-g-not-require-SuppV=R} 
		In Lemma \ref{lemma-aux-id-g0-F0-strongest-result}, the identification of $g$ essentially requires the support of $V$ to cover the support of $g(W)$. 
		For instance, when $V$ is independent of $W$, $0 \in \mathcal V$, and $\varepsilon$ has support $\mathbb R$, the support condition on $\mathcal V$ for the identification of $g$, Condition (e) in Lemma \ref{lemma-aux-id-g0-F0-strongest-result}, is satisfied whenever $ \{ g(w): w\in \mathcal W\} \subset \mathcal V$. 
		In particular, having $\mathcal V=\mathbb R$ is sufficient but not necessary for the identification of $g$.  
		However, Lemma \ref{lemma-aux-id-g0-F0-strongest-result} identifies $F$ only on $\mathcal V$, so the condition $\mathcal V = \mathbb R$ is needed to identify $F$ on $\mathbb R$. 
		In the main text, we impose $\mathcal V=\mathbb R$ for simplicity; this is mainly for the identification of $F$ on $\mathbb R$, rather than for identifying $g_0$. 
	\end{remark}

	\begin{proof}[Proof of Lemma \ref{lemma-aux-id-g0-F0-strongest-result}]
		Let $(g,F) \in \mathcal{G} \times \mathcal{F}$ be such that $p_{\theta}(X) = p_{\theta_0}(X)$ for $X$-almost everywhere. 
		That is, $F(V + g(W)) = F_0(V + g_0(W))$ for $X=(V,W')'$-almost everywhere. 
		Denote by $A \subset \mathcal{X}$ the set of $x=(v,w')'$ on which  $F(v + g(w)) = F_0(v + g_0(w))$. Notice $P_{X}(A)=1$. 
		
		We prove this lemma by the following steps. (I) Show that $F(u) = F_0(u)$ for $u\in \mathcal V$. (II) Show that $g(w)=g_0(w)$ for any $w\in \mathcal W$. 
		
		For (I), let $v\in \mathcal V$ be fixed arbitrarily. Notice $x :=(v,w_\ast')' \in \mathcal{X}$ by Condition~(d).  
		Denote by $B_{1/n}(x)$ the open ball with radius $1/n$ and center $x$ for each $n =1,2,\dots$. 
		Notice $\mathbb{P}\{ X \in B_{1/n}(x) \} >0$ since $x\in \text{Supp}(X)$, and thus, $\mathbb{P}\{ X \in B_{1/n}(x) \cap A \} >0$ since $P_{X}(A) = 1$. 
		
		We pick any $x_n =(v_n,w_n')' \in  B_{1/n}(x) \cap A$ for each $n$. Notice $\|x_n - x\| \to 0$ as $n \to \infty$. 
		Since $g, g_0$ are continuous as in Assumption~\ref{assumption-id-g0-F0}(\ref{condn-id-g-cts}),  
		we have $v_n + g(w_n) \to v + g(w_\ast)$ and $v_n + g_0(w_n) \to v + g_0(w_\ast)$ as $n \to \infty$. 
		Furthermore, $F(v_n + g(w_n)) \to F(v + g(w_\ast))$ and $F_0(v_n + g_0(w_n)) \to F_0(v + g_0(w_\ast))$ as $n \to \infty$, since $F,F_0$ are continuous by Assumption~\ref{assumption-id-g0-F0}(\ref{condn-id-eps-large-support}).  
		Therefore, $F(v + g(w_\ast)) = F_0(v + g_0(w_\ast))$ follows immediately upon noticing that $F(v_n + g(w_n)) = F_0(v_n + g_0(w_n))$ since $x_n =(v_n,w_n')' \in A$. 
		Since $g(w_\ast)=0$ under Assumption~\ref{assumption-id-g0-F0}(\ref{condn-id-gwast-0}),  
		$F(v) = F_0(v)$ for any $v \in \mathcal V$.

		For (II), let $w\in \mathcal W$ be fixed arbitrarily. 
		There exists some $v$ such that $x :=(v,w')'\in\mathcal X$, and we choose $v=v_\ast$ when $\mathcal V \subsetneq \mathbb R$. 
		Let $B_{1/n}(x)$ the open ball with radius $1/n$ and center $x$ for each $n =1,2,\dots$. Notice $\mathbb{P}\{ X \in B_{1/n}(x) \} >0$, and thus, $\mathbb{P}\{ X \in B_{1/n}(x) \cap A \} >0$ since $P_{X}(A) = 1$. 
		We pick any $x_n =(v_n,w_n')' \in  B_{1/n}(x) \cap A$ for each $n$, and thus, $F(v_n +g(w_n)) = F_0(v_n+g_0(w_n))$ for each $n$. Since $g, g_0, F, F_0$ are continuous, we have $F(v +g(w)) = \lim_{n\to\infty} F(v_n +g(w_n)) = \lim_{n\to\infty} F_0(v_n +g_0(w_n))  = F_0(v +g_0(w))$. 
		Since $v+g(w) \in \mathbb R$ and $v_\ast+g(w) \in [v_\ast+L_{lb}, v_\ast+L_{ub}] \subset \mathcal V$ in case of $\mathcal V \subsetneq \mathbb R$ under Condition (e), 
		it follows from (I) that $F_0(v +g(w)) = F(v +g(w)) = F_0(v +g_0(w))$, and thus, $g(w) = g_0(w)$ since $F_0$ is strictly increasing.  
	\end{proof}

	\begin{proof}[Proof of Theorem \ref{theorem-id-g0-F0-unique-min}]
		Since $\mathcal X$ is the support of $X$, for any $p:\mathcal X \to [0,1]$ which is continuous on $\mathcal X$, $p(x_\circ) \lessgtr p_0(x_\circ)$ at some $x_\circ\in \mathcal X$ implies that there exists a neighborhood $\mathcal N$ around $x_\circ$ such that $p(x) \lessgtr p_0(x)$ for any $x\in \mathcal N$ and $P_X(\mathcal N) >0$.  
		Thus, both $p \mapsto - \mathbb{E} ( Y\log p(X) + (1-Y)\log (1-p(X))  )$  and $p\mapsto \mathbb{E} (Y-p(X))^2$ have a unique minimum at $p_0$ over all continuous functions $p: \mathcal{X} \to [0,1]$. 
		Note that all conditions for Lemma~\ref{lemma-aux-id-g0-F0-strongest-result} are satisfied,  
		since Assumption~\ref{assumption-id-g0-F0}(\ref{condn-id-V-large-support}) implies Conditions~(d), (e) in Lemma~\ref{lemma-aux-id-g0-F0-strongest-result}. 
		Applying Lemma \ref{lemma-aux-id-g0-F0-strongest-result} yields that $\theta_0 = (g_0,F_0)$ is the unique minimum of $Q(\theta)$ over $\mathcal G \times \mathcal F$. 
	\end{proof}

	\section{Representer Form Solution: Proof for Proposition \ref{proposition-representer}}\label{appendix-proof-representer}

	\begin{proof}[Proof of Proposition \ref{proposition-representer}]
		Pick an arbitrary $F\in \mathcal{F}$, and $g\in \mathcal{G}_n$ given by $\tilde g - \tilde g(w_\ast)$. Thus, $\tilde g \in \mathbb{G}_k$ with $\|\tilde g\|_{\mathbb{G}_k} \leq B_n$.  
		Recall $W_0:=w_\ast$, and let $\tilde g_\ast$ be the orthogonal projection of $\tilde g$ onto linear span of functions $k(W_0,\cdot), k(W_1,\cdot)$, $\dots, k(W_n,\cdot) \in \mathbb{G}_k $. Then $\tilde g_\ast$ has the form $\tilde g_\ast(w) = \sum_{j=0}^n \delta_j k(W_j,w)$ for some $(\delta_j)_{j=0}^n$. Notice that, for each $j=0,1,\dots,n$, 
		\[
		\tilde g(W_j) - \tilde g_\ast(W_j)  = \langle \tilde g- \tilde g_\ast, k(W_j,\cdot) \rangle_{\mathbb{G}_k} =  0, 
		\]
		where the first equality follows from the reproducing property of an RKHS and the second equality is by the definition of $\tilde g_\ast$. 
		Thus, $g(W_j) = \tilde g(W_j)-\tilde g(W_0) =  \tilde g_\ast(W_j)-\tilde g_\ast(W_0) = g_\ast(W_j)$. 
		Moreover, $\|\tilde g_\ast\|_{\mathbb{G}_k} \leq \|\tilde g\|_{\mathbb{G}_k} \leq B_n$, and thus, $g_\ast:=\tilde g_\ast - \tilde g_\ast(w_\ast) \in \mathcal G_n$.
		
		Since $\hat Q(\theta)$ depends on $g$ only through values of $g(W_j)$,  $g(W_j) = g_\ast(W_j)$ for each $j=1,\dots,n$ implies that $(g,F)$ and $(g_\ast,F)$ give the same value for the objective function.
	\end{proof}

	\begin{remark}
		[Non-representer minimizers when the norm constraint is non-binding] 
		\label{remark-non-representer-min} 
		Proposition \ref{proposition-representer} does not ensure that every $\tilde g$ for a minimizer $g=\tilde g - \tilde g(w_\ast)$ admits the representer form. 
		To see this, let $S\subset \mathbb G_k$ be the span of $k(W_0,\cdot), \dots, k(W_n,\cdot)$ and $S_\perp\subset \mathbb G_k$ be its orthogonal complement. Suppose that $(g,F)\in \mathcal G_n\times \mathcal F_n$ is a minimizer of $\hat Q(\theta)$, where $g = \tilde g - \tilde g(w_\ast)$ for some $\tilde g\in S$. 
		If the norm constraint for this minimizer is not binding, i.e., $\|\tilde g\|_{\mathbb G_k} < B_n$, then we can choose $\tilde g_{\perp} \in S_\perp$ with RKHS norm small enough so that $\|\tilde g + \tilde g_{\perp}\|_{\mathbb G_k}^2 = |\tilde g \|_{\mathbb G_k}^2 + \| \tilde g_{\perp}\|_{\mathbb G_k}^2 \leq B_n^2$.  
		Since $\tilde g_{\perp}\in S_\perp$, we have $\tilde g_\perp(W_j) = \langle k(W_j,\cdot), \tilde g_\perp \rangle_{\mathbb G_k} = 0$ for all $j=0,1,\dots,n$, so replacing $\tilde g$ by $\tilde g + \tilde g_{\perp}$ leaves $g(W_j)$'s (and hence the objective value since $W_0=w_\ast$) unchanged. 
		Since $\tilde g +\tilde g_\perp \notin S$, this yields a minimizer that does not admit the form in Proposition \ref{proposition-representer}. 
	\end{remark}

	\section{Rank of the Gram Matrix}\label{appendix-rank-K}
	
	We provide the following lemma regarding the rank of the $n\times n$ Gram matrix $K$ whose $(i,j)$-th element is given by $k(w_i,w_j)$. 
	The lemma states that $K$ has full rank, provided that the RKHS $\mathbb G_k$ is dense in the space of all continuous functions and the $w_i$'s values are mutually different.   
	More generally, $K$ has rank the cardinality of $\{w_1,\dots,w_n\}$.\footnote{Although this may be a widely recognized fact, we were unable to find a specific reference.  
		Additionally, while it would have been easier to prove the result using the properties of $k$ as an integral operator, we instead choose to derive the rank based on the denseness property of RKHS spaces, to align with the perspective of viewing RKHS as special sieve spaces. 
	}

	\begin{lemma}
		[Rank of the Gram matrix] \label{lemma-kernel-universal-psd}
		Let the kernel function $k$ be such that its RKHS $\mathbb G_k$ is dense in the space of all continuous functions on any compact domain $\mathcal W$ under the uniform norm.
		Then, for any $n\in \mathbb N$, the $n\times n$ Gram matrix $K$ given by $k(w_i,w_j)$ for any mutually different points $w_1,\dots,w_n\in \mathcal W$ has full rank $n$. 
		
		Furthermore, if $\{w_1,\dots,w_n\}$ has cardinality $m<n$, then $K$ has rank $m$. 
	\end{lemma}

	\begin{proof}[Proof of Lemma \ref{lemma-kernel-universal-psd}]
		We first prove the first part. 
		Suppose not; then there exist mutually different points $w_1,\dots,w_n\in \mathcal W$ such that $K$ has reduced rank. Thus, there exists $y\in \mathbb R^n$ such that $K \delta \neq y $ for any  $\delta \in \mathbb R^n$.  
		Since $\{K\alpha | \alpha \in \mathbb R^n\}$ is closed, there exists $\epsilon >0$ such that 
		\begin{align}\label{pf-universal-psd-aux-1}
			\|K \delta - y \|_{\infty} > \epsilon  \quad \text{ for any } \delta \in \mathbb R^n
		\end{align} 
		For any $f \in \mathbb G_k$, let $f_\circ$ be its orthogonal projection onto the span of functions $k(w_1,\cdot), \dots$, $ k(w_n,\cdot)$, which is given by $f_\circ(\cdot) = \sum_{j=1}^n \alpha_j k(w_j,\cdot)$ 
		for some $\alpha = (\alpha_1,\dots, \alpha_n)'$ depending on $f$. 
		Notice for any $w_i$'s, 
		\begin{align*}
			f(w_i) & = \langle f, k(w_i,\cdot) \rangle_{\mathbb G_k} = \langle f - f_\circ , k(w_i,\cdot) \rangle_{\mathbb G_k}  + \langle f_\circ , k(w_i,\cdot) \rangle_{\mathbb G_k} \\
			& =  \langle f_\circ , k(w_i,\cdot) \rangle_{\mathbb G_k} = \sum_{j=1}^n \alpha_j k(w_j,w_i) 
		\end{align*}
		Then $( f(w_1), \dots, f(w_n) )' =K \alpha$, and thus, by \eqref{pf-universal-psd-aux-1} it holds that
		\begin{align}\label{pf-universal-psd-aux-2}
			|f(w_i) - y_i | > \epsilon \quad \text{ for any } i = 1,\dots,n. 
		\end{align}

		Let $g$ be a continuous function such that $g(w_i) = y_i$, whose existence is ensured by construction since $w_1,\dots,w_n$ are mutually different. 
		There exists $ f \in \mathbb G_k$ such that  $\| f - g\|_\infty <\epsilon/2$,  
		since $\mathbb G_k$ is dense in the space of all continuous functions under the uniform norm. 
		Thus, $|f(w_i) - g(w_i) | = | f(w_i) - y_i| < \epsilon/2$ for any $i = 1,\dots,n$, which contradicts against \eqref{pf-universal-psd-aux-2}.  
		Thus, $K$ has full rank if points $w_1,\dots,w_n$ are mutually different.

		\medskip 
		Now, suppose $\{w_1,\dots,w_n\}$ has cardinality $m<n$. 
		Without loss of generality, assume that $w_1,\dots, w_m$ are distinct. Then the first $m$ columns can replicate all columns of $K$, so $\text{rank}(K)\leq m$.  
		Moreover, by the first part, the $m\times m$ submatrix consisting of the first $m$ rows and the first $m$ columns of $K$ has rank $m$, which implies $\text{rank}(K)\geq m$. 
		Thus, $K$ has rank $m$. This completes the proof. 
	\end{proof}

	\section{Estimate with Spectral Cutoff: Proof for Lemma \ref{lemma-pc-regularization-error}}
	\label{appendix-proof-pc-near-opt}

	\begin{proof}[Proof of Lemma \ref{lemma-pc-regularization-error}]
		By Proposition \ref{proposition-representer}, 
		\begin{align}\label{pf-pc-regularization-error-aux-0}
			\inf_{g\in \mathcal G_n, F\in \mathcal F_n} \hat Q(\theta) = \inf_{\delta'K\delta \leq B_n^2, F\in \mathcal F_n} \hat Q(g_{\delta}, F)
		\end{align}
		where $g_{\delta}(\cdot) := \sum_{j=0}^n \delta_{j} \big(k(W_j,\cdot) - k(W_j,w_\ast) \big)$, and
		\[
		\hat Q(g_{\delta}, F) = \frac{1}{n} \sum_{i=1}^n \Big( Y_i - F\big(V_i + [K\delta]_{i+1} - [K\delta]_{1} \big)   \Big)^2. 
		\]
		
		Pick an arbitrary $\delta$ satisfying $\delta'K\delta \leq B_n^2$, and an arbitrary distribution function $F$ in $\mathcal F_n$ whose probability density is $f$. 
		Then define 
		\[
		\delta_{pc} = \hat U_m \hat U_m'\delta \quad \text{and} \quad  \zeta_{\delta} = \hat \Lambda_m \hat U_m'\delta
		\]
		and note $\delta_{pc} = \hat U_m \hat \Lambda_{m}^{-1} \zeta_{\delta}$. 
		Let $\hat \Lambda_r = \text{diag}(\hat\lambda_{m+1}, \dots, \hat \lambda_{n+1})$, and $\hat U_r$ be the matrix whose columns are the $m+1,\dots, n+1$-th eigenvectors of $K$. 
		Notice $K\delta_{pc} = (\hat U_m \hat\Lambda_m \hat U_m' + \hat U_r \hat\Lambda_r U_r' ) \delta_{pc} = \hat U_m\zeta_{\delta}$, and thus
		\[
		\hat Q(g_{\delta_{pc}},F) = \frac{1}{n} \sum_{i=1}^n \Big( Y_i - F\big(V_i + [\hat U_m\zeta_\delta]_{i+1} - [\hat U_m\zeta_\delta]_{1} \big)   \Big)^2 
		\]
		Consequently, by definition of the PC regularized KNP estimator defined through \eqref{obj-aux-1-pca} and \eqref{ghat-fhat-tau-zeta}, 
		\begin{align*}
			\hat{Q}( \hat g_{pc}, \hat F_{pc} )  & \leq \hat Q(g_{\delta_{pc}},F)  \\
			&  = \hat Q(g_\delta,F) + \hat Q(g_{\delta_{pc}},F) - \hat Q(g_\delta,F)  \leq \hat Q(g_\delta,F) + \left| \hat Q(g_\delta,F) -  \hat Q(g_{\delta_{pc}},F) \right| 
		\end{align*} 
		If we can verify that
		\begin{equation}\label{pf-pc-regularization-error-aux-1}
			\left| \hat Q(g_\delta,F) -  \hat Q(g_{\delta_{pc}},F) \right| \leq 4M_{\mathcal F} B_n\sqrt{ \hat \lambda_{m+1} } 
		\end{equation}
		then Lemma \ref{lemma-pc-regularization-error} follows immediately from \eqref{pf-pc-regularization-error-aux-0}. 
		
		Now it remains to prove \eqref{pf-pc-regularization-error-aux-1}. 
		Since $ \big|2Y_i - F\big(V_i + [K\delta]_{i+1} - [K\delta]_{1} \big)  - F\big(V_i + [\hat U_m\zeta]_{i+1} - [\hat U_m\zeta]_{1}\big)  \big| <2$, we have 
		\begin{align}\label{pf-pc-error-aux-1}
			& \left| \hat Q(g_\delta,F) -  \hat Q(g_{\delta_{pc}},F) \right| \notag \\
			& \quad = \left| \frac{1}{n} \sum_{i=1}^n \Big( Y_i - F\big(V_i + [K\delta]_{i+1} - [K\delta]_{1} \big)   \Big)^2  - \frac{1}{n} \sum_{i=1}^n \Big( Y_i - F\big(V_i + [\hat U_m\zeta_{\delta}]_{i+1} - [\hat U_m\zeta_{\delta}]_{1} \big)   \Big)^2  \right|  \notag \\
			& \quad \leq 2 \frac{1}{n}\sum_{i=1}^n  \Big|  F\big(V_i + [\hat U_m\zeta_{\delta}]_{i+1} - [\hat U_m\zeta_{\delta}]_{1} \big)  - F\big(V_i + [K\delta]_{i+1} - [K\delta]_{1}\big)   \Big|   \notag \\ 
			& \quad \leq 2 \|f\|_\infty \frac{1}{n}\sum_{i=1}^n  \Big| [K\delta]_{i+1} - [\hat U_m\zeta_{\delta} ]_{i+1} - ( [K\delta]_{1} - [\hat U_m\zeta_{\delta} ]_{1} )  \Big|   \notag \\
			& \quad \leq 2M_{\mathcal F} \frac{1}{n}\sum_{i=1}^n \Big| [\hat U_r \hat \Lambda_r \hat U_r' \delta]_{i+1} - [\hat U_r \hat \Lambda_r \hat U_r' \delta]_{1} \Big|.  
		\end{align} 
		where the last line follows from $K\delta = \hat U_m \hat \Lambda_m \hat U_m' \delta +  \hat U_r \hat \Lambda_r \hat U_r' \delta $ and $\hat U_m \zeta_{\delta} =\hat U_m \hat \Lambda_m \hat U_m' \delta$ by the definition of $\zeta_{\delta}$.  
		Moreover
		\begin{align}\label{pf-pc-error-aux-2}
			& \frac{1}{n}\sum_{i=1}^n \Big| [\hat U_r \hat \Lambda_r \hat U_r' \delta]_{i+1} - [\hat U_r \hat \Lambda_r \hat U_r' \delta]_{1} \Big| \leq \frac{1}{n}\sum_{i=1}^n \big| [\hat U_r \hat \Lambda_r \hat U_r' \delta]_{i+1} \big| + \big|[\hat U_r \hat \Lambda_r \hat U_r' \delta]_{1} \big| \notag \\
			& \quad \leq \frac{1}{\sqrt{n}} \sqrt{ \delta'\hat U_r \hat \Lambda_r^2 \hat U_r' \delta } + \sqrt{\delta'\hat U_r \hat \Lambda_r^2 \hat \hat U_r' \delta } \notag \\
			& \quad \leq 2 B_n \sqrt{ \hat \lambda_{m+1} }
		\end{align}
		where the last line is due to $\delta'\hat U_r \hat \Lambda_r^2 \hat U_r' \delta = \sum_{j=m+1}^{n+1} \hat \lambda_j^2 (\hat u_j'\delta)^2 \leq  \hat \lambda_{m+1} \sum_{j=m+1}^{n+1} \hat \lambda_j (\hat u_j'\delta)^2 \leq  \hat \lambda_{m+1}  \delta'\hat U_r \hat \Lambda_r \hat U_r' \delta \leq \hat \lambda_{m+1} \delta'K\delta \leq \hat \lambda_{m+1} B_n^2$. 
		Combining \eqref{pf-pc-error-aux-1} and \eqref{pf-pc-error-aux-2} proves \eqref{pf-pc-regularization-error-aux-1}, which was to be shown.  
	\end{proof}

	\section{Consistency: Proof for Theorem \ref{theorem-consistency} } 
	\label{appendix-consistency}
	
	\begin{proof}[Proof of Theorem \ref{theorem-consistency}] 
		
		Under the following conditions which will be verified later, the consistency for $\hat \theta$ follows from Lemma A1 in \citet[][p. 11]{chernozhukov-imbens-newey-07}.  
		\begin{enumerate}[(i)] 
			\itemsep -0.15em
			\item $Q(\theta) = \mathbb{E}\ell(\theta,Z)$ has a unique minimum at $\theta_0$ on $\Theta$;  
			\item Under the metric $d_\Theta$, $\Theta$ is compact;
			\item $\sup_{\theta \in \Theta} |\hat{Q}_n(\theta) - Q(\theta) | \to_p 0$ and $Q(\theta)$ is continuous under $d_\Theta$; 
			\item There exists $\theta_n \in \Theta_n := \mathcal{G}_n \times \mathcal{F}_n \subset \Theta$ such that $d(\theta_n,\theta_0) \to 0$ as $n \to \infty$. 
		\end{enumerate}
		This also applies to the PC regularized estimator $\hat \theta_{pc}$ due to Lemma \ref{lemma-pc-regularization-error} and $\hat \lambda_m^{1/2} B_n = o_p(1)$.  
		When provided with $d_{\Theta}(\hat\theta,\theta_0), d_{\Theta}(\hat\theta_{pc},\theta_0)\to_p$, the consistency of $\hat p$ and $\hat p_{pc}$ 
		follows immediately by Lemma \ref{lemma-aux-Fset-prelim-properties} (d).

		Now it remains to verify Conditions (i)-(v). 
		Condition (i) is satisfied by the identification result in Theorem \ref{theorem-id-g0-F0-unique-min}. 
		Condition (ii) holds due to Lemma \ref{lemma-aux-Fset-prelim-properties} (a), together with the compactness of H{\"o}lder balls under the uniform norm. 
		
		For Condition (iii), notice it is satisfied by Lemma A2 in \cite{newey-powell-03} once provided that (iii.a) $\Theta$ is compact under $d_\Theta$, (iii.b) $\hat Q(\theta) \to_p Q(\theta)$ for any $\theta \in \Theta$, (iii.c) there exists $\nu>0$ and $C_n = O_p(1)$ such that $|\hat Q(\theta) - \hat Q(\tilde \theta) | \leq C_n d_\Theta(\theta,\tilde \theta)^\nu$ for any $\theta,\tilde\theta \in \Theta$. Here, Conditions (iii.a) and (iii.b) are satisfied trivially. 
		Moreover, 
		\begin{align*}
			|\hat Q(\theta) - \hat Q(\tilde \theta) | & \leq \frac{1}{n}\sum_{i=1}^n |\ell(Z_i,\theta) - \ell(Z_i,\tilde \theta) | \\
			& = \frac{1}{n}\sum_{i=1}^n |2Y_i - p(X_i,\theta)- p(X_i,\tilde \theta) | | p(X_i,\theta) - p(X_i,\tilde\theta) | \\ 
			& \leq 2  \max\{1,M_{\mathcal F} \} d_{\Theta}(\theta,\tilde \theta)
		\end{align*}
		where the last inequality is due to Lemma \ref{lemma-aux-Fset-prelim-properties} (d) and  $|2Y_i - p(X_i,\theta)- p(X_i,\tilde \theta) | \leq 2$. 
		Hence, Condition (iii.c) and thus Condition (iii) are satisfied.
		
		Condition (iv) is satisfied by Assumption \ref{assumption-consistency} (c), (d), (e), along with Lemma \ref{lemma-aux-Fset-prelim-properties} (c).
	\end{proof}

	\section{Convergence Rate: Proofs for Theorem \ref{theorem-cgce-rate-aux}, Corollary \ref{corollary-cgce-rate}}
	\label{appendix-cgce-rate}

	\begin{proof}[Proof of Theorem \ref{theorem-cgce-rate-aux}] 
		For notational simplicity, we write $\|\cdot\|_{L_2(X)}$ as $\|\cdot\|_2$ in this proof. 
		By Theorem 3.2 in \citet[][p. 5595]{chen-07}, \eqref{ccp-cgce-rate-aux} holds once provided with the following conditions, which will be verified later. 
		\begin{enumerate}[(i)]\itemsep-0.1em
			\item There exists $C>0$ such that $\|p_{\theta} - p_{\theta_0} \|_2 \leq C d_{\Theta}(\theta,\theta_0)$ for all $\theta = (g,F)\in \mathcal G\times \mathcal F$.
			\item There exists $C_1,C_2>0$ such that $C_1 \mathbb E(\ell(Z,\theta)-\ell(Z,\theta_0)) \leq \|p_{\theta} - p_{\theta_0} \|_2^2 \leq C_2 \mathbb E(\ell(Z,\theta)-\ell(Z,\theta_0))$ for all $\theta = (g,F)\in \mathcal G\times \mathcal F$.
			\item There exists $C$ such that, for all small $\epsilon>0$, 
			\[
			\sup_{\theta\in \Theta_n: \|p_{\theta}-p_0\|_2 <\epsilon} \text{Var}\left( \ell(Z,\theta)- \ell(Z,\theta_0) \right) \leq C \epsilon^2
			\]
			\item There exists a constant $s \in(0,2)$ such that 
			\[
			\sup_{\theta\in \Theta_n: \|p_{\theta}-p_0\|_2 <\epsilon} \left| \ell(Z,\theta)- \ell(Z,\theta_0)  \right| \leq \epsilon^s U(Z) 
			\]
			with $\mathbb E U(Z)^{c} <\infty$ for some $c \geq 2$.
			\item There exists $\theta_n = (g_n,F_n) \in \mathcal G_n \times \mathcal F_n$ such that $\|p_{\theta_n} - p_0 \|_2^2 = O\!\left( J_n^{-m_e} +  \left(\log B_n\right)^{-m_w/2} ) \right)$. 
			\item There exists $b>0$ such that, for all large $n$,  
			\[
			\frac{1}{\sqrt{n} \gamma_n^2} \int_{b\gamma_n^2}^{\gamma_n} \sqrt{\log N_{[]}(\delta, \mathcal A_n, \|\cdot\|_2) } d\delta <C
			\]
			where $\mathcal A_n  = \left\{ \ell(\cdot,\theta) - \ell(\cdot,\theta_0) | \  \|p_{\theta}-p_{\theta_0} \|_2 \leq \gamma_n. \theta\in \Theta_n \right\}$. 
		\end{enumerate} 
		The statement for $\hat p_{pc}$ holds since Theorem 3.2 in \citet[][p. 5595]{chen-07} applies as long as  $\hat\lambda_m^{1/2} B_n = O_p\!\left( \delta_n \right)$ due to Lemma \ref{lemma-pc-regularization-error}.

		Now we verify Conditions (i)-(vi). 
		Notice Condition (i) is satisfied with $C = \max\{1,M_{\mathcal F}\}$ by Lemma \ref{lemma-aux-Fset-prelim-properties} (d).  
		
		Condition (ii) holds for $C_1=C_2=1$ due to the least square loss, since
		\begin{align}\label{pf-rate-aux-1}
			\mathbb E & \left( \ell(Z,\theta)- \ell(Z,\theta_0) \right)= \mathbb E\left( (Y-p(X,\theta))^2 - (Y-p(X,\theta_0))^2  \right) \notag \\
			& = \mathbb E(2Y-p(X,\theta)-p(X,\theta_0) )  (p(X,\theta)-p(X,\theta_0)) \\
			& = \mathbb E  (p(X,\theta)-p(X,\theta_0))^2 = \|p_{\theta}-p_{\theta_0}\|^2  \notag 
		\end{align}

		Condition (iii) is satisfied for $C=4$, since from \eqref{pf-rate-aux-1} we have 
		\begin{align*}
			\mathbb E & \left( \ell(Z,\theta)- \ell(Z,\theta_0) \right)^2  
			= \mathbb E(2Y-p(X,\theta)-p(X,\theta_0) )^2  (p(X,\theta)-p(X,\theta_0))^2 \\
			& \leq 4 \mathbb E(p(X,\theta)-p(X,\theta_0))^2 \leq 4 \epsilon^2. 
		\end{align*}
		when provided $\norm{p_{\theta}-p_0}_2 <\epsilon$.

		For Condition (iv), Lemma \ref{lemma-aux-p-p0-interpolation-ineq} gives 
		\[ 
		\sup_{x\in \mathcal X} |p_{\theta}(x) - p_0(x)| \leq   C \left(\norm{p_{\theta} - p_0}_2 \right)^{ 2/(2+d_x) } 
		\]
		for some constant $C>0$ which does not depend on $\theta$.  
		Thus, 
		\begin{align*}
			& \left| \ell(Z,\theta)- \ell(Z,\theta_0)  \right|  = |2Y- p(X,\theta) - p(X,\theta_0)| \times| p(X,\theta)- p(X,\theta_0)| \\
			& \quad \leq 2 |p(X,\theta)- p(X,\theta_0)| \leq 2C  \|p_{\theta} - p_{\theta_0} \|_2^{2/(2+d_x)} 
		\end{align*}
		and Condition (iv) holds with $s = 2/(2+d_x)$ and $U(Z) = 2C  $.

		To verify Condition (v), notice there exists $F_n\in \mathcal F_n$ such that $\big(\|F_n -F_0\|_\infty\big)^2 =  o\left( J_n^{-m_e} \right)$ by Lemma \ref{lemma-aux-Fn-F0-approximation-error}. Moreover, there exists $g_n$ such that $\mathbb E(g_n(W)-g_0(W))^2 = O\left( (\log B_n)^{-m_w/2} \right) $, following from Lemma  \ref{lemma-aux-RKHS-Sobolev-approximation-error} and Assumption \ref{assumption-cgce-rate}(\ref{aspn-bdd-L2Pw-L2leb}). 
		Thus, 
		\begin{align*}
			\|p_{\theta_n} - p_0 \|_2^2 & \leq 2 \mathbb E \big( F_n(V+g_n(W)) - F_0(V+g_n(W)) \big)^2 + 2 \mathbb E \big( F_0(V+g_n(W)) - F_0(V+g_0(W)) \big)^2  \\
			& \leq 2 \big(\|F_n -F_0\|_\infty\big)^2  + 2 \|f_0\|_\infty^2 \mathbb E(g_n(W)-g_0(W))^2 \\
			& \leq o\left( J_n^{-m_e} \right) + O\left( (\log B_n)^{-m_w/2} \right)
		\end{align*}
		as desired.

		Now we verify Condition (vi). 
		Notice Lemma \ref{lemma-An-local-centered-loss-covering-no} gives that 
		\[
		\log N(\delta,\mathcal A_n, \|\cdot\|_{\infty} ) \leq  C \left( \log \frac{B_n}{\delta} \right)^{d_{w}+1} + C J_n \log \frac{1}{\delta} 
		\]
		for a universal constant $C>0$. 
		Together with $N_{[]}(\delta, \mathcal A_n, \|\cdot\|_2) \leq N_{[]}(\delta, \mathcal A_n, \|\cdot\|_\infty) \leq N(\delta/2, \mathcal A_n, \|\cdot\|_\infty)$,  we have
		\begin{align}\label{pf-rate-fixed-pt-aux}
			& \frac{1}{\sqrt{n} \gamma_n^2} \int_{b\gamma_n^2}^{\gamma_n} \sqrt{\log N_{[]}(\delta, \mathcal A_n, \|\cdot\|_2) } d\delta  \notag \\
			& \quad \leq \frac{\sqrt{C}}{\sqrt{n} \gamma_n^2} \left( \int_{b\gamma_n^2}^{\gamma_n}  \left( \log \frac{2 B_n}{\delta} \right)^{\frac{d_{w}+1}{2}} d\delta + J_n^{1/2} \int_{b\gamma_n^2}^{\gamma_n} \left( \log  \frac{2}{\delta} \right)^{1/2}   d\delta \right) \notag \\
			& \quad \leq \frac{2\sqrt{C}}{\sqrt{n} \gamma_n} \left( \left( \log \frac{B_n}{b \gamma_n^2} \right)^{\frac{d_{w}+1}{2}} + J_n^{1/2}  \left(  \log \frac{1}{b \gamma_n^2} \right)^{1/2}  \right) 
		\end{align}
		By setting $\gamma_n = \sqrt{ \frac{(\log B_n)^{d_w+1} \vee J_n}{ n } \log n} $,  \eqref{pf-rate-fixed-pt-aux} is bounded by a constant under $\gamma_n = O(1)$ and $(\log B_n)^{d_w+1} \vee J_n \gtrsim (\log n)^{d_w}$. 
	\end{proof}

	\begin{proof}[Proof of Corollary \ref{corollary-cgce-rate}] 
		Optimizing \eqref{ccp-cgce-rate-aux} by over $B_n, J_n$ yields the stated result, which follows from some algebras by dividing two cases where $ (\log B_n)^{d_w+1} \lessgtr J_n$. 
	\end{proof}

	\section{Assumptions and Proof for Theorem \ref{theorem-asym-normality-wape}}
	\label{appendix-AN}
	
	We first introduce some notations before stating the technical conditions needed for Theorem~\ref{theorem-asym-normality-wape}. 
	For any $\theta=(g,F)\in \mathcal G\times \mathcal F$, define
	\begin{align*}
		\frac{\partial p(x,\theta_0) }{\partial \theta}[\theta-\theta_0] & :=\lim_{t\to 0} \frac{p\big(x,\theta_0+t(\theta-\theta_0) \big) - p(x,\theta_0)}{t}  
	\end{align*}
	and the definitions of $\frac{\partial \ell(z,\theta_0) }{\partial \theta}[\theta-\theta_0], \frac{\partial \gamma(\theta_0) }{\partial \theta}[\theta-\theta_0]$ are similar. 
	Moreover, define 
	\begin{align*}
		\frac{\partial^2 p(x, \tilde\theta)}{ \partial \theta \partial \theta} [u_1,u_2] := \lim_{t\to 0} \frac{1}{t} \bigg( \frac{\partial p(x, \tilde\theta+t u_2)}{ \partial \theta}[u_1] - \frac{\partial p(x, \tilde\theta)}{ \partial \theta}[u_1] \bigg)
	\end{align*}
	and the definition of $\frac{\partial^2 \ell(z, \tilde\theta)}{ \partial \theta \partial \theta} [u_1,u_2]$ is similar.  
	See Lemma~\ref{lemma-aux-pathwise-derivatives} in Appendix~\ref{appendix-tech-lmas-for-AN} for the forms of the pathwise derivatives. In particular,
	\begin{align}\label{partial-p-partial-theta}
		\frac{\partial p(x,\theta_0) }{\partial \theta}[\theta-\theta_0]  = F(v+g_0(w))-F_0(v+g_0(w)) + f_0(v+g_0(w))(g(w)-g_0(w)) 
	\end{align}
	and for $\ell(z,\theta) = (y-p(x,\theta))^2$, 
	\begin{align*}
		\frac{\partial \ell(z,\theta_0) }{\partial \theta}[\theta-\theta_0]  = -2(y-p(x,\theta_0 )) \frac{\partial p(x,\theta_0) }{\partial\theta}[\theta-\theta_0]
	\end{align*}
	Notice \eqref{gamma-theta-form-linear-in-ptheta} implies
	\begin{align}\label{partial-gamma-partial-theta}
		\frac{\partial \gamma(\theta_0) }{\partial \theta}[\theta-\theta_0] = \mathbb E b_\gamma(X) \frac{\partial p(X,\theta_0) }{\partial \theta}[\theta-\theta_0].
	\end{align}
	
	Define the \emph{Fisher norm} $\|\theta-\theta_0\|_F$ for $\theta=(g,F)$ by 
	\begin{align*}
		\|\theta - \theta_0\|_F^2  := \mathbb E \left( \frac{\partial p(X,\theta_0) }{\partial \theta}[\theta-\theta_0] \right)^2 
	\end{align*}
	and the norm is induced by the inner product $\langle u,\tilde u \rangle_F = \mathbb E \frac{\partial p(X,\theta_0) }{\partial \theta}[u] \frac{\partial p(X,\theta_0) }{\partial \theta}[\tilde u]$. 
	Then
	\begin{align*}
		\left| \frac{\partial \gamma(\theta_0) }{\partial \theta}[\theta-\theta_0] \right| = \left| \mathbb E b_\gamma(X) \frac{\partial p(X,\theta_0) }{\partial \theta}[\theta-\theta_0]  \right|  \leq \left( \mathbb E b_\gamma(X)^2 \right)^{1/2} \|\theta-\theta_0\|_F
	\end{align*}
	by the Cauchy-Schwarz inequality. 
	Thus, 
	\[
	\sup_{\theta \in \Theta: \|\theta-\theta_0\|_F >0} \frac{ \left| \frac{\partial \gamma(\theta_0) }{\partial \theta}[\theta-\theta_0] \right| }{\|\theta-\theta_0\|_F} \leq \left( \mathbb E b_\gamma(X)^2 \right)^{1/2}  <\infty
	\]
	By the Riesz representation theorem, there exists $v^\ast =(v_g^\ast, v_F^\ast) \in \bar V$ the completion under $\|\cdot\|_F$ of $\Theta-\{\theta_0\}$ such that, for any $\theta\in \Theta$, 
	\[
	\frac{\partial \gamma(\theta_0) }{\partial \theta}[\theta-\theta_0]  = \inprod{\theta-\theta_0, v^\ast}_F. 
	\]
	Define a neighborhood 
	\[
	\mathcal N_{0n} := \{ \theta\in \Theta_n : d_{\Theta}(\theta,\theta_0) = o(1), \|\theta-\theta_0\|_F = o(n^{-1/4}) \} .
	\]

	Now we introduce the following technical conditions for Theorem \ref{theorem-asym-normality-wape}, 
	which are needed to control the high-order terms in the expansions of the highly nonlinear $\hat Q(\theta)$ and $\gamma(\theta)$.

	\begin{assumption}\label{assumption-wape-asym-norm}
		Assume that
		\begin{enumerate}[(a)]
			\itemsep -0.1em 
			\item \label{aspn-FisherNorm-L2ptheta-pre-ptheta-remainder} 
			For any $\theta \in \Theta_n$ with $ d_{\Theta}(\theta,\theta_0) = o(1)$, 
			\begin{align*}
				\mathbb E \bigg( p(X,\theta)-p(X,\theta_0) - \frac{\partial p(X,\theta_0) }{\partial \theta} [\theta-\theta_0] \bigg)^2 = o \bigg( \mathbb E \Big(  \frac{\partial p(X,\theta_0) }{\partial \theta} [\theta-\theta_0] \Big)^2 \bigg) .
			\end{align*}

			\item \label{aspn-vn-vast-approx-error-quarter-rate} 
			There exists $v^\ast_n\in \Theta_n$ such that $\|v^\ast_n - v^\ast\|_F = o(n^{-1/4})$.

			
			\item 
			\label{aspn-ptheta-ptheta0-deviation}
			Uniformly in $\tilde\theta\in \mathcal N_{0n}$,  it holds that (i)
			\[
			\mathbb E \bigg(  \frac{\partial p(X,\tilde\theta) }{\partial \theta} [v_n^\ast] - \frac{\partial p(X,\theta_0) }{\partial \theta} [v_n^\ast] \bigg)^2 = o(n^{-1/2}) ,
			\] 
			and (ii) 
			\[
			\mathbb E \bigg( p(X,\tilde\theta)-p(X,\theta_0) - \frac{\partial p(X,\theta_0) }{\partial \theta} [\tilde\theta-\theta_0] \bigg) \frac{\partial p(X,\theta_0) }{\partial \theta} [v^\ast] = o(n^{-1/2}) . 
			\] 
			\item  \label{aspn-FisherNorm-gammatheta-expansion-remainder}
			Uniformly in $\tilde\theta\in \mathcal N_{0n}$,  
			\[
			\mathbb E \bigg( p(X,\tilde\theta)-p(X,\theta_0) - \frac{\partial p(X,\theta_0) }{\partial \theta} [\tilde\theta-\theta_0] \bigg) b_\gamma(X)   = o(n^{-1/2}) . 
			\]
			
			\item \label{aspn-FisherNorm-sample-partialell_ell0-remainder}
			Uniformly in $\tilde\theta\in \mathcal N_{0n}$,
			\[
			\frac{1}{\sqrt{n}} \sum_{i=1}^n \bigg( \frac{\partial \ell(Z,\tilde\theta)}{\partial \theta}[v_n^\ast] -\frac{\partial \ell(Z,\theta_0)}{\partial \theta}[v_n^\ast]  
			- \mathbb E \bigg[ \frac{\partial \ell(Z,\tilde\theta)}{\partial \theta}[v_n^\ast] -\frac{\partial \ell(Z,\theta_0)}{\partial \theta}[v_n^\ast] \bigg] \bigg) = o_p(1) .
			\] 
			\item \label{aspn-higher-order-ell-theta_theta0-approx}
			For all $\tilde\theta \in \mathcal N_{0n}$ and $t_n = o(n^{-1/2})$,  
			$\left|\frac{\partial^2 \ell(z,\tilde \theta + t_n v_n^\ast)}{\partial\theta \partial\theta}[v_n^\ast,v_n^\ast] \right| \leq c(z)$ for some function $c(z)$ such that $\mathbb E c(Z)^2 <\infty$.

			\item \label{aspn-Qhat-deviation-local-alternative}
			For some sequence $\epsilon_n = o(n^{-1/2})$, $\hat Q(\hat \theta) \leq \hat Q(\hat \theta \pm t \epsilon_n v_n^\ast) + o_p(\epsilon_n n^{-1/2})$ 
			for $t\in [0,1]$. 
			
		\end{enumerate}
	\end{assumption}

	Assumption~\ref{assumption-wape-asym-norm} includes high-level conditions to control the expansions of the highly nonlinear $p(x,\theta)$ and $\gamma(\theta)$ around $\theta_0$.  
	In particular, Assumption~\ref{assumption-wape-asym-norm}(\ref{aspn-FisherNorm-L2ptheta-pre-ptheta-remainder}) requires that the remainder term 
	\begin{align*} 
		R(X,\theta):= p(X,\theta) - p(X,\theta_0) -   \frac{\partial p(X,\theta_0) }{\partial \theta}[\theta-\theta_0]   ,
	\end{align*}
	after approximating $p(X,\theta) - p(X,\theta_0) $ using $\frac{\partial p(X,\theta_0) }{\partial \theta}[\theta-\theta_0] $,  has a smaller order in the $L_2(X)$-norm, i.e., $\mathbb E R(X,\theta)^2 = o\big( \mathbb E \big( \frac{\partial p(X,\theta_0) }{\partial \theta}[\theta-\theta_0]  \big)^2 \big) = o( \|\theta-\theta_0\|_F^2)$. 
	Thus, Assumption~\ref{assumption-wape-asym-norm}(\ref{aspn-FisherNorm-L2ptheta-pre-ptheta-remainder}) implies the equivalence between the following norms
	\begin{align}\label{eq-FisherNorm-L2ptheta}
		\|\theta-\theta_0\|_F^2 \asymp \mathbb E( p(X,\theta) - p(X,\theta_0) )^2, 
	\end{align} 
	because
	\begin{align*} 
		& \mathbb E( p(X,\theta) - p(X,\theta_0) )^2  
		= \mathbb E \Big(  \frac{\partial p(X,\theta_0) }{\partial \theta}[\theta-\theta_0] 
		+  p(X,\theta) - p(X,\theta_0) -   \frac{\partial p(X,\theta_0) }{\partial \theta}[\theta-\theta_0]  \Big)^2  \\
		& = \|\theta-\theta_0\|_F^2   
		+ 2 \mathbb E   \frac{\partial p(X,\theta_0) }{\partial \theta}[\theta-\theta_0] R(X,\theta) + \mathbb E  R(X,\theta)^2    \\
		& = \|\theta-\theta_0\|_F^2   + o( \|\theta-\theta_0\|_F^2  ) ,
	\end{align*}
	since $\big| \mathbb E   \frac{\partial p(X,\theta_0) }{\partial \theta}[\theta-\theta_0] R(X,\theta) \big| \leq \|\theta-\theta_0\|_F \sqrt{\mathbb E R(X,\theta)^2} = o(\|\theta-\theta_0\|_F^2)$.

	Analogously, 
	Assumption~\ref{assumption-wape-asym-norm}(\ref{aspn-FisherNorm-gammatheta-expansion-remainder}) ensures that the remainder term $\gamma( \tilde\theta) - \gamma(\theta_0) - \frac{\partial \gamma(\theta_0) }{\partial \theta}[\tilde\theta-\theta_0]$ after approximating $\gamma( \tilde\theta) - \gamma(\theta_0)$ by $\frac{\partial \gamma(\theta_0) }{\partial \theta}[\tilde\theta-\theta_0]$ is of a smaller order than $\frac{\partial \gamma(\theta_0) }{\partial \theta}[\tilde\theta-\theta_0]$. 
	Specifically, Assumption~\ref{assumption-wape-asym-norm}(\ref{aspn-FisherNorm-gammatheta-expansion-remainder}), together with \eqref{gamma-theta-form-linear-in-ptheta} and \eqref{partial-gamma-partial-theta}, implies that 
	\[
	\gamma( \tilde\theta) - \gamma(\theta_0) - \frac{\partial \gamma(\theta_0) }{\partial \theta}[\tilde\theta-\theta_0] =  \mathbb E b_\gamma(X) \bigg(  p(X, \tilde\theta)-p(X,\theta_0) - \frac{\partial p(X,\theta_0) }{\partial \theta} [ \tilde\theta-\theta_0] \bigg) = o(n^{-1/2})
	\]

	To provide further intuitions of these conditions, we explain what these conditions are and what they reduce to in the parametric case, and then we provide a set of more interpretable primitive conditions when $F_0$ is known.

	\begin{remark}[Interpreting Assumption \ref{assumption-wape-asym-norm} under parametric case]
		Suppose that $p(x,\theta)$ is parametric in the sense that $\theta\in \mathbb R^d$, or more precisely, $(g,F)$ are fully parametrized by a $d$-dimensional parameter $\theta$. Assume that, for any $x$, $p(x,\theta)$ is twice continuously differentiable with respect to $\theta$. 
		Let $\nabla_\theta p(x,\theta),\nabla_{\theta \theta}^2 p(x,\theta)$ denote the gradient and hessian of $p(x,\theta)$ with respect to $\theta$.  
		Assume that, in some neighborhood $\mathcal N$ of $\theta_0$, all elements of $\nabla_\theta p(x,\theta), \nabla_{\theta \theta}^2p(x,\theta)$ are bounded by a constant $C$ that does not depend on $x$. 
		Assume further that $\mathbb E \nabla_\theta p(X,\theta_0) \nabla_\theta p(X,\theta_0)'$ is positive definite. 
		
		In this case, all of the conditions in Assumption \ref{assumption-wape-asym-norm} are satisfied. More specifically, the pathwise derivative $\frac{\partial p(x,\theta_0) }{\partial \theta}[\theta-\theta_0]$ reduces to $(\theta-\theta_0)'\nabla_\theta p(X,\theta_0)$, 
		so that  Assumption~\ref{assumption-wape-asym-norm}(\ref{aspn-FisherNorm-L2ptheta-pre-ptheta-remainder}) reduces to 
		\[
		\mathbb E \Big( p(X,\theta) - p(X,\theta_0)- (\theta-\theta_0)'\nabla_\theta p(X,\theta_0) \Big)^2 
		= o\Big(  \mathbb E( (\theta-\theta_0)'\nabla_\theta p(X,\theta_0) )^2 \Big), 
		\]
		which holds since, for some $\bar\theta$ between $\theta$ and $\theta_0$,  $\mathbb E \big( p(X,\theta) - p(X,\theta_0)- (\theta-\theta_0)'\nabla_\theta p(X,\theta_0) \big)^2 = \mathbb E \big( (\theta-\theta_0)' \nabla_{\theta\theta}^2 p(X,\bar\theta) (\theta-\theta_0)\big)^2 \leq C\|\theta-\theta_0\|^4$, whereas $\mathbb E( (\theta-\theta_0)'\nabla_\theta p(X,\theta_0) )^2=  (\theta-\theta_0)' \big( \mathbb E \nabla_\theta p(X,\theta_0) \nabla_\theta p(X,\theta_0)' \big) (\theta-\theta_0) \geq c \|\theta-\theta_0\|^2$ for $c>0$ the smallest eigenvalue of  $\mathbb E \nabla_\theta p(X,\theta_0) \nabla_\theta p(X,\theta_0)'$.  
		Moreover, 
		Assumption~\ref{assumption-wape-asym-norm}(\ref{aspn-vn-vast-approx-error-quarter-rate}) is trivially satisfied, because $v^\ast$ reduces to a vector in $\mathbb R^d$ and we can choose $v^\ast_n = v^\ast$ so there is no approximation error.  
		Assumption~\ref{assumption-wape-asym-norm}(\ref{aspn-ptheta-ptheta0-deviation}.i) is satisfied, 
		since, for some $\bar\theta$ between $\tilde\theta$ and $\theta_0$, $\mathbb E \big(  \frac{\partial p(X,\tilde\theta) }{\partial \theta} [v_n^\ast] - \frac{\partial p(X,\theta_0) }{\partial \theta} [v_n^\ast] \big)^2 = \mathbb E \big( (\nabla_\theta p(X,\tilde\theta) - \nabla_\theta p(X,\theta_0)' v_n^\ast \big)^2 = \mathbb E \big( (\tilde\theta-\theta_0)' \nabla_{\theta \theta}^2 p(X,\bar\theta)  v_n^\ast \big)^2  = o(n^{-1/2})$ by $\|\tilde \theta -\theta_0\| = o_p(n^{-1/4})$. 
		Moreover, recall that 
		$\mathbb E \big( p(X,\theta) - p(X,\theta_0)- (\theta-\theta_0)'\nabla_\theta p(X,\theta_0) \big)^2  
		\leq C \|\theta-\theta_0\|^4 
		$,   
		and thus, Assumption~\ref{assumption-wape-asym-norm}(\ref{aspn-ptheta-ptheta0-deviation}.ii) and Assumption~\ref{assumption-wape-asym-norm}(\ref{aspn-FisherNorm-gammatheta-expansion-remainder}) hold by Cauchy-Schwarz inequality. 
		Since $\ell(z,\theta) = (y-p(x,\theta))^2$ with $\nabla_\theta p(x,\theta)$ continuous in $\theta$ and each element bounded by $C$, 
		Assumption~\ref{assumption-wape-asym-norm}(\ref{aspn-FisherNorm-sample-partialell_ell0-remainder}) holds by the uniform central limit theorem and $\|\tilde \theta -\theta_0\| =o(1)$.  
		Assumption~\ref{assumption-wape-asym-norm}(\ref{aspn-higher-order-ell-theta_theta0-approx}) holds since $\ell(z,\theta) = (y-p(x,\theta))^2$ and elements of $\nabla_\theta p(x,\theta), \nabla_{\theta \theta}^2p(x,\theta)$ are bounded.  
		Lastly, Condition~\ref{assumption-wape-asym-norm}(\ref{aspn-Qhat-deviation-local-alternative}) is also trivially satisfied, since $\hat Q(\hat \theta) \leq \hat Q(\theta)$ for any $\theta \in \mathbb R^d$. 
	\end{remark}

	We now give a set of more interpretable primitive conditions for Assumption \ref{assumption-wape-asym-norm} in the case where $F_0$ is known.  
	
	\begin{assumption}[Primitive conditions for Assumption \ref{assumption-wape-asym-norm} when $F_0$ is known]
		\label{assumption-wape-asym-norm-F0known}
		Assume that 
		\begin{enumerate}[(a)]
			\itemsep -0.1em 
			\item 
			There exists a constant $c$ such that $\mathbb E \big(f_0(V+g_0(W)) (g(W)- g_0(W)) \big)^2 \geq c \mathbb E (g(W)- g_0(W))^2$ for all $g$ such that $\|g-g_0\|_\infty = o(1)$; 
			\label{aspn-FisherNorm-L2ptheta-F0known}
			
			\item 
			$v^\ast_g$ is differentiable up to order $m_w$ with uniformly bounded derivatives; 
			\label{aspn-vast-g-bdd}

			\item $\mathbb E \big(  | b_\gamma(X) |  \big| W \big)$ is bounded by a constant $C$ almost everywhere; 
			\label{aspn-b_gamma_W-bdd}
			
		\end{enumerate}

	\end{assumption}
	
	Assumption~\ref{assumption-wape-asym-norm-F0known}(\ref{aspn-FisherNorm-L2ptheta-F0known}) holds if there exists a compact set $A\subset \mathbb R$ such that $\int_A f_0^2 >0$ and 
	\[
	\inf_{v\in A, w\in \mathcal W} f_{V|W=w} (v-g_0(w)) \geq \tilde c, \quad \text{for some } \tilde c>0,
	\]
	where $f_{V|W=w}$ is the Lebesgue density of $\mathcal L(V|W=w)$.\footnote{
		To see this, note that $\mathbb E \big(f_0(V+g_0(W))^2 |W=w\big) = \int f_{V|W=w} (v-g_0(w))  f_0(v)^2 dv \geq \tilde c \int_A f_0^2 $ for any $w\in\mathcal W$, and thus, $\mathbb E \big(f_0(V+g_0(W)) (g(W)- g_0(W)) \big)^2 \geq \big( \tilde c \int_A f_0^2 \big) \mathbb E (g(W)- g_0(W))^2 $, so Assumption~\ref{assumption-wape-asym-norm-F0known}(\ref{aspn-FisherNorm-L2ptheta-F0known}) holds.  
	}
	When $V$ is independent of $W$, this holds whenever $\inf_{v\in A, -M\leq u\leq M} f_V(v-u) >0$, since $g_0$ is bounded by $M$.

	\begin{proposition}\label{proposition-F0known-condns}
		Assume that $F_0$ is known and Assumption \ref{assumption-wape-asym-norm-F0known} holds. Let conditions in Corollary \ref{corollary-cgce-rate} hold with $\beta>1$.  
		Then Assumption \ref{assumption-wape-asym-norm} is satisfied with $\theta=(g,F_0)$.
	\end{proposition}
	
	\noindent
	Proposition \ref{proposition-F0known-condns} shows that Assumption \ref{assumption-wape-asym-norm} is satisfied under Assumption \ref{assumption-wape-asym-norm-F0known}. 
	The proof of Proposition \ref{proposition-F0known-condns} is provided in Section \ref{appendix-subsection-proof-propn-lowlevelcondn}.

	\begin{remark}
		In the fully nonparametric case with unknown $F_0$, verifying Assumption~\ref{assumption-wape-asym-norm} under low-level primitive conditions is more involved.  
			For instance, Assumption~\ref{assumption-wape-asym-norm}(\ref{aspn-FisherNorm-L2ptheta-pre-ptheta-remainder}) can be written as 
			\begin{align*}
				& \mathbb E \Big( F(V+g(W)) - F(V+g_0(W)) - f_0(V+g_0(W)) (g(W)-g_0(W) ) \Big)^2 \\
				& = o\Big(  \mathbb E\Big[ F(V+g_0(W)) - F_0(V+g_0(W)) + f_0(V+g_0(W)) (g(W)-g_0(W) )  \Big]^2 \Big)
			\end{align*}
			While this condition is a natural remainder requirement and may appear intuitive (e.g., in the special case $F=F_0$), 
			it is not straightforward to derive from primitive conditions when $F_0$ is unknown, because the scale on the right-hand side involves the two errors $F(V+g_0(W)) - F_0(V+g_0(W))$ and $f_0(V+g_0(W)) (g(W)-g_0(W) )$ jointly.  
		For this reason, Assumption \ref{assumption-wape-asym-norm} remains a high-level condition in Theorem \ref{theorem-asym-normality-wape}, as is common in the nonparametric literature where parameters enter the objective function and the smooth functional in a highly nonlinear way; see, e.g.,  
		\cite{ai-chen-03}; \cite{chen-pouzo-15}. 
	\end{remark}

	\subsection{Proof for Theorem \ref{theorem-asym-normality-wape}}

	\begin{proof}[Proof of Theorem \ref{theorem-asym-normality-wape}] 
		Notice $\|\hat p - p_0 \|_{L_2(X)} = o_p(n^{-1/4})$ by Corollary \ref{corollary-cgce-rate} and $\beta>1$. 
		Assumption~\ref{assumption-wape-asym-norm}(\ref{aspn-FisherNorm-L2ptheta-pre-ptheta-remainder}) further implies \eqref{eq-FisherNorm-L2ptheta}, and thus, $\|\hat \theta - \theta_0\|_F = o_p(n^{-1/4})$. 
		Let $\epsilon_n$ be a sequence satisfying $\epsilon_n = o(n^{-1/2})$ in Assumption~\ref{assumption-wape-asym-norm}(\ref{aspn-Qhat-deviation-local-alternative}). Let $v_n^\ast \in \Theta_n$ be such that $\|v_n^\ast - v^\ast\| = o(n^{-1/4})$ in Assumption~\ref{assumption-wape-asym-norm}(\ref{aspn-vn-vast-approx-error-quarter-rate}).  
		Below we denote by $u^\ast = \pm v^\ast$ to indicate that the results hold for either $v^\ast$ or $-v^\ast$. Similarly, we denote by $u^\ast_n = \pm v_n^\ast$.

		Define $\hat\theta_u^\ast := \hat\theta + \epsilon_n u_n^\ast$ as a local alternative of $\hat \theta$ for some $\epsilon_n = o(n^{-1/2})$. 
		Define $\bar \theta(t) = \hat \theta + t( \hat\theta_u^\ast - \hat\theta)$ for $t\in[0,1]$, so $\bar\theta(1)=\hat\theta_u^\ast$ and $\bar\theta(0) = \hat\theta$.  
		By Assumption~\ref{assumption-wape-asym-norm}(\ref{aspn-Qhat-deviation-local-alternative}) and a Taylor expansion of $\hat Q(\bar\theta(t))$ around $t=0$ up to second-order, we have 
		\begin{align}\label{pf-asym-normality-aux}
			o_p( \epsilon_n n^{-1/2})  & \geq \hat Q(\hat \theta) - \hat Q(\hat\theta_u^\ast) = \hat Q(\bar{\theta}(0)) - \hat Q(\bar{\theta}(1))  \notag \\
			& = - \frac{d\hat Q(\bar\theta(t) )}{ dt} \bigg|_{t=0} - \frac{1}{2}\frac{d^2\hat Q(\bar\theta(t) )}{ dt^2} \bigg|_{t=s_\ast} \quad \text{for some $s_\ast\in[0,1]$}  
		\end{align}
		Notice 
		\begin{align}\label{pf-asym-normality-aux-part1}
			& \frac{d\hat Q(\bar\theta(t) )}{ dt} \bigg|_{t=0}  := \lim_{t\to0} \frac{\hat Q(\bar\theta(t)) - \hat Q(\bar\theta(0)) }{t} =  \frac{1}{n} \sum_{i=1}^n  \frac{\partial \ell(Z_i,\hat\theta)}{\partial \theta}[\hat\theta_u^\ast-\hat\theta]  \notag \\ 
			& \quad = \epsilon_n  \frac{1}{n} \sum_{i=1}^n  \frac{\partial \ell(Z_i,\hat\theta)}{\partial \theta}[u_n^\ast]  \notag \\  
			& \quad  = \epsilon_n \bigg(  2 \langle \hat\theta-\theta_0, u^\ast \rangle_F + \frac{1}{n} \sum_{i=1}^n \bigg(  \frac{\partial \ell(Z,\theta_0)}{\partial \theta}[u^\ast]  - \mathbb E \frac{\partial \ell(Z,\theta_0)}{\partial \theta}[u^\ast]  \bigg) + o_p(n^{-1/2})  \bigg)
		\end{align}
		where the second line is by the linearity of $\frac{\partial \ell(x,\tilde\theta)}{\partial \theta}[u]$ in $u$ as in Lemma \ref{lemma-aux-pathwise-derivatives} (ii). The last line follows from Lemma \ref{lemma-aux-sample-partialell-hattheta}.  
		
		Moreover, 
		\begin{align*} 
			& \frac{d^2\hat Q(\bar\theta(t)) }{ dt^2} \bigg|_{t=s_\ast}  := \lim_{\tau\to 0} \frac{1}{\tau} \bigg( \frac{d\hat Q(\bar\theta(t) )}{ dt} \bigg|_{t=s_\ast+\tau} - \frac{d\hat Q(\bar\theta(t) )}{ dt} \bigg|_{t=s_\ast}  \bigg) \\
			& \quad = \lim_{\tau\to 0} \frac{1}{\tau}   \frac{1}{n} \sum_{i=1}^n  \bigg( \frac{\partial \ell(Z_i, \bar\theta (s_\ast+\tau) ) }{\partial \theta}[\epsilon_n u_n^\ast] - \frac{\partial \ell(Z_i, \bar\theta(s_\ast) ) }{\partial \theta}[\epsilon_n u_n^\ast] \bigg)  \\ 
			& \quad = \epsilon_n \lim_{\tau\to 0} \frac{1}{\tau}   \frac{1}{n} \sum_{i=1}^n  \bigg( \frac{\partial \ell(Z_i, \bar\theta (s_\ast+\tau) ) }{\partial \theta}[u_n^\ast] - \frac{\partial \ell(Z_i, \bar\theta(s_\ast) ) }{\partial \theta}[u_n^\ast] \bigg) \\
			& \quad = \epsilon_n \frac{1}{n} \sum_{i=1}^n  \frac{\partial^2 \ell(Z_i, \hat \theta + s_\ast \epsilon_n u_n^\ast)}{ \partial \theta \partial \theta} [u_n^\ast, \epsilon_n u_n^\ast] 
			= \epsilon_n^2 \frac{1}{n} \sum_{i=1}^n  \frac{\partial^2 \ell(Z_i, \hat \theta + s_\ast \epsilon_n u_n^\ast)}{ \partial \theta \partial \theta} [u_n^\ast, u_n^\ast] 
		\end{align*} 
		where the third line is due to the linearity of $ \frac{\partial \ell(z, \tilde\theta ) }{\partial \theta}[u] $ in $u$ as in Lemma \ref{lemma-aux-pathwise-derivatives} (ii), and the last equality follows from Lemma \ref{lemma-aux-pathwise-derivatives} (iv). 
		Thus, 
		\begin{align} \label{pf-asym-normality-aux-part2}
			\frac{d^2\hat Q(\bar\theta(t)) }{ dt^2} \bigg|_{t=s_\ast} & =\epsilon_n^2 \frac{1}{n} \sum_{i=1}^n  \frac{\partial^2 \ell(Z_i, \hat \theta + s_\ast \epsilon_n u_n^\ast)}{ \partial \theta \partial \theta} [u_n^\ast, u_n^\ast]  \notag \\
			& = O_p(\epsilon_n^2) 
		\end{align} 
		where the last line is by Assumption~\ref{assumption-wape-asym-norm}(\ref{aspn-higher-order-ell-theta_theta0-approx}). 
		
		Combining \eqref{pf-asym-normality-aux}, \eqref{pf-asym-normality-aux-part1}, \eqref{pf-asym-normality-aux-part2} and noticing $u^\ast = \pm v^\ast$ and the linearity of $\frac{\partial \ell(z,\tilde\theta)}{\partial \theta}[u]$ in $u$ yield 
		\begin{align}\label{pf-asym-normality-aux-final}
			\sqrt{n}\langle \hat\theta-\theta_0, v^\ast \rangle_F & = - \frac{1}{2} \frac{1}{\sqrt{n}} \sum_{i=1}^n \bigg(  \frac{\partial \ell(Z,\theta_0)}{\partial \theta}[v^\ast]  - \mathbb E \frac{\partial \ell(Z,\theta_0)}{\partial \theta}[v^\ast]  \bigg) + o_p(1) \notag \\
			& \to_d N\left(0,  \mathbb E p(X,\theta_0) (1-p(X,\theta_0)) b_\gamma(X)^2  \right)
		\end{align}
		since 
		\begin{align*}
			\mathbb E \frac{\partial \ell(Z,\theta_0)}{\partial \theta}[v^\ast]   = - 2 \mathbb E(Y-p(X,\theta_0)) \frac{\partial p(X,\theta_0 )}{\partial\theta} [v^\ast] = 0
		\end{align*}
		and 
		\begin{align*}
			\mathbb E & \bigg(\frac{\partial \ell(Z,\theta_0)}{\partial \theta}[v^\ast]  \bigg)^2 = 4 \mathbb E(Y-p(X,\theta_0))^2 \bigg(  \frac{\partial p(X,\theta_0 )}{\partial\theta} [v^\ast]  \bigg)^2 \\
			&  = 4 \mathbb E p(X,\theta_0) (1-p(X,\theta_0)) \bigg(  \frac{\partial p(X,\theta_0 )}{\partial\theta} [v^\ast]  \bigg)^2 
			= 4 \mathbb E p(X,\theta_0) (1-p(X,\theta_0)) b_\gamma(X)^2 
		\end{align*}
		by the definition of $v^\ast$ which gives $\frac{\partial \gamma(\theta_0) }{\partial \theta}[\theta-\theta_0]  = \inprod{\theta-\theta_0, v^\ast}_F$ for all $\theta \in \Theta$.

		Notice  
		\begin{align*}
			\gamma(\hat \theta) - \gamma(\theta_0) - \frac{\partial \gamma(\theta_0) }{\partial \theta}[\hat \theta-\theta_0] 
			= \mathbb E b_\gamma(X) \bigg(  p(X,\hat \theta)-p(X,\theta_0) - \frac{\partial p(X,\theta_0) }{\partial \theta} [\hat\theta-\theta_0] \bigg) 
			= o_p(n^{-1/2})
		\end{align*}
		by Assumption~\ref{assumption-wape-asym-norm}(\ref{aspn-FisherNorm-gammatheta-expansion-remainder}) and $\|\hat\theta-\theta_0\|_F = o_p(n^{-1/4})$.  
		Therefore, 
		\begin{align*}
			\sqrt{n}\big( \gamma(\hat\theta) - \gamma(\theta_0) \big) & = \sqrt{n} \frac{\partial \gamma(\theta_0) }{\partial \theta}[\hat \theta-\theta_0]  + \sqrt{n} \left( \gamma(\hat \theta) - \gamma(\theta_0) - \frac{\partial \gamma(\theta_0) }{\partial \theta}[\hat \theta-\theta_0]  \right) \\
			& = \sqrt{n} \frac{\partial \gamma(\theta_0) }{\partial \theta}[\hat \theta-\theta_0]  + o_p(1)  =  \sqrt{n} \langle \hat\theta-\theta_0, v^\ast\rangle_F + o_p(1)  \\
			& 
			\to_d N\left(0,  \mathbb E p(X,\theta_0) (1-p(X,\theta_0)) b_\gamma(X)^2  \right) 
		\end{align*}
		where the last line is by \eqref{pf-asym-normality-aux-final}. 
		This completes the proof for Theorem \ref{theorem-asym-normality-wape}.  
	\end{proof}

	\subsection{Proof for Proposition \ref{proposition-F0known-condns}}\label{appendix-subsection-proof-propn-lowlevelcondn}

	\begin{proof}[Proof of Proposition \ref{proposition-F0known-condns}]
		We first verify Assumption~\ref{assumption-wape-asym-norm}(\ref{aspn-FisherNorm-L2ptheta-pre-ptheta-remainder}).  
		For $\theta = (g,F_0)$, \eqref{partial-p-partial-theta} reduces to 
		\begin{align}\label{partial-p-partial-theta-F0known}
			\frac{\partial p(x,\theta_0) }{\partial \theta}[\theta-\theta_0] 
			= f_0(v+g_0(w))(g(w)-g_0(w)), 
		\end{align} 
		and thus, when $\|g-g_0\|_\infty = o(1)$, 
		\[
		\|\theta-\theta_0\|_F^2  = \mathbb E f_0(V+g_0(W))^2(g(W)-g_0(W))^2 \leq \|f_0\|_\infty^2 \mathbb E (g(W)-g_0(W))^2 =o(1), 
		\] 
		and 
		\begin{align}\label{pf-wape-F0known-g-g0-bdd-by-theta-theta0}
			\mathbb E \big( g(W) - g_0(W) \big)^2 
			\leq c^{-1} \mathbb E \big(f_0(V+g_0(W)) (g(W)- g_0(W)) \big)^2  = c^{-1}\|\theta-\theta_0\|_F^2 
		\end{align}
		by Assumption~\ref{assumption-wape-asym-norm-F0known}(\ref{aspn-FisherNorm-L2ptheta-F0known}). 
		Note that there exists $\bar g$ between $g$ and $g_0$ such that 
		\[
		p(x,\theta) - p(x,\theta_0) =  F_0(v+g(w)) - F_0(v+g_0(w)) = f_0(v+\bar g(w)) (g(w) -g_0(w)), 
		\]
		which, together with \eqref{partial-p-partial-theta-F0known}, implies that 
		\begin{align}\label{pf-wape-F0known-ptheta-ptheta0-remainder}
			&  p(x,\theta) - p(x,\theta_0) -   \frac{\partial p(x,\theta_0) }{\partial \theta}[\theta-\theta_0] 
			= \big( f_0(v+\bar g(w))  - f_0(v+g_0(w)) \big) (g(w)-g_0(w)) \notag \\
			& = f_0'(v+ \tilde g(w))(\bar g(w)-g_0(w)) (g(w)-g_0(w)), 
		\end{align}
		where the last line holds for some $\tilde g$ between $\bar g$ and $g_0$. 
		Since $\bar g$ is between $g$ and $g_0$, 
		it holds that 
		\begin{align}\label{pf-wape-F0known-condn-A-aux-1}
			& \mathbb E \Big( p(X,\theta) - p(X,\theta_0) -   \frac{\partial p(X,\theta_0) }{\partial \theta}[\theta-\theta_0]  \Big)^2  
			\leq \|f_0'\|_\infty^2 \|g -g_0\|_\infty^2 \mathbb E (g(W)-g_0(W))^2   \notag \\
			& \leq 
			c^{-1} \|f_0'\|_\infty^2 \|g -g_0\|_\infty^2 \|\theta-\theta_0\|_F^2  
			= o( \|\theta-\theta_0\|_F^2  ) 
		\end{align} 
		where the second line is by \eqref{pf-wape-F0known-g-g0-bdd-by-theta-theta0},  
		and the last equality is because $\|g-g_0\|_\infty = o(1)$. This verifies Assumption~\ref{assumption-wape-asym-norm}(\ref{aspn-FisherNorm-L2ptheta-pre-ptheta-remainder}).

		Now we verify Assumption~\ref{assumption-wape-asym-norm}(\ref{aspn-vn-vast-approx-error-quarter-rate}). 
		Note that knowing $F_0$ is equivalent to setting $\mathcal F = \{F_0\}$, under which $v^\ast_F =0$ trivially, since $(v^\ast_g, v_F^\ast)$ is in the completion of $\Theta-\{\theta_0\} $ under $\|\cdot\|_F$ and $\mathcal F - \{F_0\} = \{0\}$. 
		Hence, we can set $v^\ast_{nF} = 0$ for $v_n^\ast = (v_{nF}^\ast, v^\ast_{ng})$. Therefore,
		\begin{align*}
			& \|v^\ast_n - v^\ast\|_F^2 
			= \mathbb E  \bigg( \frac{\partial p(X,\theta_0)}{\partial\theta}[v^\ast_n - v^\ast]\bigg)^2
			= \mathbb E f_0(V+g_0(W))^2 ( v^\ast_{ng}(W)-v^\ast_g(W))^2 \\
			& \leq \|f_0\|_\infty^2 \mathbb E ( v^\ast_{ng}(W)-v^\ast_g(W))^2 = O\big( (\log B_n )^{-m_w/2} \big) = o(n^{-1/2}) , 
		\end{align*} 
		where the second line is by Lemma~\ref{lemma-aux-RKHS-Sobolev-approximation-error} and Assumption~\ref{assumption-wape-asym-norm-F0known}(\ref{aspn-vast-g-bdd}), and that Theorem \ref{theorem-cgce-rate-aux} and the choice of $B_n$ in Corollary \ref{corollary-cgce-rate} imply that $(\log B_n)^{-m_w/2}  \lesssim n^{- \frac{\beta}{1+\beta} } = o(n^{-1/2})$.

		To verify Assumption~\ref{assumption-wape-asym-norm}(\ref{aspn-ptheta-ptheta0-deviation}), 
		note that $\frac{\partial p(x,\theta_0) }{\partial \theta} [v^\ast] = v_g^\ast(w) f_0(v+g_0(w))$ and $\frac{\partial p(x,\theta) }{\partial \theta} [v^\ast_n] = v_{ng}^\ast(w) f_0(v+g(w))$, since $v_F^\ast = v_{nF}^\ast = 0$ and by Lemma~\ref{lemma-aux-pathwise-derivatives}(i). 
		Then $\big|  \frac{\partial p(x,\theta_0) }{\partial \theta} [v^\ast] \big| \leq \|f_0\|_\infty \|v_g^\ast\|_\infty$ by Assumption~\ref{assumption-wape-asym-norm-F0known}(\ref{aspn-vast-g-bdd}), 
		and thus, for any $\theta = (g,F_0)\in \Theta_n$ with $\|\theta-\theta_0\|_F = o(n^{-1/4})$, it follows from \eqref{pf-wape-F0known-ptheta-ptheta0-remainder}  that 
		\begin{align*}
			& \bigg| \mathbb E \bigg( p(X,\theta)-p(X,\theta_0) - \frac{\partial p(X,\theta_0) }{\partial \theta} [ \theta-\theta_0] \bigg) \frac{\partial p(X,\theta_0) }{\partial \theta} [v^\ast]  \bigg|  \\
			& \leq \|f_0\|_\infty \|v_g^\ast\|_\infty   \mathbb E \big| f_0'(V+ \tilde g(W))(\bar g(W)-g_0(W)) (g(W)-g_0(W))  \big| \\
			& \leq \|f_0\|_\infty \|v_g^\ast\|_\infty \|f_0'\|_\infty \mathbb E (g(W)-g_0(W) )^2 
			\leq 
			c^{-1}\|f_0\|_\infty \|v_g^\ast\|_\infty \|f_0'\|_\infty  \|\theta-\theta_0\|_F^2 
			= o(n^{-1/2}), 
		\end{align*}
		where the last line is because  $\bar g$ is between $g$ and $g_0$ and \eqref{pf-wape-F0known-g-g0-bdd-by-theta-theta0}. Hence, Assumption~\ref{assumption-wape-asym-norm}(\ref{aspn-ptheta-ptheta0-deviation}.ii) is satisfied. 
		Moreover, Assumption~\ref{assumption-wape-asym-norm}(\ref{aspn-ptheta-ptheta0-deviation}.i) also holds, because 
		\begin{align}\label{pf-wape-F0known-ptheta-theta0-partial-vnast-diff}
			& \mathbb E \bigg(  \frac{\partial p(X,\tilde\theta) }{\partial \theta} [v_n^\ast] - \frac{\partial p(X,\theta_0) }{\partial \theta} [v_n^\ast] \bigg)^2 
			= \mathbb E \Big(  
			v^\ast_{ng}(W) \big( f_0(W+\tilde g(W)) -  f_0(W+g_0(W)) \big)  \Big)^2 \notag  \\
			& \leq   
			\| v^\ast_{ng}\|_\infty^2 \|f_0'\|_\infty^2 \mathbb E \big( \tilde g(W) - g_0(W) \big)^2 
			\leq \| v^\ast_{ng}\|_\infty^2 \|f_0'\|_\infty^2 c^{-1} \|\tilde\theta-\theta_0\|_F^2   
			= o(n^{-1/2}) ,
		\end{align}
		where the last line uses \eqref{pf-wape-F0known-g-g0-bdd-by-theta-theta0} and the definition of $\mathcal N_{0n}$. 
		This completes the verification of Assumption~\ref{assumption-wape-asym-norm}(\ref{aspn-ptheta-ptheta0-deviation}).

		Assumption~\ref{assumption-wape-asym-norm}(\ref{aspn-FisherNorm-gammatheta-expansion-remainder}) is satisfied, since \eqref{pf-wape-F0known-ptheta-ptheta0-remainder} and Assumption \ref{assumption-wape-asym-norm-F0known}(\ref{aspn-b_gamma_W-bdd}) imply that, for any $\theta = (g,F_0)\in \Theta_n$ with $\|\theta-\theta_0\|_F = o(n^{-1/4})$,
		\begin{align*}
			& \bigg| \mathbb E \bigg( p(X,\theta)-p(X,\theta_0) - \frac{\partial p(X,\theta_0) }{\partial \theta} [ \theta-\theta_0] \bigg) b_\gamma(X)  \bigg|  \\
			& = \big|  \mathbb E f_0'(V+ \tilde g(W))(\bar g(W)-g_0(W)) (g(W)-g_0(W)) b_\gamma(X) \big| \\
			& \leq  \|f_0'\|_\infty  \mathbb E \Big( \big|  \bar g(W)-g_0(W) \big| \big| g(W)-g_0(W)  \big| \mathbb E \big( \big| b_\gamma(X) \big| |W) \big) \Big)  \\
			& \leq C \|f_0'\|_\infty \mathbb E (g(W)-g_0(W) )^2 
			\leq  
			c^{-1}C \|f_0'\|_\infty  \|\theta-\theta_0\|_F^2 
			= o(n^{-1/2}), 
		\end{align*}
		where the last line is by \eqref{pf-wape-F0known-g-g0-bdd-by-theta-theta0}.

		Now we verify Assumption \ref{assumption-wape-asym-norm}(\ref{aspn-FisherNorm-sample-partialell_ell0-remainder}). 
		Write 
		\begin{align*}
			& \frac{\partial \ell(z,\theta)}{\partial \theta}[v_n^\ast] -\frac{\partial \ell(z,\theta_0)}{\partial \theta}[v_n^\ast]  
			= -2  (y-p(x,\theta)) \frac{\partial p(x,\theta)}{\partial \theta}[v_n^\ast] +2 (y-p(x,\theta_0)) \frac{\partial p(x,\theta_0)}{\partial \theta}[v_n^\ast]  \\
			& = 2 \big( p(x,\theta) - p(x,\theta_0) \big) \frac{\partial p(x,\theta_0)}{\partial \theta} [v_n^\ast] - 2  \big( y - p(x,\theta) \big) \bigg(  \frac{\partial p(x,\theta)}{\partial \theta}[v_n^\ast] -  \frac{\partial p(x,\theta_0)}{\partial \theta}[v_n^\ast]\bigg) \\
			& =: 2 h_1(x,\theta) - 2 h_2(x,\theta) . 
		\end{align*}
		Let $\mathcal H_{1n},\mathcal H_{2n}$ consist of functions $h_1(\cdot,\theta), h_2(\cdot,\theta)$, respectively, where $\theta$ are such that $\|\theta - \theta_0\|_F^2  = \mathbb E f_0(V+g_0(W))^2(g(W)-g_0(W))^2  = o(n^{-1/2})$ and $\|g-g_0\|_\infty = o(1)$. 
		Denote $\mathbb P_n h := \frac{1}{n}\sum_{i=1}^n h(X_i)$ and  $P h := \mathbb E h(X)$. 
		To verify Assumption \ref{assumption-wape-asym-norm}(\ref{aspn-FisherNorm-sample-partialell_ell0-remainder}), it suffices to show that 
		\begin{align} 
			& \mathbb E \sqrt{n} \sup_{ h \in \mathcal H_{1n}} |(\mathbb P_n - P) h | = o(1)  \label{pf-wape-F0known-stoch-equicts-main-1} \\
			& \mathbb E \sqrt{n} \sup_{ h \in \mathcal H_{2n}} |(\mathbb P_n - P) h | = o(1). \label{pf-wape-F0known-stoch-equicts-main-2}
		\end{align} 
		We first focus on \eqref{pf-wape-F0known-stoch-equicts-main-1}. 
		Note that $\big|\frac{\partial p(x,\theta_0)}{\partial \theta} [v_n^\ast]\big| = \big|v_{ng}^\ast(w) f_0(v+g_0(w))\big| \leq C_1$ for a constant $C_1$ since  $ v_{ng}^\ast$ is uniformly bounded and $\|f_0\|_\infty <\infty$. Hence, for any function $h(\cdot,\theta)$ in $\mathcal H_{1n}$, $\mathbb E h(X,\theta)^2 \leq C^2 \mathbb E \big( p(X,\theta) - p(X,\theta_0) \big)^2 \leq \sigma_n$  
		for some $\sigma_n = o(n^{-1/4})$, since $\|\theta - \theta_0\|_F = o(n^{-1/4})$ and \eqref{eq-FisherNorm-L2ptheta} holds as Assumption \ref{assumption-wape-asym-norm}(\ref{aspn-FisherNorm-L2ptheta-pre-ptheta-remainder}) is verified. 
		Moreover, $\|h\|_\infty \leq C_1$ since $p(\cdot,\theta),p(\cdot,\theta_0)\in [0,1]$. 
		Thus, applying Lemma 19.36 in \cite{vandervaart-00} yields
		\begin{align}\label{pf-wape-F0known-stoch-equicts-aux-1}
			\mathbb E \sqrt{n} \sup_{ h \in \mathcal H_{1n}} |(\mathbb P_n - P) h |
			\lesssim J_{[]}(\sigma_n, \mathcal H_{1n}, L_2(P)) \bigg( 1+\frac{J_{[]}(\sigma_n, \mathcal H_{1n}, L_2(P)) }{\sigma_n^2 \sqrt{n} } C_1 \bigg) , 
		\end{align}
		where
		\begin{align*}
			J_{[]}(\sigma_n, \mathcal H_{1n}, L_2(P)) 
			:=\int_0^{\sigma_n} \sqrt{ \log N_{[]} (\epsilon, \mathcal H_{1n}, L_2(P) )} d\epsilon 
		\end{align*}
		For any $\theta_1=(g_1,F_0), \theta_2=(g_2,F_0)$, $ \big| p(x,\theta_1) - p(x,\theta_2) \big| = | F_0(v+g_1(w)) -  F_0(v+g_2(w)) | \leq \|f_0\|_\infty |g_1(w)-g_2(w)|$, and thus, 
		$|h(x,\theta_1) - h(x,\theta_2)|\leq  \big| p(x,\theta_1) - p(x,\theta_2) \big| \big| \frac{\partial p(x,\theta_0)}{\partial \theta} [v_n^\ast] \big| \leq C_1 \|f_0\|_\infty |g_1(w)-g_2(w)|$. Applying Theorem 2.7.11 in \cite{vandervaart-wellner-96} yields that
		\begin{align*}
			&\log N_{[]} (\epsilon, \mathcal H_{1n}, L_2(P) 
			\leq \log N( \epsilon/C_2,  \mathcal G_n, \|\cdot\|_\infty ) \\
			& \leq C_2 \bigg( \log B_n + \log \frac{C_3}{\epsilon} \bigg)^{d_w+1} \leq C_4 (\log B_n)^{d_w+1} + C_4 \Big( \log \frac{C_3}{\epsilon} \Big)^{d_w+1}
		\end{align*}
		for some universal constants $C_2,C_3,C_4$, 
		where the second inequality is by \eqref{covering-no-Gn}. 
		Thus, 
		\begin{align*}
			J_{[]}(\sigma_n, \mathcal H_{1n}, L_2(P)) 
			\lesssim   (\log B_n)^{\frac{d_w+1}{2}} \sigma_n + \big( \log (1/\sigma_n) \big)^{ \frac{d_w+1}{2} } \sigma_n \lesssim \sigma_n^2 \sqrt{n}
		\end{align*}
		since Theorem \ref{theorem-cgce-rate-aux} and the choice of $B_n$ in Corollary \ref{corollary-cgce-rate} imply that $(\log B_n)^{\frac{d_w+1}{2}} /\sqrt{n} \lesssim n^{- \frac{\beta}{2(1+\beta)} } = o(n^{-1/4})$. 
		Thus, \eqref{pf-wape-F0known-stoch-equicts-aux-1} implies that 
		\begin{align*}
			\mathbb E \sqrt{n} \sup_{ h \in \mathcal H_{1n}} |(\mathbb P_n - P) h |
			\lesssim \sigma_n^2 \sqrt{n} = o(1)
		\end{align*}
		since $\sigma_n = o(n^{-1/4})$. This verifies \eqref{pf-wape-F0known-stoch-equicts-main-1}.  
		The verification of \eqref{pf-wape-F0known-stoch-equicts-main-2} follows from the same argument as that for \eqref{pf-wape-F0known-stoch-equicts-main-1}. Note that \eqref{pf-wape-F0known-ptheta-theta0-partial-vnast-diff} implies that 
		for any function $h(\cdot,\theta)$ in $\mathcal H_{2n}$, $\mathbb E h(X,\theta)^2 = \mathbb E \big( Y - p(X,\theta) \big)^2 \big(  \frac{\partial p(x,\theta)}{\partial \theta}[v_n^\ast] -  \frac{\partial p(x,\theta_0)}{\partial \theta}[v_n^\ast]\big)^2 \leq \sigma_n$ for some $\sigma_n = o(n^{-1/4})$. 
		Moreover, $\frac{\partial p(x,\theta)}{\partial \theta}[v_n^\ast] -  \frac{\partial p(x,\theta_0)}{\partial \theta}[v_n^\ast] = v_{ng}^\ast(w) f_0(v+g(w)) -v_{ng}^\ast(w) f_0(v+g_0(w)) $, and thus, for any $h(\cdot,\theta_1), h(\cdot,\theta_2)\in \mathcal H_{2n}$, 
		\begin{align*}
			& |h(x,\theta_1) - h(x,\theta_2)| \\
			& \leq  \Big| ( y - p(x,\theta_1) - ( y - p(x,\theta_2) ) )   v_{ng}^\ast(w)f_0(v+g_1(w) ) \Big|  \\
			& \quad + \Big|  ( y - p(x,\theta_2) ) \big(   v_{ng}^\ast(w)f_0(v+g_1(w) ) - v_{ng}^\ast(w)f_0(v+g_2(w) ) \big) \Big|  \\
			& \quad + \Big| ( y - p(x,\theta_2) - ( y - p(x,\theta_1) ) )   v_{ng}^\ast(w)f_0(v+g_0(w)  \Big|  \\
			& \leq 
			2 \|v_{ng}^\ast\|_\infty \|f_0\|_\infty | p(x,\theta_2) -  p(x,\theta_1) | +  \|v_{ng}^\ast\|_\infty \|f_0'\|_\infty  |g_1(w) - g_2(w)| \lesssim |g_1(w)-g_2(w)|
		\end{align*}
		since $\|f_0\|_\infty, \|f_0'\|_\infty<\infty$, and $v_{ng}^\ast$ is bounded. 
		Therefore, \eqref{pf-wape-F0known-stoch-equicts-main-2} can be verified using exactly the same argument as that for \eqref{pf-wape-F0known-stoch-equicts-main-1}.

		Now we verify Assumption \ref{assumption-wape-asym-norm}(\ref{aspn-higher-order-ell-theta_theta0-approx}). 
		By Lemma \ref{lemma-aux-pathwise-derivatives} (iii) and that $v_{nF}^\ast = 0$,  $\frac{\partial^2 p(x,\theta)}{\partial\theta\partial\theta} [v_n^\ast,v_n^\ast] 
		= \big(v_{ng}^\ast(w) \big)^2 F_0''(v+g(w))$. 
		Together with $\frac{\partial p(x,\theta) }{\partial \theta} [v^\ast_n] = v_{ng}^\ast(w) f_0(v+g(w))$, Lemma \ref{lemma-aux-pathwise-derivatives} (iii) gives
		\begin{align*}
			\frac{\partial^2 \ell(z,\theta)}{\partial\theta\partial\theta} [v_n^\ast,v_n^\ast] 
			= -2 \Big( \big(y- F_0(v+g(w)) \big) \big(v_{ng}^\ast(w) \big)^2 f_0'(v+g(w)) -  \big(v_{ng}^\ast(w) \big)^2 f_0(v+g(w))^2 \Big)
		\end{align*} 
		and thus, $\big|\frac{\partial^2 \ell(z,\theta)}{\partial\theta\partial\theta} [v_n^\ast,v_n^\ast]\big| \leq 2 \| v_{ng}^\ast\|_\infty^2 (\|f_0'\|_\infty + \|f_0\|_\infty^2) )$. Assumption \ref{assumption-wape-asym-norm}(\ref{aspn-higher-order-ell-theta_theta0-approx}) is thus satisfied. 
		
		For Assumption \ref{assumption-wape-asym-norm}(\ref{aspn-Qhat-deviation-local-alternative}), note that $\hat F = F_0$ since $F_0$ is known. Recall $v_{nF}^\ast = 0$ and $\mathcal F = \mathcal F_n = \{F_0\}$, so $\hat \theta \pm t\epsilon_n v_n^\ast = (\hat g\pm t\epsilon_n v_{ng}^\ast, F_0) \in \Theta_n =\mathcal G_n \times \mathcal F_n$. 
		By the definition of $\hat\theta$ which minimizes $\hat Q(\theta)$ over $\Theta_n$,  Assumption \ref{assumption-wape-asym-norm}(\ref{aspn-Qhat-deviation-local-alternative}) is satisfied trivially. 
	\end{proof}

	\section{Technical Lemmas}
	\label{appendix-tech-lmas}

	\subsection{Technical Lemmas for Theorem \ref{theorem-consistency}}\label{appendix-tech-lmas-consistency}
	
	The following lemma collects some properties that will be used repeatedly.  
	Most of them are immediate results from \cite{gallant-nychka-87}. 
	
	\begin{lemma}\label{lemma-aux-Fset-prelim-properties}
		Let $\mathcal F, \mathcal F_n$ be given in Assumption \ref{assumption-consistency}. 
		\begin{enumerate}[(a)] \itemsep -0.1em
			\item $\mathcal F$ is compact under the uniform norm $\|\cdot\|_\infty$ on $\mathbb R$. 
			\item There exists a constant $M_{\mathcal F}>0$ such that any probability densities $f$ of $F\in \mathcal F$ satisfies $\norm{f}_{\infty} <M_{\mathcal F}$. 
			\item There exists $F_n \in \mathcal F_n$ such that $\|F_n-F_0\|_\infty \to 0$ as $n\to \infty$. 
			\item For any $F, \tilde F \in \mathcal F$ and any functions $g, \tilde g: \mathcal W \to \mathbb R$, 
			\[
			\sup_{x \in \mathcal X} \big| F(v+g(w))  - \tilde F(v+\tilde g(w) ) \big|  \leq \max\{1,M_{\mathcal F} \}  d_{\Theta}(\theta,\tilde \theta)
			\]
			where $d_{\Theta}(\theta, \tilde \theta) :=  \sup_{w\in \mathcal W} | g(w) - \tilde g(w) | + \sup_{u\in \mathbb R} |F(u) - \tilde F(u) |$. 
		\end{enumerate}
	\end{lemma}

	\begin{proof}[Proof of Lemma \ref{lemma-aux-Fset-prelim-properties}]
		Let $\eta \in (1/2, \eta_0)$, and define the weighted Sobolev norm
		\[
		\|f\|_{m_e,\infty,\eta} := \max_{ 0\leq \lambda \leq m_e } \sup_{u \in \mathbb{R} } \left| f^{(\lambda)} (u) \right| (1+u^2)^{\eta}. 
		\] 
		We first notice that, for any distribution function $F$ with Lebesgue density $f$, 
		\begin{align}\label{ineq-density-norm-GN87-UnifNorm-cdf}
			\sup_{z\in \mathbb R} |F(z)| \leq \sup_{z \in \mathbb{R} } \left| \int_{-\infty}^z f(u) du \right| &  \leq \int_{-\infty}^\infty | f(u) |  du  = \int_{\mathbb{R}} | f(u) | (1+u^2)^\eta (1+u^2)^{-\eta} du \notag \\
			& \leq \|f\|_{m_e,\infty,\eta} \int_{\mathbb{R}} (1+u^2)^{-\eta} du 
		\end{align}
		where the last line is due to $\eta>1/2$ and thus $\int_{\mathbb{R}} (1+u^2)^{-\eta} du <\infty$.

		For Part (a), note that the set of densities defining $\mathcal F$ is compact under norm $\|\cdot\|_{m_e,\infty,\eta}$. This follows from Theorem 1 in \cite{gallant-nychka-87}, which shows the precompactness, and Lemma A.1 in \cite{santos-12}, which further implies the closedness and thus compactness. Although \cite{gallant-nychka-87} imposes a zero mean condition, which we do not, the proof of their Theorem 1 holds without the zero mean condition. 
		The compactness of $\mathcal F$ under $\|\cdot\|_\infty$ follows immediately from \eqref{ineq-density-norm-GN87-UnifNorm-cdf}. 
		
		Part (b) is satisfied trivially, since the compactness and thus boundness under $\|\cdot\|_{m_e,\infty,\eta}$ of the set of densities defining $\mathcal F$ implies further the boundedness under $\|\cdot\|_\infty$. 
		
		Note that the proof of Theorem 2 in \cite{gallant-nychka-87} shows that the set of densities defining $\mathcal F_n$ becomes dense in the set of densities defining $\mathcal F$ under norm $\|\cdot\|_{m_e,\infty,\eta}$.  
		Thus, Part (c) follows from \eqref{ineq-density-norm-GN87-UnifNorm-cdf}.

		Now we show part (d). 
		Denote by $\tilde f$ the density of $\tilde F$. Notice
		\begin{align*}
			& \sup_{x \in \mathcal X} \big| F(v+g(w))  - \tilde F(v+\tilde g(w) ) \big| \\
			& \quad \leq  \sup_{x \in \mathcal X} \left| F(v+ g(w)) - \tilde F(v+ g(w))  \right|  + \sup_{x\in\mathcal{X} } \left| \tilde F(v+ g(w)) - \tilde F(v+\tilde g(w)) \right|  \\
			& \quad  \leq  \sup_{u\in \mathbb R} | F(u) - \tilde F(u) |  +  \|\tilde f\|_{\infty}  \sup_{w\in\mathcal{W} } \left|   g(w) - \tilde g(w) \right|  
			\leq  \max\{1,M_{\mathcal F} \} d_{\Theta}(\theta,\tilde \theta)
		\end{align*} 
		where the last inequality is due to $\|\tilde f\|_\infty <M_{\mathcal F}$ by Part (a). 
	\end{proof}

	\subsection{Technical Lemmas for Corollary \ref{corollary-cgce-rate}}
	
	\begin{lemma}\label{lemma-aux-p-p0-interpolation-ineq}
		Let $\mathcal G, \mathcal F$ be given in Assumption \ref{assumption-consistency} (b) and (e). Let Assumption \ref{assumption-cgce-rate}(\ref{aspn-bdd-L2leb-L2Px}) hold.   
		Then there exists a constant $C>0$, which does not depend on $\theta$, such that
		\begin{align*}
			\sup_{x\in \mathcal X} |p_{\theta}(x) - p_0(x)| \leq   C \left(\norm{p_{\theta} - p_0}_2 \right)^{ 2/(2+d_x) } 
		\end{align*}
		for any $\theta = (g,F) \in \mathcal G \times \mathcal F$, where $p_{\theta}(x) = F(v+g(w))$, $p_0 = p_{\theta_0}$, and $\norm{p_{\theta}-p_0}_2 := \left( \mathbb E (p_{\theta}(X) - p_0(X))^2 \right)^{1/2}$. 
	\end{lemma}

	\begin{proof}[Proof of Lemma \ref{lemma-aux-p-p0-interpolation-ineq}]
		The proof is based on \cite{gabushin-67}. 
		We give a complete proof here since the existing results cannot be applied here, as $p_{\theta}$ may not necessarily be Lebesgue integrable: $p_{\theta}(x) =  F(v+g(w))$ is close to one if $v$ is large when $\mathcal V$ is unbounded.

		Let $\theta = (g,F) \in \mathcal G \times \mathcal F$ be fixed arbitrarily. Denote by $h_\theta = p_{\theta} - p_0$. 
		Let $x\in \mathcal X$ be fixed arbitrarily, and denote by $x_i$ the $i$-th component of $x$ for $i=1,\dots,d_{x}$. Let $\delta>0$ be fixed arbitrarily. Denote by the $d_x$-dimensional cube $\mathcal C := \prod_{i=1}^{d_x} [x_i -\delta/2, x_i +\delta/2]$. 
		Then there exists $x_\ast\in \mathcal X \cap \mathcal C$ such that  
		\[ 
		|h_\theta(x_\ast) |  = \min_{ \tilde x\in \mathcal X \cap \mathcal C} |h_\theta(\tilde x) | 
		\]
		since $h_\theta$ is a continuous function, and $\mathcal X$ is closed and $\mathcal C$ is compact. 
		Observe that
		\[
		\frac{\partial}{\partial x} h_\theta(x) = f(v+g(w)) \begin{pmatrix}
			1 \\ \frac{\partial}{\partial w} g(w)
		\end{pmatrix} -  f_0(v+g_0(w)) \begin{pmatrix}
			1 \\ \frac{\partial}{\partial w} g_0(w)
		\end{pmatrix}
		\]
		Notice by Lemma \ref{lemma-aux-Fset-prelim-properties} (b), there exists a constant $M_{\mathcal F}>0$ such that all probability densities $f$ of $F\in \mathcal F$ satisfies $\norm{f}_{\infty} <M_{\mathcal F}$. Moreover, Assumption \ref{assumption-consistency} (b)  
		implies that all first order derivatives of $g\in\mathcal G$ are bounded by $M$. 
		Consequently,   
		\begin{align*} 
			\sup_{x\in\mathcal X} \norm{\frac{\partial}{\partial x} h_\theta(x) } \leq C_1
		\end{align*}
		for some constant $C_1$ that does not depend on $g,F$.  
		Thus, 
		\begin{align}\label{pf-p-theta-p0-norm-interpolation-ineq-aux-1}
			\left| h_\theta(x) - h_\theta(x_\ast) \right|  \leq C_1 \norm{x-x_\ast}  \leq C_1 \sqrt{d_x} \delta/2
		\end{align}

		Moreover, 
		\begin{align}\label{pf-p-theta-p0-norm-interpolation-ineq-aux-2}
			\mathbb E \left(p_{\theta}(X) - p_0(X) \right)^2 &  \geq \mathbb E 1\{ X \in \mathcal C\} \left(p_{\theta}(X) - p_0(X) \right)^2   \geq h_\theta(x_\ast)^2 \mathbb E 1\{ X \in \mathcal C\}  \notag \\
			& \geq h_\theta(x_\ast)^2 M_{1,op}^{-2} \int_{\mathcal X}  1\{ x \in \mathcal C\} dx 
			\geq h_{\theta}(x_\ast)^2 M_{1,op}^{-2} \delta^{d_x} 
		\end{align}
		by the definition of $x_\ast$, Assumption \ref{assumption-cgce-rate}(\ref{aspn-bdd-L2leb-L2Px}), and that the Lebesgue measure of $\mathcal C$ is $\delta^{d_x}$. 
		Combining \eqref{pf-p-theta-p0-norm-interpolation-ineq-aux-1} and \eqref{pf-p-theta-p0-norm-interpolation-ineq-aux-2} yields
		\begin{align*}
			\left| h_{\theta}(x) \right| & \leq \left|h_{\theta}(x_\ast) \right| + \left| h_{\theta}(x) - h_{\theta}(x_\ast)  \right| 
			\leq M_{1,op} \delta^{-d_x/2} \norm{p_{\theta}-p_0}_2 + C_1 \sqrt{d_x} \delta/2 \\
			& \leq C/2 \left( \delta^{-d_x/2} \norm{p_{\theta} - p_0}_2 + \delta \right)
		\end{align*}
		for $C :=  2 \max\{M_{1,op}, C_1 \sqrt{d_x}/2\}$. 
		Choosing $\delta = \left(\norm{p_{\theta} - p_0}_2 \right)^{ 2/(2+d_x) } $ gives 
		\[
		\left| h_{\theta}(x) \right| = \left| p_{\theta}(x) - p_0(x) \right| \leq C \left(\norm{p_{\theta} - p_0}_2 \right)^{ 2/(2+d_x) } 
		\]
		Recall that $x\in \mathcal X$ was picked arbitrarily. Notice that $C$ does not depend on $x$ or $\theta$. 
		This completes the proof. 
	\end{proof}

	\begin{lemma}\label{lemma-aux-Fn-F0-approximation-error}
		Let $\mathcal F, \mathcal F_n$ be given in Assumption \ref{assumption-consistency} (e)-(f). Let Assumption \ref{assumption-cgce-rate}(\ref{aspn-density-tail}) hold.  
		Then there exists $F_n \in \mathcal F_n$ such that
		\[
		\sup_{u\in \mathbb R} \left|F_n(u)-F_0(u) \right| = o\!\left( J_n^{-m_e/2} \right)
		\]
	\end{lemma}
	
	\begin{proof}[Proof of Lemma \ref{lemma-aux-Fn-F0-approximation-error}] 
		Following \cite{fenton-gallant-96}, it is equivalent to rewrite the densities defining $\mathcal F_n$ as 
		\begin{align}\label{pf-Fn-density-hermite} 
			f(u;\tau) = \left( \sum_{j=0}^{J_n} \tau_j \overline{H}_{e_j}(u) \right)^2 e^{-u^2/2}, 
			\quad \sum_{j=0}^{J_n} \tau_{j}^2 = 1  
		\end{align}
		where $\overline{H}_{e_j}(u)$ are the same as defined in \citet[][p. 720]{fenton-gallant-96} for $j=0,1,\dots$. 
		In particular, $\{\overline{H}_{e_j}\}_{j=0}^\infty$ is a set of orthonormal basis functions for the space $\big\{ h:\mathbb R \to \mathbb R \big|  \int h(u)^2 e^{-u^2/2} du < \infty \big\}$ endowed with inner product $\langle h_1,h_2 \rangle = \int h_1(u)h_2(u) e^{-u^2/2} du$ and the induced norm. Notice the constraint 
		$\sum_{j=0}^{J_n} \tau_{j}^2 = 1$ ensures $f(u;\tau)$ is a proper density function.

		Under Assumption \ref{assumption-cgce-rate}(\ref{aspn-density-tail}), 
		$f_0(u) = \left( \sum_{j=0}^\infty\tau_{j0} \overline{H}_{e_j}(u) \right)^2 $ for $\tau_{j0} = \langle h_0, \overline{H}_{e_j} \rangle$.  
		Notice $\sum_{j=0}^\infty \tau_{j0}^2 = \int h_0(u)^2 e^{-u^2/2} du = \int f_0(u)  du = 1$.  
		Define the truncated vector
		$\tau^{(n)} = ( \tau^{(n)}_0, \tau^{(n)}_1,\dots, \tau^{(n)}_{J_n} )' \in \mathbb R^{1+J_n}$ 
		by $\tau^{(n)}_j = \frac{1 }{  \sqrt{ \sum_{i=0}^{J_n} \tau_{i0}^2} } \tau_{j0}$. Notice $\sum_{j=0}^{J_n} (\tau^{(n)}_j)^2 =1$ by construction of $\tau^{(n)}_j$'s. Thus, defining $F_n$ as the cdf of $f_n(u)  :=  f(u;\tau^{(n)} ) $, we have $F_n \in \mathcal F_n$ for each $n\in \mathbb N$.

		Denote by 
		\[
		a_n(u) := \sum_{j=0}^{J_n} \tau^{(n)}_j \overline{H}_{e_j}(u), \quad a(u) := \sum_{j=0}^{\infty} \tau_{j0} \overline{H}_{e_j}(u). 
		\]
		Then 
		\begin{align}\label{pf-Fn-F0-approximation-error-aux-1}
			\sup_{u\in \mathbb R} & \left|F_n(u)-F_0(u) \right| \leq \int |f_n(u)-f_0(u) | du  = \int |a_n(u)^2 - a(u)^2| e^{-u^2/2} du  \notag \\
			& = \int |a_n(u)+a(u)| e^{-u^2/4} |a_n(u)-a(u)| e^{-u^2/4} du \notag \\
			& \leq \left( \int (a_n(u)+a(u))^2 e^{-u^2/2} du \right)^{1/2} \left( \int (a_n(u)-a(u))^2 e^{-u^2/2} du\right)^{1/2}  \notag \\
			& \leq 4 \left( \sum_{j=J_n+1}^\infty \tau_{j0}^2 \right)^{1/2} 
		\end{align}
		where the last line is due to 
		\begin{align*}
			\int (a_n(u)+a(u))^2 e^{-u^2/2} du & \leq 2 \int a_n(u)^2  e^{-u^2/2} du  + 2\int a(u)^2  e^{-u^2/2} du   \\
			& = 2 \sum_{j=0}^{J_n} (\tau^{(n)}_j)^2 + 2 \sum_{j=0}^{\infty}  \tau_{j0}^2  
			= 4 
		\end{align*}
		and 
		\begin{align}\label{pf-Fn-F0-approximation-error-aux-2}
			\int & (a(u)-a_n(u))^2 e^{-u^2/2} du  = \int \left( \sum_{j=J_n+1}^\infty \tau_{j0}  \overline{H}_{e_j}(u) + \sum_{j=0}^{J_n} ( \tau_{j0} - \tau_{j}^{(n)} ) \overline{H}_{e_j}(u)  \right)^2 e^{-u^2/2} du  \notag \\ 
			&  \leq 2 \int \left( \sum_{j=J_n+1}^\infty \tau_{j0}  \overline{H}_{e_j}(u) \right)^2 e^{-u^2/2} du + 2 \int \left( \sum_{j=0}^{J_n} ( \tau_{j0} - \tau_{j}^{(n)} ) \overline{H}_{e_j}(u)  \right)^2 e^{-u^2/2} du  \notag \\
			& = 2 \sum_{j=J_n+1}^\infty \tau_{j0} ^2 + 2 \sum_{j=0}^{J_n}  \left( \tau_{j0} - \tau_{j}^{(n)} \right)^2  
			\leq 4 \sum_{j=J_n+1}^\infty \tau_{j0}^2.
		\end{align}
		Here the last inequality in \eqref{pf-Fn-F0-approximation-error-aux-2} follows from 
		$\sum_{j=0}^{J_n} \left( \tau_{j}^{(n)} \right)^2 = 1 = \sum_{j=0}^\infty \tau_{j0} ^2 $, 
		and thus, 
		\begin{align*}
			\sum_{j=0}^{J_n} &  \left( \tau_{j0} - \tau_{j}^{(n)} \right)^2  = \sum_{j=0}^{J_n}   \tau_{j0} ^2 +   \sum_{j=0}^{J_n}   \left( \tau_{j}^{(n)} \right)^2 - 2\sum_{j=0}^{J_n} \tau_{j0}  \tau_{j}^{(n)} \\
			&  
			= \sum_{j=0}^{J_n}  \tau_{j0} ^2 + \sum_{j=0}^\infty \tau_{j0}^2 - 2 \sqrt{\sum_{j=0}^{J_n} \tau_{j0} ^2} 
			\leq \sum_{j=0}^{J_n}  \tau_{j0} ^2 + \sum_{j=0}^\infty \tau_{j0} ^2  - 2  \sum_{j=0}^{J_n} \tau_{j0} ^2 = \sum_{j=J_n+1}^\infty  \tau_{j0} ^2 , 
		\end{align*}
		where the inequality is because $\sum_{j=0}^{J_n} \tau_{j0}^2 \leq \sum_{j=0}^\infty \tau_{j0}^2 = 1$ so its square root is not smaller than itself. 
		The stated result follows from \eqref{pf-Fn-F0-approximation-error-aux-1} and Lemma 1 in \cite{fenton-gallant-96} which gives that $\sum_{j=J_n+1}^\infty \tau_{j0}^2 = o\!\left( J_n^{-m_e} \right)$.  
	\end{proof}

	We use the following approximation result on the error of approximating a function in Sobolev space by functions in Gaussian RKHS balls. The result was first shown in \cite{smale-zhou-03} and later reorganized in \cite{zhou-13}. 
	
	\begin{lemma}\label{lemma-aux-RKHS-Sobolev-approximation-error}
		Let $h: \mathbb R^d \to \mathbb R$ be an $m$-times differentiable function, where all its derivatives up to order $m$ are square integrable. 
		Let $k(s,t) = \exp\big( - \frac{\|s-t\|^2}{ 2\sigma^2} \big)$ for any $u,v\in \mathbb R^d$ for some fixed $\sigma$, and $\mathbb G_k$ be its reproducing kernel Hilbert space. 
		Let $\mathcal S \subset \mathbb R^d$ be bounded. 
		Then there exists a universal constant $C>0$, depending only on $\text{diam}(\mathcal S)$, $d, m, h$, and independent of $B$, such that 
		\[
		\inf_{ g: \|g\|_{\mathbb G_k} \leq B } \Big( \int_{\mathcal S} (g-h)^2 \Big)^{1/2}   \leq C \big( \log B)^{-m/4}
		\]
		for all large $B$. In addition, the infimum is attainable.
	\end{lemma}

	\begin{proof}[Proof of Lemma \ref{lemma-aux-RKHS-Sobolev-approximation-error}]
		This approximation result is provided in \cite{zhou-13}, Proposition 18, with only the extension being that $\mathcal S$ is allowed to be a general bounded set in $ \mathbb R^d$ instead of $[0,1]^d$. The proof follows exactly the same as in \cite{zhou-13}, except that it uses a larger set of node points, rather than $\{0, 1/N, \dots, (N-1)/N\}^d$, to construct the function approximation for the general case where $\mathcal S$ is not limited to $[0,1]^d$. 
	\end{proof}

	\begin{lemma}\label{lemma-aux-Fn-cover-no}
		Let $\mathcal F_n$ be given in Assumption \ref{assumption-consistency} (f) and $J_n \geq 1$. Then 
		\[
		\log N\!\left( \delta, \mathcal F_n, \norm{\cdot}_{\infty} \right) \leq  (J_n+1) \log\left(1+ \frac{4}{\delta} \right) 
		\] 
	\end{lemma}
	
	\begin{proof}[Proof of Lemma \ref{lemma-aux-Fn-cover-no}]
		As in the proof for Lemma \ref{lemma-aux-Fn-F0-approximation-error}, the densities defining $\mathcal F_n$ can be written equivalently as in \eqref{pf-Fn-density-hermite}. 
		Let $n \in \mathbb R$ be fixed arbitrarily.  
		Pick two arbitrary $F,\tilde F\in \mathcal F_n$ whose densities are given by $f(\cdot;\tau), f(\cdot;\tilde\tau)$ according to \eqref{pf-Fn-density-hermite}, where
		\[
		\tau,\tilde\tau\in \mathcal T_n := \bigg\{ \tau \in \mathbb R^{1+J_n} \bigg|  \sum_{j=0}^{J_n} \tau_{j}^2 = 1  \bigg\}
		\]
		Denote by 
		\[
		a(u) := \sum_{j=0}^{J_n} \tau_j \overline{H}_{e_j}(u), \quad b(u) := \sum_{j=0}^{J_n} \tilde \tau_{j} \overline{H}_{e_j}(u). 
		\]
		Then 
		\begin{align}\label{pf-Fn-cover-no-aux-F1-F2-tau1-tau2}
			\sup_{u\in \mathbb R} & | F(u)-\tilde F(u)| \leq \int \left| f(u;\tau) - f(u;\tilde \tau) \right| du = \int |a(u) +b(u)| e^{-u^2/4} |a(u) -b(u)| e^{-u^2/4} du  \notag \\
			& \leq \left( \int (a(u) +b(u) )^2 e^{-u^2/2} du \right)^{1/2} \left( \int (a(u) - b(u) )^2 e^{-u^2/2} du \right)^{1/2} \notag \\
			& \leq 2 \norm{\tau-\tilde \tau}
		\end{align}
		where the last line follows from the properties $\overline{H}_{e_j}(\cdot)$'s which give
		\begin{align*}
			\int (a(u) +b(u) )^2 e^{-u^2/2} du &  \leq 2 \int a(u)^2 e^{-u^2/2} du  + 2 \int b(u)^2 e^{-u^2/2} du  \\
			& = 2 \sum_{j=0}^{J_n} \tau_j^2 + 2 \sum_{j=0}^{J_n} \tilde \tau_j^2 = 4 
		\end{align*}
		and 
		\begin{align*}
			\int (a(u) - b(u) )^2 e^{-u^2/2} du  &  = \int  \left( \sum_{j=0}^{J_n} (\tau_j -\tilde \tau_j ) \overline{H}_{e_j}(u) \right)^2 e^{-u^2/2} du  = \sum_{j=0}^{J_n} (\tau_j - \tilde \tau_j)^2 = \norm{\tau-\tilde \tau}^2
		\end{align*}
		Notice $\mathcal T_n \subset  \big\{\tau \in \mathbb R^{1+J_n} \big| \norm{\tau}\leq 1  \big\}$, and thus, 
		\begin{align}\label{pf-Fn-cover-no-aux-2}
			\log N\!\left( \delta, \mathcal T_n, \norm{\cdot} \right) \leq \log N\!\left( \delta, \big\{\tau \in \mathbb R^{1+J_n} \big| \|\tau\|\leq 1  \big\}, \norm{\cdot} \right) \leq (J_n+1) \log \left(1+ \frac{2}{\delta} \right)
		\end{align} 
		where the second inequality follows  
		from, e.g., Example 5.8 in \citet[][p. 126]{wainwright-19}.  
		Since $\tau,\tilde\tau\in \mathcal T_n$, combining \eqref{pf-Fn-cover-no-aux-F1-F2-tau1-tau2} and \eqref{pf-Fn-cover-no-aux-2} yields
		\[
		\log N\!\left( \delta, \mathcal F_n, \norm{\cdot}_{\infty} \right) \leq \log N\!\left( \delta/2, \mathcal T_n, \norm{\cdot} \right) \leq (J_n+1) \log\left(1+ \frac{4}{\delta} \right) 
		\] 
		as was to be shown. 
	\end{proof}

	\begin{lemma}\label{lemma-aux-RKHS-ball-covering-no}
		Let $\mathbb{G}_k$ be the RKHS with reproducing kernel $k(s,t) = \exp(-\|s-t\|^2/\sigma^2)$ for any $u,v \in \mathbb R^d$. 
		Let $\mathcal W \subset\mathbb R^d$ be bounded. Then for any $B>0, \delta>0$,  
		\[
		\log N\!\big(\delta,  \big\{g:\mathcal W \mapsto \mathbb R \big|\ g(\cdot) = f(\cdot), f\in \mathbb{G}_k,  \|f\|_{\mathbb{G}_k} \leq B \big\}, \|\cdot\|_\infty)  \leq   
		C_1 \frac{ \left(\log \frac{4B}{\delta} \right)^{d_{w}+1} }{  \left( \log\log \frac{4B}{\delta} \right)^{d_{w}+1}  }
		\]
		where 
		\[
		C_1 = \max\big\{ \sigma^{-d} 3^d (\text{diam}(\mathcal W) )^d, 1 \big\} e^{-d} \frac{1}{d!} \prod_{i=1}^d (4e+i)   
		\]
		depends only on $\sigma,d, \text{diam}(\mathcal W)$, and does not depend on $\delta, B$. 
	\end{lemma}

	\begin{proof}[Proof of Lemma \ref{lemma-aux-RKHS-ball-covering-no}]
		This lemma is a slight extension of the results in \cite{steinwart-fischer-21} to allow for a general radius $B$.

		Since $\mathcal W$ is bounded, applying Theorem 2.4 and Theorem 2.1 in \citet[][]{steinwart-fischer-21} gives that, for any $\delta>0$,
		\[
		\log N(\delta, \big\{g:\mathcal W \mapsto \mathbb R \big|\ g(\cdot) = f(\cdot), f\in \mathbb{G}_k,  \|f\|_{\mathbb{G}_k} \leq 1 \big\}, \|\cdot\|_\infty)  \leq C_1  \frac{ \big( \log ( 4/\delta) \big)^{d_{w}+1}  }{ \big( \log\log ( 4/\delta) \big)^{d_{w}+1}  }
		\]
		where the constant 
		\[
		C_1 = \max\big\{ \sigma^{-d} 3^d (\text{diam}(\mathcal W) )^d, 1 \big\} \frac{1}{d!} \left( \prod_{i=1}^d (4e+i) \right)e^{-d}.  
		\]
		follows from \citet[][]{steinwart-fischer-21}'s Theorem 2.1 and the remarks below their Theorem 2.4, and note that $\sigma$ therein corresponds to $1/\sigma$ here. 
		Thus, 
		\begin{align*}
			\log & N(\delta, \big\{g:\mathcal W \mapsto \mathbb R \big|\ g(\cdot) = f(\cdot), f\in \mathbb{G}_k,  \|f\|_{\mathbb{G}_k} \leq B \big\}, \|\cdot\|_\infty)  \\
			&  \leq \log N( \delta/B,  \big\{g:\mathcal W \mapsto \mathbb R \big|\ g(\cdot) = f(\cdot), f\in \mathbb{G}_k,  \|f\|_{\mathbb{G}_k} \leq 1 \big\}, \|\cdot\|_\infty) \\
			& \leq C_1 \frac{ \left(\log \frac{4B}{\delta} \right)^{d_{w}+1} }{  \left( \log\log \frac{4B}{\delta} \right)^{d_{w}+1}  }
		\end{align*}
		which completes the proof. 
	\end{proof}

	\begin{lemma}\label{lemma-An-local-centered-loss-covering-no}
		Let $\mathcal A_n  = \left\{ \ell(\cdot,\theta) - \ell(\cdot,\theta_0) | \  \theta\in \Theta_n \right\}$ where $\Theta_n = \mathcal G_n\times \mathcal F_n$ and $\mathcal G_n, \mathcal F_n$ are given in Assumption \ref{assumption-consistency} (c), (f). Provided that $J_n\geq 1$ and $B_n \geq c$ for a small constant $c$, 
		there exists a universal constant $C>0$ such that
		\[
		\log N(\delta,\mathcal A_n, \|\cdot\|_{\infty} )\leq C \left( \log \frac{ B_n}{\delta} \right)^{d_{w}+1} + C J_n \log \frac{1}{\delta} 
		\] 
		for all $\delta$ small. 
	\end{lemma}
	
	\begin{proof}[Proof of Lemma \ref{lemma-An-local-centered-loss-covering-no}]
		Notice there exists a constant $M_{\mathcal F}$ such that all densities of the cdf in $\mathcal F$ satisfy $ \norm{f}_\infty\leq M_{\mathcal F}$. 
		Then for any $\theta = (g,F),\tilde\theta = (\tilde g,\tilde F) \in \mathcal G_n\times \mathcal F_n$, it holds that
		\begin{align}\label{pf-rate-covering-aux-1}
			& \sup_{z=(y,x')'\in\{0,1\}\times \mathcal X } \left| \big(\ell(z,\theta) - \ell(z,\theta_0) \big) - \big(\ell(z,\tilde \theta) - \ell(z,\theta_0) \big) \right|  \notag \\
			& \quad = \sup_{y\in\{0,1\}, x\in \mathcal X } \left|2y-p(x,\theta)-p(x,\tilde \theta) \right| \left|p(x,\theta)-p(x,\tilde \theta) \right|   \leq 2 \sup_{x\in \mathcal X } \left|p(x,\theta)-p(x,\tilde \theta) \right|  \notag \\
			& \quad \leq 2 \|f\|_\infty \sup_{w\in\mathcal W} |g(w)-\tilde g(w)| + 2 \sup_{u\in\mathbb R} |F(u)-\tilde F(u)| \notag \\
			& \quad \leq 2M_{\mathcal F}  \sup_{w\in\mathcal W} |g(w)-\tilde g(w)| + 2 \sup_{u\in\mathbb R} |F(u)-\tilde F(u)| 
		\end{align}
		Notice that 
		\begin{align}\label{covering-no-Gn}
			\log N(\delta, \mathcal G_n, \|\cdot\|_\infty) &\leq \log N\!\big( \delta/2,  \big\{\tilde g:\mathcal W \mapsto \mathbb R \big|\ \tilde g(\cdot) = g(\cdot), g\in \mathbb{G}_k,  \|g\|_{\mathbb{G}_k} \leq B_n \big\}, \|\cdot\|_\infty \big)  \notag \\
			& \leq C_1 \left(\log \frac{8 B_n}{\delta} \right)^{d_{w}+1} 
		\end{align}
		where the last inequality is by Lemma \ref{lemma-aux-RKHS-ball-covering-no}, and $C_1$ is the universal constant therein. 
		Moreover, Lemma \ref{lemma-aux-Fn-cover-no} gives
		\[ 
		\log N(\delta, \mathcal F_n, \|\cdot\|_\infty ) \leq (J_n+1) \log \!\left(1+\frac{4}{\delta} \right) \leq 2J_n \log \frac{5}{\delta}  
		\]
		for any $J_n\geq 1, \delta\leq 1$.

		Let $N_1 := N(\delta/(4M_{\mathcal F}), \{g\in\mathcal G_n | \|g\|_{\mathbb{G}_k} \leq B_n \}, \|\cdot\|_\infty)$ and $\{g^{(i)} \}_{i=1}^{N_1}$ be a set of covering. Let $N_2 := N(\delta/4, \mathcal F_n, \|\cdot\|_\infty )$ and $\{F^{(j)} \}_{j=1}^{N_1}$ be a set of covering. Then for any $\theta =(g,F) \in \Theta_n$, there exists $g^{(i)}$ and $F^{(j)}$ such that $\|g-g^{(i)}\|_\infty \leq \delta/(4M_{\mathcal F})$ and $\|F-F^{(j)}\|_\infty \leq \delta/4$. Let $\theta_\ast = (g^{(i)}, F^{(i)})$ and note
		\[
		\sup_{z=(y,x')'\in\{0,1\}\times \mathcal X } \left| \big(\ell(z,\theta) - \ell(z,\theta_0) \big) - \big(\ell(z,\theta_\ast) - \ell(z,\theta_0) \big) \right| \leq \delta
		\]
		by \eqref{pf-rate-covering-aux-1}. 
		Therefore, $N(\delta,\mathcal A_n, \|\cdot\|_{\infty} ) \leq   N_1 N_2 $ and 
		\begin{align*}
			\log N(\delta,\mathcal A_n, \|\cdot\|_{\infty} ) & \leq \log N\!\left( \frac{\delta}{4M_{\mathcal F}}, \mathcal G_n, \|\cdot\|_\infty) \right) +  \log N\!\left( \frac{\delta}{4}, \mathcal F_n, \|\cdot\|_\infty ) \right) \\
			& \leq C_1 \left( \log \frac{32M_{\mathcal F} B_n}{\delta} \right)^{d_{w}+1}  + 2J_n \log \frac{20}{\delta}
		\end{align*}
		Notice $|a+b|^r \leq  2^{r-1} (|a|^r + |b|^r)$ for any $r\geq 1, a,b\in \mathbb R$. 
		Thus, the stated result follows: For example, when setting $C = \max\{2C_1 2^{d_w},4\}$, the stated result holds 
		for all $\delta>0$ such that $\delta\leq 1/20$ and $\delta\leq c/(32 M_{\mathcal F})$, where $\delta\leq c/(32 M_{\mathcal F})$ implies $\delta< B_n/(32M_{\mathcal F}) $. 
	\end{proof}

	\subsection{Technical Lemmas for Theorem \ref{theorem-asym-normality-wape}}\label{appendix-tech-lmas-for-AN}

	\begin{lemma}\label{lemma-aux-pathwise-derivatives}
		Let  $\Theta = \mathcal G \times \mathcal F$ be given in Assumption \ref{assumption-consistency}, and $\ell(z,\theta) = (y-p(x,\theta))^2$. 
		\begin{enumerate}[(i)]
			\itemsep -0.1em
			\item Let $u=(u_g,u_F)$ where $u_g: \mathbb R^{d_w} \to \mathbb R$, $u_F:\mathbb R \to \mathbb R$ and $u_F$ is continuous.
			For $\tilde \theta = (\tilde g,\tilde F) \in \Theta$, the pathwise derivative of $p(x,\theta)$ at $\tilde \theta$ along direction $u$ is 
			\begin{align*}
				\frac{\partial p(x,\tilde \theta) }{\partial \theta}[u] 
				=  u_F\!\big(v+ \tilde g(w) \big) + u_g(w) \tilde f(v+\tilde g(w))  
			\end{align*}
			where $\tilde f (\cdot)$ is the derivative of $\tilde F$. Consequently, $\frac{\partial p(x,\tilde \theta) }{\partial \theta}[u] $ is linear in $u$, and 
			\[
			\frac{\partial p(x,\theta_0) }{\partial \theta}[\theta-\theta_0]  = F(v+g_0(w))-F_0(v+g_0(w)) + f_0(v+g_0(w))(g(w)-g_0(w))
			\]
			\item 
			\[
			\frac{\partial \ell (z,\tilde \theta) }{\partial \theta}[u] =  - 2(y - p(x,\tilde\theta) ) \frac{\partial p(x,\tilde \theta) }{\partial \theta}[u]
			\]
			Thus, $\frac{\partial \ell (z,\tilde \theta) }{\partial \theta}[u]$ is linear in $u$, and $\frac{\partial \ell(z,\theta_0) }{\partial \theta}[\theta-\theta_0]  = -2(y-p(x,\theta_0 )) \frac{\partial p(x,\theta_0) }{\partial\theta}[\theta-\theta_0]$. 
			\item Let $u_1 = (u_{1g},u_{1F}), u_2 = (u_{2g},u_{2F})$ where $u_{1F}, u_{2F}$ are continuously differentiable with derivatives $ u_{1F}', u_{2F}'$ respectively. For $\tilde \theta = (\tilde g,\tilde F) \in \Theta$ where $\tilde F$ is twice continuously differentiable with second derivative  $ \tilde F''$, 
			\begin{align*}
				\frac{\partial^2 p(x, \tilde\theta)}{ \partial \theta \partial \theta} [u_1,u_2]  = u_{1g}(w) \tilde F''\!\big( v+\tilde g(w) \big) u_{2g}(w) + u_{1F}'\!(v+\tilde g(w)) u_{2g}(w) + u_{1g}(w) u_{2F}'\!(v+\tilde g(w))
			\end{align*} 
			and 
			\begin{align*}
				\frac{\partial^2 \ell(z, \tilde\theta)}{ \partial \theta \partial \theta} [u_1,u_2] =  
				-2 \bigg( (y-p(x,\tilde\theta)) \frac{\partial^2 p(x, \tilde\theta)}{ \partial \theta \partial \theta} [u_1,u_2] - \frac{\partial p(x,\tilde\theta)}{\partial\theta}[u_1] \frac{\partial p(x,\tilde\theta)}{\partial\theta}[u_2] \bigg)
			\end{align*} 
			\item  For any constant $c \in \mathbb R$,  
			\[
			\frac{\partial^2 p(x, \tilde\theta)}{ \partial \theta \partial \theta} [u_1,c u_2] = c \frac{\partial^2 p(x, \tilde\theta)}{ \partial \theta \partial \theta} [u_1,u_2]
			\]
			and 
			\[
			\frac{\partial^2 \ell(z, \tilde\theta ) }{\partial \theta \partial\theta}[u_1, c u_2]  = c\frac{\partial^2 \ell(z, \tilde\theta ) }{\partial \theta \partial\theta}[u_1, u_2].
			\]
		\end{enumerate}
	\end{lemma}

	\begin{proof}[Proof of Lemma \ref{lemma-aux-pathwise-derivatives}]
		For $\tilde \theta = (\tilde g,\tilde F) \in \Theta$, notice $\tilde F(\cdot)$ is continuously differentiable and denote by $\tilde f(\cdot)$ its derivative. 
		
		For Part (i), we write
		\begin{align*}
			& p\big(x, \tilde \theta+t u \big) - p(x,\tilde \theta) = (\tilde F + tu_{F})\!\big(v+ \tilde g(w) + tu_g(w)\big) - \tilde F(v+\tilde g(w)) \\
			& \quad = t  u_F\big(v+ \tilde g(w) + t u_g(w)\big)  + \Big( \tilde F\big(v+ \tilde g(w) + tu_g(w)\big)  - \tilde F(v+\tilde g(w)) \Big).
		\end{align*}
		Thus,
		\begin{align}\label{pf-pathwise-derivatives-aux-ptheta-cts}
			p(x,\tilde \theta+tu) \to p(x,\tilde\theta)  \quad \text{ as $\ t \to 0$ }
		\end{align}
		due to the continuity of $\tilde F$, and 
		\begin{align*}
			\frac{\partial p(x,\tilde \theta) }{\partial \theta}[u] &:= \lim_{t\to 0} \frac{1}{t} \bigg( t  u_F\big(v+ \tilde g(w) + t u_g(w)\big)  + \Big( \tilde F\big(v+ \tilde g(w) + tu_g(w)\big)  - \tilde F(v+\tilde g(w)) \Big) \bigg) \\
			& =  u_F\!\big(v+ \tilde g(w) \big) + u_g(w) \tilde f(v+\tilde g(w))
		\end{align*}
		by continuity of $u_F$ and that $\tilde F$ is continuously differentiable with derivative $\tilde f$. Notice $u_F\!\big(v+ \tilde g(w) \big) + u_g(w) \tilde f(v+\tilde g(w))$ is linear in $u=(u_g,u_F)$. This proves Part (i). 
		
		For Part (ii), notice $\ell(z,\theta) = (y-p(x,\theta))^2$ gives
		\begin{align*}
			\frac{\partial \ell (z,\tilde \theta) }{\partial \theta}[u] &:= \lim_{t\to 0} \frac{1}{t} \bigg( (y-p(x,\tilde\theta+tu) )^2 - (y-p(x,\tilde\theta) )^2 \Big) \bigg)  \\
			& =  \lim_{t\to 0} \frac{1}{t} \bigg( (2y - p(x,\tilde\theta+tu)-p(x,\tilde\theta) ) ( - p(x,\tilde\theta+tu) + p(x,\tilde\theta) ) \\
			& = - 2(y - p(x,\tilde\theta) ) \frac{\partial p(x,\tilde \theta) }{\partial \theta}[u]
		\end{align*}
		by \eqref{pf-pathwise-derivatives-aux-ptheta-cts}.  
		This proves Part (ii).

		For Part (iii), let $u_1 = (u_{1g},u_{1F}), u_2 = (u_{2g},u_{2F})$ where $u_{1F}, u_{2F}$ are continuously differentiable with derivatives $ u_{1F}', u_{2F}'$ respectively. Using Part (i), we have
		\begin{align*}
			& \frac{\partial^2 p(x, \tilde\theta)}{ \partial \theta \partial \theta} [u_1,u_2] := \lim_{t\to 0} \frac{1}{t} \bigg( \frac{\partial p(x, \tilde\theta+t u_2)}{ \partial \theta}[u_1] - \frac{\partial p(x, \tilde\theta)}{ \partial \theta}[u_1] \bigg) \\ 
			& \quad =  \lim_{t\to 0} \frac{1}{t} \bigg(  u_{1g}(w) \big(\tilde F' + t u_{2F}' )\!\big( v+\tilde g(w) + tu_{2g}(w) \big) + u_{1F}\!\big( v+\tilde g(w) + tu_{2g}(w) \big)   \\
			& \quad \quad \quad  \quad \quad \quad -   u_{1g}(w) \tilde F'\!\big( v+\tilde g(w) \big) - u_{1F}\!(v+\tilde g(w)) \bigg) \\ 
			& \quad = \lim_{t\to 0} \frac{1}{t} \bigg(  u_{1g}(w) \tilde F'\!\big( v+\tilde g(w) + tu_{2g}(w) \big) -   u_{1g}(w) \tilde F'\!\big( v+\tilde g(w) \big) \bigg) \\
			& \quad \quad +  \lim_{t\to 0} \frac{ u_{1g}(w)t u_{2F}'\!\big( v+\tilde g(w) + tu_{2g}(w) \big)}{t} + \lim_{t\to 0} \frac{u_{1F}\!\big( v+\tilde g(w) + tu_{2g}(w) \big)  -  u_{1F}\!(v+\tilde g(w)) }{t} \\ 
			& \quad = u_{1g}(w) \tilde F''\!\big( v+\tilde g(w) \big) u_{2g}(w)  + u_{1g}(w) u_{2F}'\!(v+\tilde g(w)) + u_{1F}'\!(v+\tilde g(w)) u_{2g}(w)
		\end{align*}
		where the last line follows from the continuity of $\tilde F'', u_{1F}', u_{2F}'$. 
		This proves the first equality in Part (iii). 
		For the second equality in Part (iii), notice Part (ii) gives
		\begin{align*}
			& \frac{\partial^2 \ell(z, \tilde\theta)}{ \partial \theta \partial \theta} [u_1,u_2]  := \lim_{t\to 0} \frac{1}{t} \bigg( \frac{\partial \ell(z, \tilde\theta+t u_2)}{ \partial \theta}[u_1] - \frac{\partial \ell(z, \tilde\theta)}{ \partial \theta}[u_1] \bigg)   \\
			& \quad = -2 \lim_{t\to 0} \frac{1}{t} \bigg( (y-p(x,\tilde\theta+tu_2)) \frac{\partial p(x, \tilde\theta+t u_2)}{ \partial \theta}[u_1] - (y-p(x,\tilde\theta)) \frac{\partial p(x, \tilde\theta)}{ \partial \theta}[u_1] \bigg)   \\
			& \quad = -2 \lim_{t\to 0} \frac{1}{t}  \bigg[  (y-p(x,\tilde\theta+tu_2))  \bigg(  \frac{\partial p(x, \tilde\theta+t u_2)}{ \partial \theta}[u_1]  - \frac{\partial p(x, \tilde\theta)}{ \partial \theta}[u_1] \bigg) \\
			& \quad \quad  \quad \quad \quad \quad  \quad + \Big( (y-p(x,\tilde\theta+tu_2)) - (y-p(x,\tilde\theta)) \Big) \frac{\partial p(x, \tilde\theta)}{ \partial \theta}[u_1]  \bigg] \\
			& \quad = -2 \bigg( (y-p(x,\tilde\theta)) \frac{\partial^2 p(x, \tilde\theta)}{ \partial \theta \partial \theta} [u_1,u_2] - \frac{\partial p(x,\tilde\theta)}{\partial\theta}[u_1] \frac{\partial p(x,\tilde\theta)}{\partial\theta}[u_2] \bigg)
		\end{align*}
		where the last line follows from $p(x,\tilde\theta+tu_2) \to p(x,\tilde \theta) $ when $t\to 0$ in \eqref{pf-pathwise-derivatives-aux-ptheta-cts}. 
		
		For Part (iv), the first equality follows from the form of $\frac{\partial^2 p(x, \tilde\theta)}{ \partial \theta \partial \theta} [u_1,u_2] $ in Part (iii), together with $(c u_{2F})' = c u_{2F}' $. The second equality follows from the form of $\frac{\partial^2 \ell(z, \tilde\theta)}{ \partial \theta \partial \theta} [u_1,u_2]$ in Part (iii), and the linearity of  $\frac{\partial p(x, \tilde\theta)}{\partial \theta} [u] $ in $u$.  
	\end{proof}

	\begin{lemma}\label{lemma-aux-partialell-fisherinprod}
		Let Assumptions \ref{assumption-wape-asym-norm} (\ref{aspn-FisherNorm-L2ptheta-pre-ptheta-remainder}), (\ref{aspn-vn-vast-approx-error-quarter-rate}), (\ref{aspn-ptheta-ptheta0-deviation})  
		hold. For any $\tilde\theta \in \{ \theta\in \Theta_n : d_{\Theta}(\theta,\theta_0) = o(1), \|\theta-\theta_0\|_F = o(n^{-1/4}) \}$, 
		it holds that 
		\[
		\mathbb E \frac{\partial \ell(Z,\tilde\theta)}{\partial \theta}[v_n^\ast] = 2\langle \tilde\theta-\theta_0, v_n^\ast\rangle_F  + o(n^{-1/2})
		\] 
	\end{lemma}

	\begin{proof}[Proof of Lemma \ref{lemma-aux-partialell-fisherinprod}]
		By Lemma \ref{lemma-aux-pathwise-derivatives} (ii), we have
		\begin{align*}
			\mathbb E & \frac{\partial \ell(Z,\tilde\theta)}{\partial \theta}[u] - 2\langle \tilde\theta-\theta_0, u\rangle_F \\
			& = -2 \mathbb E (Y-p(X,\tilde\theta) ) \frac{\partial p(X,\tilde\theta) }{\partial \theta} [u] -2 \mathbb E \frac{\partial p(X,\theta_0) }{\partial \theta} [\tilde\theta-\theta_0] \frac{\partial p(X,\theta_0) }{\partial \theta} [u] \\
			& = -2 \mathbb E (p(X,\theta_0) -p(X,\tilde\theta) ) \frac{\partial p(X,\tilde\theta) }{\partial \theta} [u] -2 \mathbb E \frac{\partial p(X,\theta_0) }{\partial \theta} [\tilde\theta-\theta_0] \frac{\partial p(X,\theta_0) }{\partial \theta} [u] \\
			& = 2 A_{n1}(u) + 2 A_{n2}(u)
		\end{align*}
		where 
		\begin{align*}
			A_{n1}(u) & := \mathbb E \big( p(X,\tilde\theta)-p(X,\theta_0) \big) \bigg( \frac{\partial p(X,\tilde\theta) }{\partial \theta} [u] - \frac{\partial p(X,\theta_0) }{\partial \theta} [u]  \bigg)  \\
			A_{n2}(u) & := \mathbb E \bigg( p(X,\tilde\theta)-p(X,\theta_0) - \frac{\partial p(X,\theta_0) }{\partial \theta} [\tilde\theta-\theta_0] \bigg) \frac{\partial p(X,\theta_0) }{\partial \theta} [u]
		\end{align*}
		It suffices to show that $A_{n1}(v_n^\ast), A_{n2}(v_n^\ast) = o(n^{-1/2})$, 
		which are satisfied since 
		\begin{align*}
			|A_{n1}  (v_n^\ast) |
			& \leq \bigg( \mathbb E \big( p(X,\tilde\theta)-p(X,\theta_0) \big)^2 \bigg)^{1/2}  \bigg( \mathbb E \bigg[  \frac{\partial p(X,\tilde\theta) }{\partial \theta} [v_n^\ast] - \frac{\partial p(X,\theta_0) }{\partial \theta} [v_n^\ast] \bigg]^2 \bigg)^{1/2}  \\
			& = o(n^{-1/4}) o(n^{-1/4}) = o(n^{-1/2})
		\end{align*}
		by Assumption (\ref{aspn-ptheta-ptheta0-deviation}.i), $\|\tilde \theta-\theta_0\|_F = o(n^{-1/4})$, which implies $\mathbb E \big( p(X,\tilde\theta)-p(X,\theta_0))^2 = o(n^{-1/2})$ by \eqref{eq-FisherNorm-L2ptheta} under Assumption \ref{assumption-wape-asym-norm}(\ref{aspn-FisherNorm-L2ptheta-pre-ptheta-remainder}). 
		Moreover, 
		\begin{align*}
			|A_{n2} (v_n^\ast) | 
			\leq |A_{n2} (v^\ast) | + |A_{n2} (v_n^\ast) -A_{n2} (v^\ast) |   = o(n^{-1/2})
		\end{align*}
		since $A_{n2} (v^\ast) = o(n^{-1/2})$ by Assumption \ref{assumption-wape-asym-norm}(\ref{aspn-ptheta-ptheta0-deviation}.ii), and 
		\begin{align*}
			& | A_{n2} (v_n^\ast) -A_{n2} (v^\ast) |  \\
			& \leq  \bigg( \mathbb E \bigg[  p(X,\tilde\theta)-p(X,\theta_0) - \frac{\partial p(X,\theta_0) }{\partial \theta} [\tilde\theta-\theta_0] \bigg]^2 \bigg)^{1/2} \bigg( \mathbb E \bigg[  \frac{\partial p(X,\theta_0) }{\partial \theta} [v_n^\ast] - \frac{\partial p(X,\theta_0) }{\partial \theta} [v^\ast]  \bigg]^2 \bigg)^{1/2}   \\
			&=  o(n^{-1/4}) \|v_n^\ast -v^\ast\|_F = o(n^{-1/2})
		\end{align*}
		by Assumption \ref{assumption-wape-asym-norm}(\ref{aspn-FisherNorm-L2ptheta-pre-ptheta-remainder}), $\|\theta-\theta_0\|_F = o(n^{-1/4})$, and Assumption \ref{assumption-wape-asym-norm}(\ref{aspn-vn-vast-approx-error-quarter-rate}). 
	\end{proof}

	\begin{lemma}\label{lemma-aux-sample-partialell-hattheta}
		Let Assumptions \ref{assumption-wape-asym-norm} (\ref{aspn-FisherNorm-L2ptheta-pre-ptheta-remainder}), (\ref{aspn-vn-vast-approx-error-quarter-rate}), (\ref{aspn-ptheta-ptheta0-deviation}), 
		(\ref{aspn-FisherNorm-sample-partialell_ell0-remainder}) hold. 
		For any $\tilde\theta \in \{ \theta\in \Theta_n : d_{\Theta}(\theta,\theta_0) = o(1), \|\theta-\theta_0\|_F = o(n^{-1/4}) \}$, 
		it holds that
		\[
		\frac{1}{n} \sum_{i=1}^n \frac{\partial \ell(Z,\tilde\theta)}{\partial \theta}[v_n^\ast] 
		= 2 \langle \tilde\theta-\theta_0, v^\ast \rangle_F + \frac{1}{n} \sum_{i=1}^n \bigg(  \frac{\partial \ell(Z,\theta_0)}{\partial \theta}[v^\ast]  - \mathbb E \frac{\partial \ell(Z,\theta_0)}{\partial \theta}[v^\ast]  \bigg) + o_p(n^{-1/2})
		\]
	\end{lemma}

	\begin{proof}[Proof of Lemma \ref{lemma-aux-sample-partialell-hattheta}]
		Write 
		\begin{align}\label{pf-wape-aux-lemma-0}
			\frac{1}{n} \sum_{i=1}^n \frac{\partial \ell(Z,\tilde\theta)}{\partial \theta}[v_n^\ast] = \frac{1}{\sqrt{n}}  A_{n1} + \frac{1}{\sqrt{n}} A_{n2} + \frac{1}{\sqrt{n}} A_{n3} + A_{n4}
		\end{align}
		where
		\begin{align*}
			A_{n1} & := \frac{1}{\sqrt{n}} \sum_{i=1}^n \bigg( \frac{\partial \ell(Z,\tilde\theta)}{\partial \theta}[v_n^\ast] -\frac{\partial \ell(Z,\theta_0)}{\partial \theta}[v_n^\ast]  
			- \mathbb E \bigg[ \frac{\partial \ell(Z,\tilde\theta)}{\partial \theta}[v_n^\ast] -\frac{\partial \ell(Z,\theta_0)}{\partial \theta}[v_n^\ast] \bigg] \bigg) \\
			A_{n2} & := \frac{1}{\sqrt{n}} \sum_{i=1}^n \bigg( \frac{\partial \ell(Z,\theta_0)}{\partial \theta}[v_n^\ast] -\frac{\partial \ell(Z,\theta_0)}{\partial \theta}[v^\ast]  - \mathbb E\bigg[ \frac{\partial \ell(Z,\theta_0)}{\partial \theta}[v_n^\ast] -\frac{\partial \ell(Z,\theta_0)}{\partial \theta}[v^\ast] \bigg] \bigg) \\
			A_{n3} & :=  \frac{1}{\sqrt{n}} \sum_{i=1}^n \bigg(  \frac{\partial \ell(Z,\theta_0)}{\partial \theta}[v^\ast]  - \mathbb E \frac{\partial \ell(Z,\theta_0)}{\partial \theta}[v^\ast]  \bigg) \\
			A_{n4} &: = \mathbb E  \frac{\partial \ell(Z,\tilde \theta)}{\partial \theta}[v_n^\ast]
		\end{align*}
		Notice 
		\begin{align}\label{pf-wape-aux-lemma-0-A1}
			A_{n1} = o_p(1)
		\end{align}
		by Assumption \ref{assumption-wape-asym-norm}(\ref{aspn-FisherNorm-sample-partialell_ell0-remainder}). Moreover, 
		\begin{align}\label{pf-wape-aux-lemma-0-A2}
			A_{n2} = o_p(1)
		\end{align}
		since 
		\begin{align*}
			\mathbb E\bigg[ \frac{\partial \ell(Z,\theta_0)}{\partial \theta}[v_n^\ast] -\frac{\partial \ell(Z,\theta_0)}{\partial \theta}[v^\ast] \bigg] 
			= - 2 \mathbb E(Y-p(X,\theta_0))\frac{\partial p(X,\theta_0)}{\partial \theta}[v_n^\ast-v^\ast] = 0
		\end{align*}
		by Lemma \ref{lemma-aux-pathwise-derivatives} (ii) and $\mathbb E(Y|X) =p(X,\theta_0)$, 
		and 
		\begin{align*}
			\mathbb E &\bigg[ \frac{\partial \ell(Z,\theta_0)}{\partial \theta}[v_n^\ast] -\frac{\partial \ell(Z,\theta_0)}{\partial \theta}[v^\ast] \bigg]^2 
			= 4 \mathbb E (Y-p(X,\theta_0))^2 \bigg(  \frac{\partial p(X,\theta_0)}{\partial \theta}[v_n^\ast-v^\ast] \bigg)^2 \\
			& \leq 4 \mathbb E  \bigg(  \frac{\partial p(X,\theta_0)}{\partial \theta}[v_n^\ast-v^\ast] \bigg)^2 = 4 \|v_n^\ast-v^\ast\|_F^2 = o(1)
		\end{align*}
		By Lemma \ref{lemma-aux-partialell-fisherinprod}, 
		\begin{align}\label{pf-wape-aux-lemma-0-A4}
			A_{n4} & = 2\langle \tilde\theta-\theta_0, v_n^\ast \rangle_F  + o(n^{-1/2}) \notag  \\
			& = 2\langle \tilde\theta-\theta_0, v^\ast \rangle_F  + o(n^{-1/2}) 
		\end{align}
		where the last line follows from 
		$|\langle \tilde \theta-\theta_0, v_n^\ast -v^\ast \rangle_F| \leq \|\tilde \theta-\theta_0\|_F \|v_n^\ast-v^\ast\|_F = o(n^{-1/4}) o(n^{-1/4}) = o(n^{-1/2})$ by Assumption \ref{assumption-wape-asym-norm}(\ref{aspn-vn-vast-approx-error-quarter-rate}) and $ \|\tilde \theta-\theta_0\|_F  =o(n^{-1/4})$. 
		Combining \eqref{pf-wape-aux-lemma-0}, \eqref{pf-wape-aux-lemma-0-A1}, \eqref{pf-wape-aux-lemma-0-A2}, \eqref{pf-wape-aux-lemma-0-A4} yields
		\begin{align*}
			\frac{1}{n} \sum_{i=1}^n \frac{\partial \ell(Z,\tilde\theta)}{\partial \theta}[v_n^\ast] 
			= 2 \langle \tilde\theta-\theta_0, v^\ast \rangle_F + \frac{1}{n} \sum_{i=1}^n \bigg(  \frac{\partial \ell(Z,\theta_0)}{\partial \theta}[v^\ast]  - \mathbb E \frac{\partial \ell(Z,\theta_0)}{\partial \theta}[v^\ast]  \bigg) + o_p(n^{-1/2})
		\end{align*} 
		as was to be shown.  
	\end{proof}

	\section{Robustness Check: Single-city Subsample (New York)}\label{appendix-NY-sample}
	
	This subsection presents a robustness analysis 
	motivated by the concern that $\varepsilon_i$ may contain unobserved time components or courts' location-specific component 
	that are correlated with environmental exposure $W_{jt}$ or with $V_{i}$.

	To mitigate heterogeneity arising from the fact that courts in different cities face intrinsically different weather and, as mentioned in \cite{heyes-saberian-19}, judges are largely attached to a single court, we restrict the sample to cases decided in New York (NY). In the restricted sample, $W_{jt}$ follows a common local environment process $W_{t}^{NY}$ over time for all judges. 
	We choose New York because it is the city with the largest number of cases in the data (30.3\% of the total, with 29,478 observations).

	For this NY subsample, the baseline score $V_{i}$ for case $i$ is constructed using observations decided in \emph{non-NY} cities only. 
	Specifically, for each case $i$, we compute the average approval rate using cases in non-NY cities whose applicants' nationality, case-type, and calendar month (year-month) are the same as case $i$,  and $V_{i}$ is the log-odds transform of this rate. 
	This construction provides a baseline score specific to the cell of the applicant's nationality, time in year-month, and case type. 
	$V_{i}$, by construction, captures heterogeneity in applications' nationality over time, and avoids mechanical dependence between $V_{i}$ and the NY outcomes used in estimation. 
	In this restricted sample, the empirical range of $V_i$ is $[-3.738, 2.708]$; the $0.1\%, 0.5\%, 1\%, 50\%$, $ 99\%, 99.5\%$ and $99.9\%$ quantiles are $-3.22, -2.71, -2.46, -0.43$, $1.54, 1.95$, and $2.48$, respectively.

	To further control for unobserved time heterogeneity, we include year dummy variables (for 2000-2004) linearly in the systematic component. 
	Implementation is the same as in the baseline specification, 
	except that we augment the systematic component by year indicators and jointly optimize over their coefficients along $(\zeta,\tau)$ in \eqref{obj-aux-1-pca}. 
	Such coarse time controls are feasible and can be easily incorporated into our procedure. 
	On the other hand, including very fine fixed effects (e.g., court location-by-month or judge-by-month indicators) is not feasible in this generic nonparametric nonlinear setting, 
	since it would require estimating too many parameters.

	Table \ref{table-apes-NY} reports the 90\% confidence intervals 
	for the APE and cAPEs in the New York subsample, with year indicators included linearly in the the systematic component.  
	Compared to Table \ref{table-apes}, 
	KNP in this robustness check yields a similarly negative $cAPE_{T|T>70^{\circ}F}$, while $cAPE_{T|T<70^{\circ}F}$ becomes significantly positive and the APE becomes insignificant.  
	In addition, under the LPM benchmark, the estimated APE and both cAPEs in this restricted sample are not significant at the 10\% significance level.

	\begin{table}[h!]
		\caption{ 
			90\% confidence interval of estimated APEs and cAPEs (in \% of prob) of temperature: New York subsample 
		} 
		\label{table-apes-NY}
		\begin{center}
			\centering
			\begin{tabular}{c c c c c c c c}
				\hline \hline

				&	&	 APE &&  	$cAPE_{T|T<70^{\circ}F}$ 	&&		$cAPE_{T|T>70^{\circ}F}$ 	\\ [0.5ex] \hline  
				
				KNP  &  &	$[-0.050, 0.121]$	&& $ [0.069, 0.379] $	&&	$ [-0.655, -0.146]  $		\\   
				LPM     && $[-0.161, 1.137]$  && $[-0.161, 1.137]$  && $[-0.161, 1.137]$  \\

				\hline \hline
			\end{tabular}
		\end{center}    
		
		\footnotesize
		Notes:  
		This table presents 90\% confidence intervals for the estimated APEs and cAPEs obtained from fitting KNP and the linear probability model (LPM), using New York subsample.  
		Compared to Table \ref{table-apes}, the specification here is modified to include year dummy variables (2000-2004) linearly, along with the linear $V_i$. 
		Effects are reported in percentage points of granting probability and correspond to a one standard deviation increase in temperature ($17.8^{\circ}$F in the New York subsample). 
		The effective sample size is $n= 29,478$.   
	\end{table}

	\section{Additional Details for Implementation}\label{appendix-implementation}

	\subsection{Closed Form Distribution Functions in $\mathcal F_n$}\label{appendix-implementation-cdf-form} 
	
	In this section, we give the simple closed form of $F(\cdot;\tau) \in \mathcal F_n$ without integrals.

	For implementation, it is convenient to rewrite the form of the density of a distribution function in $\mathcal{F}_n$ as 
	\begin{align}\label{ftau-aux-1}
		f(u; \tau) = \frac{1}{\psi_{J_n}} \left( \sum_{r=0}^{J_n} \tau_r  u^r \right)^2 \phi(u)
	\end{align}
	where 
	\[
	\psi_{J_n} = \int \left( \sum_{r=0}^{J_n} \tau_r  u^r \right)^2 \phi(u) du 
	\]
	is a normalization constant to ensure that $f$ is a proper probability density function, and $\phi$ is the density of the standard normal distribution. Since $f_{\tau}$ is invariant to multiplication of $(\tau_0,\tau_1,\dots,\tau_{J_n})'$ by a scalar, we set $\tau_0 = 1$ as a normalization, and redefine $\tau = (\tau_1,\dots, \tau_{J_n})'$.  
	The optimization involves $F(u; \tau) = \int_{-\infty}^{u} f_{\tau}(z) dz$, which has a simple closed-form due to the specific form of $f(u; \tau)$ given by the Hermite polynomial approximation. To obtain the specific forms of $f(\cdot; \tau)$ and $F(\cdot; \tau)$ used for computation, we define\footnote{Constants $\gamma_h$ for $h=0,1,\dots,2{J_n}$ are defined based on the Hermite polynomial coefficients $\tau_1,\dots,\tau_{J_n}$ in \eqref{ftau-aux-1}, along with the fixed $\tau_0=1$ for normalization. We suppress this dependence for notational simplicity.}
	\begin{align}\label{appendix-gamma-h}
		\gamma_h = \sum_{r= 0 \vee (h-{J_n})}^{h \wedge {J_n}} \tau_r \tau_{h-r}
	\end{align}
	for $h = 0,1,\dots, 2{J_n}$. Some algebras show that
	\begin{equation}\label{ftau-Ftau-aux-1}
		\begin{split}
			f(u; \tau) & =  \frac{1}{\psi_{J_n}} \sum_{h=0}^{2{J_n}} \gamma_h \left( u^h \phi(u)\right), \quad \text{ where } \psi_{J_n} = \sum_{h=0}^{2{J_n}} \gamma_h \int u^h \phi(u) du \\
			F(u; \tau) & =  \frac{1}{\psi_{J_n}} \sum_{h=0}^{2{J_n}} \gamma_h A_h(u), \quad \text{ where } A_h(u) = \int_{-\infty}^u z^h \phi(z) dz. 
		\end{split}
	\end{equation}
	We notice that $a_h :=\int u^h \phi(u) du$, and thus $\psi_{J_n}$, can be easily computed. In particular, we have 
	\begin{align}\label{ah-recursive}
		\begin{split}
			& a_0=1, \quad  a_1 = 0, \quad  a_2=1,   \\
			& a_{h} = (h-1) a_{h-2} \quad \text{for} \quad  h=3,4, \dots
		\end{split}
	\end{align} 
	which can be obtained recursively. This ensures easy evaluations of $f(\cdot; \tau)$.  Furthermore, $F(u; \tau) $ can also be easily obtained, since $A_h(u)$ can be evaluated based on a recursive procedure without numerical integration. More specifically, some algebras show that
	\begin{equation}\label{Ah-recursive}
		\begin{split}
			A_0(u) & = \Phi(u), \quad 
			A_1(u) = - \phi(u) \\
			A_2(u) & = u A_1(u) + A_0(u) \\ 
			A_{h}(u) & = u\big( A_{h-1}(u) - (h-2) A_{h-3}(u)  \big) + (h-1) A_{h-2}(u)  \quad \text{for} \quad  h=3,4, \dots  
		\end{split}
	\end{equation} 
	Thus, for any given Hermite polynomial coefficients $\tau$, $f(u; \tau)$ and $F(u; \tau)$ in \eqref{ftau-Ftau-aux-1} can be easily evaluated. 
	Specifically,  
	\begin{equation}\label{ftau-Ftau-aux-2}
		\begin{split}
			f(u; \tau)  =  \frac{1}{\psi_{J_n}} \sum_{h=0}^{2{J_n}} \gamma_h  u^h \phi(u) 
			\quad \text{and} \quad   
			F(u; \tau) =  \frac{1}{\psi_{J_n}} \sum_{h=0}^{2{J_n}} \gamma_h A_h(u) ,  
		\end{split}
	\end{equation}
	where constants $(\gamma_h)_{h=0}^{2J_n}$ and constant $\psi_{J_n}$ are given by 
	\begin{align*}
		\gamma_h = \sum_{r= 0 \vee (h-{J_n})}^{h \wedge {J_n}} \tau_r \tau_{h-r} 
		\quad \text{and} \quad   
		\psi_{J_n} = \sum_{h=0}^{2{J_n}} \gamma_h a_h, 
	\end{align*}
	and constants $(a_h)_{h=0}^{2J_n}$ are obtained recursively by \eqref{ah-recursive}, and functions $A_h(\cdot), h=0,1,\dots, 2J_n$ are obtained recursively by \eqref{Ah-recursive}.

	\subsection{Analytical Form of the Gradients for Numerical Optimization}\label{appendix-implementation-gradient-form}
	
	Based on Section \ref{appendix-implementation-cdf-form}, we now give the closed-form expressions of the gradients for the  
	objective \eqref{obj-aux-1-pca}. 
	The closed-form expressions of the gradients for the objective in \eqref{obj-aux-1} can be obtained in the same way, and thus, are omitted.

	Following Section \ref{appendix-implementation-cdf-form}, the optimization in \eqref{obj-aux-1-pca} becomes  
	\begin{align*}
		\min_{\zeta\in \mathbb R^{m}, \tau \in \mathbb R^{J_n} } & \quad  
		\bigg\{ \hat Q(\zeta,\tau) := \frac{1}{n}\sum_{i=1}^n  \Big( Y_i - F\big(V_i + [\hat U_m\zeta]_{i+1} -  [\hat U_m\zeta]_{1} ; \tau \big)   \Big)^2    \bigg\}   
		\\
		\textsl{s.t. } & \quad  \zeta'\hat \Lambda_m^{-1} \zeta \leq B_n^2, 
	\end{align*}
	where $F(\cdot;\tau)$ for $\tau=(\tau_1,\dots,\tau_{J_n})$ is evaluated based on \eqref{ftau-Ftau-aux-2}. 
	
	This optimization can be solved using standard constrained numerical optimization packages in existing software, for example, \texttt{fmincon} in MATLAB with an interior-point method to handle the constraint. 
	We supply the optimization routine with the gradients of the objective with respect to $\zeta, \tau$, whose analytical forms are given as follows.\footnote{E.g., set \texttt{SpecifyObjectiveGradient = true} using MATLAB's optimization routines.} 
	
	Let $\iota_j$ be an $(n+1)\times 1$ vector whose $(j+1)$-element is one and all other elements are zero.  
	The gradients of the objective $\hat Q(\zeta,\tau)$ with respect to $\zeta, \tau$ are, respectively, 
	\begin{align*}
		& \nabla_{\zeta} \hat Q(\zeta,\tau) 
		= -\frac{2}{n} \sum_{i=1}^n \Big( Y_i - F\big(V_i +  (\iota_i - \iota_0)'\hat U_m \zeta  ; \tau \big)   \Big) f\big(V_i + (\iota_i - \iota_0)'\hat U_m \zeta; \tau \big)   \hat U_m'(\iota_i - \iota_0) \\
		& \nabla_{\tau} \hat Q(\zeta,\tau) 
		= -\frac{2}{n} \sum_{i=1}^n \Big( Y_i - F\big(V_i +  (\iota_i - \iota_0)'\hat U_m \zeta  ; \tau \big)   \Big) \nabla_\tau F(V_i +  (\iota_i - \iota_0)'\hat U_m \zeta ;\tau)
	\end{align*}
	where  $F(\cdot;\tau), f(\cdot;\tau)$ are evaluated using the closed form in Appendix \ref{appendix-implementation-cdf-form},  
	constants $(a_h)_{h=0}^{2J_n}$ 
	and the gradient $\nabla_\tau F(u;\tau)$ has the closed form 
	\begin{align*}
		\nabla_\tau F(u;\tau) = 
		-\frac{1}{\psi_{J_n}^2}  \bigg( \sum_{h=0}^{2{J_n}} a_h \big( \nabla_{\tau} \gamma_h\big)  \bigg) \bigg( \sum_{h=0}^{2J_n} \gamma_h A_h(u) \bigg) 
		+ \frac{1}{\psi_{J_n}} \sum_{h=0}^{2{J_n}} \big(\nabla_{\tau} \gamma_h \big) A_h(u), 
	\end{align*}
	where 
	the gradients $\nabla_{\tau} \gamma_h$ for $h=0,1,\dots, 2J_n$ can be easily obtained from \eqref{appendix-gamma-h}  
	\begin{align*}
		\begin{pmatrix}
			\frac{\partial\gamma_0}{\partial\tau_1} & \frac{\partial\gamma_1}{\partial\tau_1} & \dots  
			& \frac{\partial\gamma_{2J_{n}}}{\partial\tau_1} \\
			& & \dots \\ 
			\frac{\partial\gamma_0}{\partial\tau_{J_n}} & \frac{\partial\gamma_1}{\partial\tau_{J_n}} & \dots  
			&\frac{\partial\gamma_{2J_{n}}}{\partial\tau_{J_n}} 
		\end{pmatrix}
		= \begin{pmatrix}
			0 & 2\tau_0 & 2\tau_1 & \dots & 2\tau_{J_n-1} & 2\tau_{J_n}  & 0 & \dots & 0 \\
			0 & 0 & 2\tau_0 & \dots & 2\tau_{J_n-2} & 2\tau_{J_n-1}  & 2\tau_{J_n} & \dots & 0  \\
			& & & \dots & & & \dots  & \\
			0 & 0 & 0 & \dots & 2\tau_0 & 2\tau_1  & 2\tau_2 & \dots & 2\tau_{J_n}
		\end{pmatrix}
	\end{align*} 
	We may also provide the gradient of the constraint for the optimization routine. 
	In practice, it requires to supply the gradients of $\zeta'\hat \Lambda_m^{-1}\zeta - B_n^2$ with respect to $\zeta,\tau$, which are, respectively, $2\hat\Lambda_m^{-1}\zeta$ and zero.

	\section{Additional Tables and Figures}
	\label{appendix-table-figures}
	
	This section contains the following: 
	(i) Tables~\ref{table-comparison-1d-ntrain1k} and \ref{table-comparison-1d-ntrain500} are supplements to Table~\ref{table-comparison-1d} in Section~\ref{subsection-simulation-1d}, which show the simulation results using $ntrain\in \{500,1000\}$, respectively.  
	(ii) Table~\ref{table-comparison-10d-appendix-J1or3} reports simulation results for specifications (IIIA) and (IIIB) in Section~\ref{subsection-simulation-10d}, with $J_n$ being fixed (at $J_n=1$ for (IIIA) and $J_n=3$ for (IIIB)), rather than selected by 5-fold cross-validation as in Table \ref{table-comparison-10d}. 
	(iii) Table~\ref{table-boot-ci-length} reports the average lengths of bootstrap confidence intervals in Section~\ref{subsection-simulation-ape}. 
	(iv) Fig.~\ref{fig-v-kde-hist} presents the histogram and estimated density of $V_i$ in the empirical application in Section~\ref{section-application}. 
	
	\afterpage{
		\clearpage
		\begin{table}[h!]
			\caption{Comparison of methods' performance by simulation: $d_w=1$,  $ntrain=1000$, $Nsim = 1000$  
			}  
			\label{table-comparison-1d-ntrain1k} 
			\small 
			\begin{center}
				\centering
				\begin{tabular}{c c c c c c c c c c c c c c c c c c}
					\hline \hline			
					
					Method&&	KNP     & KPB	& SNP	 & Probit &	P2PB   & P3PB	& P4PB   \\  [0.5ex] \hline 
					for $F_0$	&&	GN(87) & probit	& GN(87) & probit &	probit & probit	& probit \\ [0.5ex]
					for $g_0$	&&	RKHS   & RKHS   & linear & linear &	Poly2  & Poly3	& Poly4	 \\ \hline 
					&&		&		&		&		&		& & \\   
					&&	 \multicolumn{7}{c}{ Specification (IIB) } 				\\ 
					RMSE($\hat g$)&& 0.351 & 0.426 & 1.282 & 1.123 & 0.697 & 0.662 & 0.690 &  \\  
					MAD($\hat g$)&& 0.291 & 0.360 & 1.110 & 0.996 & 0.591 & 0.563 & 0.581 &  \\   [0.5ex]  
					RMSE($\hat p$)&& 0.052 & 0.144 & 0.141 & 0.143 & 0.111 & 0.108 & 0.109 &  \\  
					MAD($\hat p$)&& 0.038 & 0.122 & 0.105 & 0.112 & 0.091 & 0.087 & 0.087 &  \\   [1ex]  
					&&	 \multicolumn{7}{c}{ Specification (IIA) } 				\\ 
					RMSE($\hat g$)&& 0.242 & 0.228 & 1.110 & 1.114 & 0.666 & 0.684 & 5.299 &  \\  
					MAD($\hat g$)&& 0.170 & 0.162 & 0.952 & 0.944 & 0.583 & 0.562 & 2.107 &  \\   [0.5ex]  
					RMSE($\hat p$)&& 0.042 & 0.039 & 0.225 & 0.282 & 0.181 & 0.154 & 0.091 &  \\  
					MAD($\hat p$)&& 0.029 & 0.026 & 0.178 & 0.226 & 0.139 & 0.116 & 0.056 &  \\   [1ex]  
					&&	 \multicolumn{7}{c}{ Specification (IB) } 				\\ 
					RMSE($\hat g$)&& 0.547 & 0.665 & 0.149 & 0.172 & 0.415 & 0.435 & 0.485 &  \\  
					MAD($\hat g$)&& 0.411 & 0.472 & 0.129 & 0.149 & 0.293 & 0.314 & 0.359 &  \\   [0.5ex]  
					RMSE($\hat p$)&& 0.051 & 0.154 & 0.034 & 0.074 & 0.071 & 0.072 & 0.074 &  \\  
					MAD($\hat p$)&& 0.037 & 0.133 & 0.026 & 0.061 & 0.057 & 0.058 & 0.059 &  \\   [1ex]  
					&&	 \multicolumn{7}{c}{ Specification (IA) } 				\\ 
					RMSE($\hat g$)&& 0.188 & 0.180 & 0.055 & 0.051 & 0.086 & 0.111 & 0.132 &  \\  
					MAD($\hat g$)&& 0.140 & 0.135 & 0.047 & 0.044 & 0.063 & 0.080 & 0.093 &  \\   [0.5ex]  
					RMSE($\hat p$)&& 0.038 & 0.035 & 0.017 & 0.011 & 0.017 & 0.022 & 0.026 &  \\  
					MAD($\hat p$)&& 0.026 & 0.024 & 0.013 & 0.008 & 0.011 & 0.014 & 0.016 &  \\   [1ex] 
					
					\hline \hline
				\end{tabular}
			\end{center}
			
			\footnotesize
			Notes:  
			Refer to the explanations under Table 1. The only difference in the simulation procedure is that the training sample now consists of $ntrain=1000$ observations, compared to 2000 in Table \ref{table-comparison-1d}. 
		\end{table}
		
		\normalsize
		\clearpage 
	}

	\afterpage{
		\clearpage
		\begin{table}[h!]
			\caption{Comparison of methods' performance by simulation: $d_w=1$,  $ntrain=500$, $Nsim = 1000$  
			} 
			\label{table-comparison-1d-ntrain500} 
			\small 
			\begin{center}
				\centering
				\begin{tabular}{c c c c c c c c c c c c c c c c c c}
					\hline \hline					
					
					Method&&	KNP     & KPB	& SNP	 & Probit &	P2PB   & P3PB	& P4PB   \\  [0.5ex] \hline 
					for $F_0$	&&	GN(87) & probit	& GN(87) & probit &	probit & probit	& probit \\ [0.5ex]
					for $g_0$	&&	RKHS   & RKHS   & linear & linear &	Poly2  & Poly3	& Poly4	 \\ \hline 
					&&		&		&		&		&		& & \\   
					&&	 \multicolumn{7}{c}{ Specification (IIB) } 				\\ 
					RMSE($\hat g$)&& 0.559 & 0.501 & 1.304 & 1.132 & 0.730 & 0.709 & 0.785 &  \\  
					MAD($\hat g$)&& 0.456 & 0.419 & 1.124 & 0.999 & 0.622 & 0.594 & 0.635 &  \\   [0.5ex]  
					RMSE($\hat p$)&& 0.077 & 0.152 & 0.144 & 0.145 & 0.115 & 0.113 & 0.115 &  \\  
					MAD($\hat p$)&& 0.057 & 0.130 & 0.109 & 0.114 & 0.093 & 0.090 & 0.091 &  \\   [1ex]  
					&&	 \multicolumn{7}{c}{ Specification (IIA) } 				\\ 
					RMSE($\hat g$)&& 0.413 & 0.386 & 1.112 & 1.116 & 0.672 & 0.702 & 5.439 &  \\  
					MAD($\hat g$)&& 0.269 & 0.254 & 0.952 & 0.944 & 0.588 & 0.570 & 2.188 &  \\   [0.5ex]  
					RMSE($\hat p$)&& 0.058 & 0.054 & 0.226 & 0.282 & 0.182 & 0.156 & 0.098 &  \\  
					MAD($\hat p$)&& 0.041 & 0.037 & 0.179 & 0.226 & 0.140 & 0.117 & 0.061 &  \\   [1ex]  
					&&	 \multicolumn{7}{c}{ Specification (IB) } 				\\ 
					RMSE($\hat g$)&& 0.679 & 0.722 & 0.206 & 0.225 & 0.463 & 0.504 & 0.596 &  \\  
					MAD($\hat g$)&& 0.519 & 0.501 & 0.178 & 0.194 & 0.332 & 0.367 & 0.441 &  \\   [0.5ex]  
					RMSE($\hat p$)&& 0.067 & 0.159 & 0.044 & 0.077 & 0.076 & 0.078 & 0.081 &  \\  
					MAD($\hat p$)&& 0.049 & 0.137 & 0.035 & 0.063 & 0.061 & 0.062 & 0.065 &  \\   [1ex]  
					&&	 \multicolumn{7}{c}{ Specification (IA) } 				\\ 
					RMSE($\hat g$)&& 0.281 & 0.263 & 0.082 & 0.073 & 0.123 & 0.166 & 0.219 &  \\  
					MAD($\hat g$)&& 0.206 & 0.194 & 0.071 & 0.064 & 0.090 & 0.119 & 0.146 &  \\   [0.5ex]  
					RMSE($\hat p$)&& 0.054 & 0.050 & 0.024 & 0.015 & 0.024 & 0.032 & 0.038 &  \\  
					MAD($\hat p$)&& 0.038 & 0.034 & 0.018 & 0.011 & 0.015 & 0.020 & 0.024 &  \\   [1ex]  
					
					\hline \hline
				\end{tabular}
			\end{center}
			
			\footnotesize
			Notes:  
			Refer to the explanations under Table 1. The only difference in the simulation procedure is that the training sample now consists of $ntrain=500$ observations, compared to 2000 in Table \ref{table-comparison-1d}. 
			
		\end{table}
		
		\normalsize
		\clearpage 
	}

	\afterpage{ 
		\clearpage 
		\begin{landscape}
			
			\begin{table}[h!]
				\caption{Robustness check: Designs (IIIA), (IIIB)}  
				\label{table-comparison-10d-appendix-J1or3}
				
				\small
				\begin{center}
					\centering
					\begin{tabular}{c c c c c c c c c c c c c c c c c c c c c c}
						\hline \hline
						
						Method&&	KNP     & KPB	& SNP	 & Probit &	P2PB      	&& KNP	&	KPB	&	SNP	&	Probit	&	P2PB    	&& KNP	&	KPB	&	SNP	&	Probit	&	P2PB \\  [0.5ex] \hline 
						&&		&		&		&		&		    &&		&		&		&		&		       &&		&		&		&		&		& & \\   
						
						&	&	 \multicolumn{17}{c}{Specification (IIIB),  $J_n = 3$ } 																															\\ [0.5ex]
						&	&	 \multicolumn{5}{c}{ntrain = 2000} 									&	&	 \multicolumn{5}{c}{ntrain = 5000} 									&	&	 \multicolumn{5}{c}{ntrain = 10000} 									\\
						RMSE($\hat g$)	&	&	0.529	&	0.904	&	0.260	&	0.333	&	1.261	&	&	0.471	&	0.899	&	0.202	&	0.209	&	0.765	&	&	0.456	&	0.894	&	0.180	&	0.146	&	0.560	\\
						MAD($\hat g$)	&	&	0.483	&	0.872	&	0.225	&	0.294	&	1.168	&	&	0.442	&	0.877	&	0.179	&	0.185	&	0.706	&	&	0.435	&	0.876	&	0.164	&	0.128	&	0.519	\\ [0.5ex]																							
						RMSE($\hat p$)	&	&	0.055	&	0.123	&	0.043	&	0.067	&	0.105	&	&	0.042	&	0.118	&	0.032	&	0.061	&	0.079	&	&	0.037	&	0.116	&	0.027	&	0.059	&	0.068	\\
						MAD($\hat p$)	&	&	0.041	&	0.101	&	0.033	&	0.054	&	0.084	&	&	0.031	&	0.096	&	0.024	&	0.049	&	0.063	&	&	0.027	&	0.094	&	0.020	&	0.048	&	0.055	\\  [1ex]

						&	&	 \multicolumn{17}{c}{Specification (IIIA),  $J_n = 1$ } 																															\\ [0.5ex]
						&	&	 \multicolumn{5}{c}{ntrain = 2000} 									&	&	 \multicolumn{5}{c}{ntrain = 5000} 									&	&	 \multicolumn{5}{c}{ntrain = 10000} 									\\
						RMSE($\hat g$)	&	&	0.463	&	0.269	&	0.250	&	0.136	&	0.400	&	&	0.404	&	0.217	&	0.158	&	0.085	&	0.233	&	&	0.373	&	0.204	&	0.111	&	0.060	&	0.162	\\
						MAD($\hat g$)	&	&	0.424	&	0.218	&	0.221	&	0.110	&	0.308	&	&	0.375	&	0.172	&	0.140	&	0.069	&	0.181	&	&	0.349	&	0.162	&	0.098	&	0.049	&	0.126	\\ [0.5ex]
						
						RMSE($\hat p$)	&	&	0.048	&	0.048	&	0.033	&	0.031	&	0.085	&	&	0.036	&	0.038	&	0.021	&	0.020	&	0.052	&	&	0.031	&	0.034	&	0.014	&	0.014	&	0.036	\\
						MAD($\hat p$)	&	&	0.033	&	0.031	&	0.021	&	0.020	&	0.053	&	&	0.025	&	0.025	&	0.013	&	0.013	&	0.032	&	&	0.021	&	0.022	&	0.009	&	0.009	&	0.023	\\  [1ex]

						\hline \hline
					\end{tabular}
				\end{center}
				
				\footnotesize
				Notes:   
				See the footnote of Table \ref{table-comparison-10d} for descriptions of specifications (IIIB) and (IIIA), and the footnote of Table \ref{table-comparison-1d} for explanations of each method. 
				The only difference from Table \ref{table-comparison-10d} is that 
				we now fix $J_n=3$ for specification (IIIB) and $J_n=1$ for specification (IIIA), instead of selecting $J_n$ by 5-fold cross-validation. 
			\end{table} 
		\end{landscape}
		\normalsize
		\clearpage 
	}

	\afterpage{ 
		\begin{table}[h!]
			\caption{Average lengths of bootstrap confidence intervals: $Nsim=1000$, $Nboot = 1000$ 
			} 
			\label{table-boot-ci-length} 
			\small 
			\begin{center}
				\centering
				\begin{tabular}{c c c c c c c c c c c c c c c c c c}
					\hline \hline		
					
					Confidencen level	&	&	90\% 	&	&	95\%	&	&	99\%	\\  [0.5ex] \hline  
					&	&	 \multicolumn{5}{c}{$ntrain= 2000$} 							\\ [0.5ex]
					$APE_{W_1}$	&	&	0.141	&	&	0.168	&	&	0.222	\\
					$cAPE_{W_1|W_1<0.5}$	&	&	0.304	&	&	0.364	&	&	0.481	\\
					$cAPE_{W_1|W_1>0.5}$	&	&	0.307	&	&	0.367	&	&	0.486	\\  [1ex]
					
					&	&	 \multicolumn{5}{c}{$ntrain= 5000$} 							\\ [0.5ex]
					$APE_{W_1}$	&	&	0.075	&	&	0.090	&	&	0.120	\\
					$cAPE_{W_1|W_1<0.5}$	&	&	0.166	&	&	0.199	&	&	0.265	\\
					$cAPE_{W_1|W_1>0.5}$	&	&	0.168	&	&	0.201	&	&	0.267	\\  [1ex]
					
					&	&	 \multicolumn{5}{c}{$ntrain= 10000$} 							\\ [0.5ex]
					$APE_{W_1}$	&	&	0.047	&	&	0.056	&	&	0.074	\\
					$cAPE_{W_1|W_1<0.5}$	&	&	0.103	&	&	0.124	&	&	0.164	\\
					$cAPE_{W_1|W_1>0.5}$	&	&	0.105	&	&	0.126	&	&	0.166	\\  [1ex]
					
					\hline \hline
				\end{tabular}
			\end{center}
			
			\footnotesize 
			Notes:  
			The DGP is specification (IVB) described in Section~\ref{subsection-simulation-10d}. 
			For each sample size $n\in\{2000,5000,10000\}$, the average CI lengths are computed from $Nsim=1000$ replications. 
			See the footnote of Table \ref{table-boot-ci-coverage} for implementation details. 
			The true values of $APE_{W_1}$, $cAPE_{W_1|W_1<0.5}$, and $cAPE_{W_1|W_1>0.5}$ are 0.034, 0.158, -0.090, respectively.

		\end{table}
		
		\normalsize 
	}

	\begin{figure}[h] 
		\caption{Histogram and estimated density of $V_i$}  
		\label{fig-v-kde-hist}
		\begin{center} 
			\includegraphics[width=0.8\textwidth]{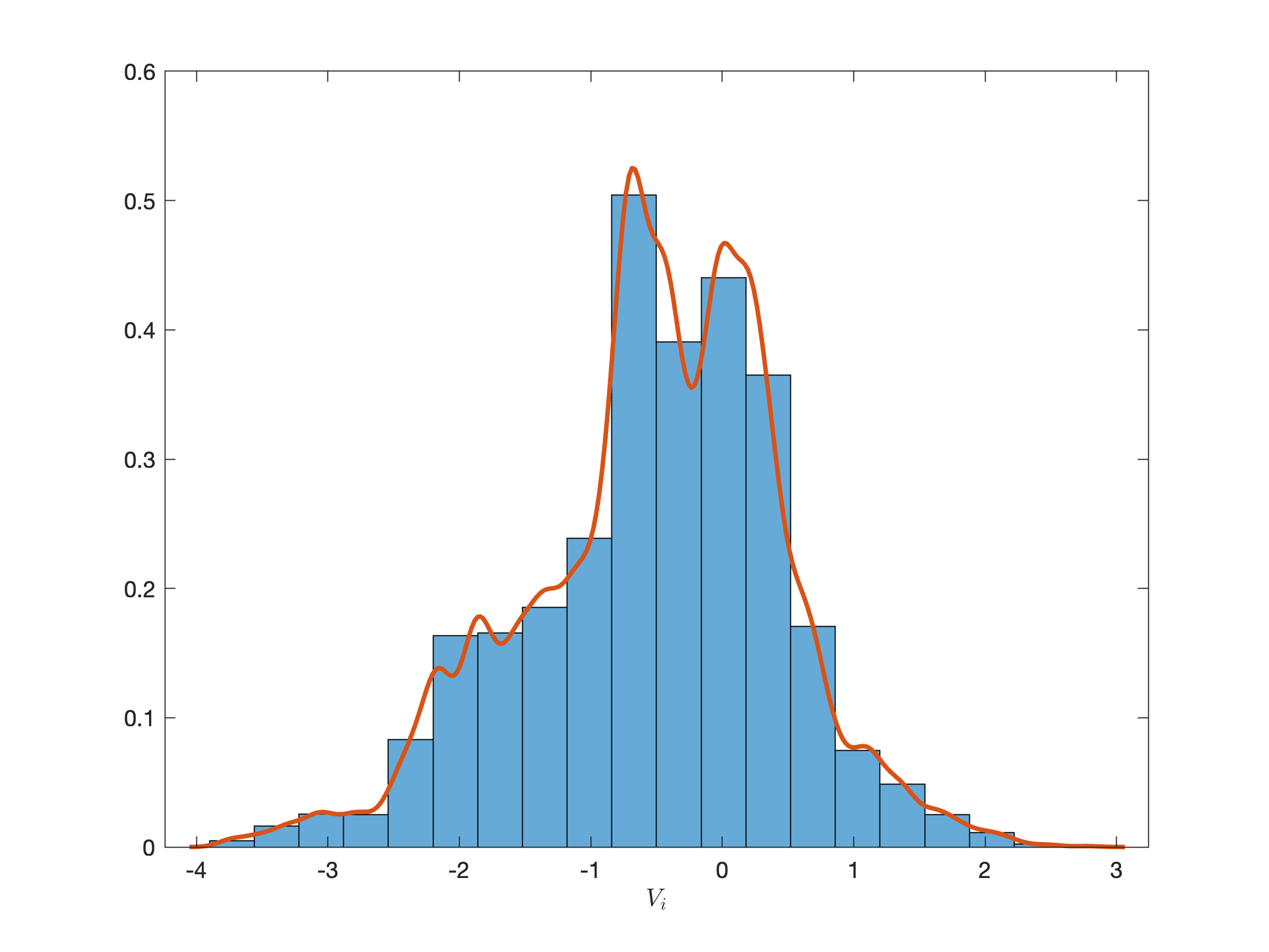}
		\end{center}
		
		\footnotesize
		Notes: The red curve plots the estimated density of $V_i$, using kernel density estimation and the bandwidth is selected by the Silverman's rule of thumb. 
		The blue blocks are the histogram of $V_i$, normalized to the density. 
		This figure suggests that $V_i$'s support is relatively large. 
		Moreover, 
		the empirical range of $V_i$ is $[-3.761, 2.772]$; the $0.1\%, 0.5\%, 1\%, 50\%$, $ 99\%, 99.5\%$ and $99.9\%$ quantiles are $-3.71, -3.32, -3.08, -0.45$, $1.61, 1.87$, and $2.20$, respectively. 
	\end{figure}

\end{document}